\pdfoutput = 1
\documentclass[english]{article}
\usepackage{geometry}
\usepackage[small]{caption}
\usepackage{graphicx}
\usepackage{amsmath}
\usepackage{amsthm}
\usepackage{amssymb}
\usepackage{booktabs}
\usepackage{algorithm}
\usepackage{algorithmic}
\usepackage{multirow,array}
\usepackage{verbatim} 

\newtheorem{definition}{Definition}
\newtheorem{theorem}{Theorem}
\newtheorem{assumption}{Assumption}
\newtheorem{example}{Example}
\newtheorem{remark}{Remark}
\newtheorem{proposition}{Proposition}
\newtheorem{fact}{Fact}
\newtheorem{lemma}{Lemma}

\title{Model and Reinforcement Learning for Markov Games with Risk Preferences }

\begin{document}

\author{{{Wenjie Huang}},\textsuperscript{1,2}
	{{Pham Viet Hai}},\textsuperscript{3}
	{{William B. Haskell}} \textsuperscript{4}
	\\
	\\
	\textsuperscript{1}{{\small{Shenzhen Research Institute of Big Data (SRIBD)}}}\\
	\textsuperscript{2}{{\small{Institute for Data and Decision Analysis, The Chinese University of Hong Kong, Shenzhen}}}\\
	\textsuperscript{3}{{\small{Department of Computer Science, School of Computing, National University of Singapore (NUS)}}}\\
	\textsuperscript{4}{{\small{Supply Chain and Operations Management Area, Krannert School of Management, Purdue University}}} \\
	{\small{wenjiehuang@cuhk.edu.cn,
	dcspvh@nus.edu.sg, 
	whaskell@purdue.edu}}}

\maketitle

\begin{abstract}
	We motivate and propose a new model for non-cooperative Markov
	game which considers the interactions of risk-aware players. This model characterizes the time-consistent dynamic “risk” from \emph{both} stochastic state transitions (inherent to the game) and
	randomized mixed strategies (due to all other players). An appropriate risk-aware equilibrium concept is proposed and the existence of such equilibria is demonstrated in stationary strategies by an application of Kakutani’s fixed point theorem. We further propose a
	simulation-based $Q$-learning type algorithm for risk-aware equilibrium computation. This
	algorithm works with a special form of minimax risk measures which can naturally be written as saddle-point
	stochastic optimization problems, and covers many widely investigated risk measures.
	Finally, the almost sure convergence of this simulation-based algorithm to an equilibrium is demonstrated
	under some mild conditions. Our numerical experiments on a two player queuing game validate the properties of our model and algorithm, and demonstrate their worth and applicability in real life competitive
	decision-making.
\end{abstract}

\section{Introduction}

Markov games (a.k.a stochastic games) generalize Markov decision processes (MDPs) to the multi-player
setting. In the classical case, each player seeks to minimize his
expected costs. In a corresponding equilibrium, no player can decrease
his expected costs by changing his strategy. We often want to compute
equilibria to predict the outcome of the game and understand the
behavior of the players.

In this paper, we directly account for the \textit{risk preferences}
of the players in a Markov game. Informally, risk aversion is at least weakly preferring a gamble with smaller variance when payoffs are the same. Risk-averse players give more attention to low probability but high cost events compared to risk-neutral players. Models for the risk preferences
of a single agent are well established \cite{ADEH,Ruszczynski2006}
for the static problems and \cite{Ruszczynski2010,Shen2013}
for the dynamic case. We extend these ideas to general
sum Markov games and extend the framework of Markov risk measures \cite{Ruszczynski2010,Shen2013} to the multi-agent setting.  Our model specifically addresses the risk from the stochastic state transitions as well as the risk from the randomized strategies of the other players. The traditional multilinear formulation approach \cite{Aghassi2006, Kardes2011}
for computing equilibria in robust games fails in our settings, because our model has an intrinsic bilinear term due to the product of probabilities (the state transitions and mixed strategies) which leads to computational intractability.
Thus, it is necessary to develop an alternative algorithm to compute equilibria. 

\paragraph{Risk Preferences}

Expected utility theory \cite{Engelmann2007,Thomas2016,JVNOM} is
a highly developed framework for modeling risk preferences. Yet, some experiments \cite{Levin2006} show that real human behavior may violate
the independence axiom of expected utility theory. Risk measures (as
developed in \cite{ADEH,Ruszczynski2006}) do not require the independence
axiom and have favorable properties for optimization. 

In the dynamic setting, \cite{Ruszczynski2010,Shen2013} develop the class of Markov (a.k.a. dynamic/nested/iterated) risk measures and establish
their connection to time-consistency. This class of risk measures
is notable for its recursive formulation, which leads to dynamic programming
equations. Practical computational schemes for solving large-scale risk-aware MDPs have been proposed, for instance, $Q$-learning type algorithms \cite{Huang2017,Huang2018a,Jiang2017} and simulation-based fitted value iteration \cite{Yu2018}.  

\paragraph{Risk-sensitive/Robust Games}
Risk-sensitive games have already been considered in \cite{Basu2017,bauerle2017zero,ghosh2016zero,Jose2018,Klompstra2000}.
Risk-sensitivity refers to the specific certainty equivalent $\left(1/\theta\right)\ln\left(\mathbb{E}\left[\exp\left(\theta\,X\right)\right]\right)$
where $\theta>0$ is the risk sensitivity parameter. \cite{Basu2017, ghosh2016zero} focus on zero-sum risk-sensitive games under continuous time setting. 

\textit{Robust} games study ambiguity about costs and/or state
transition probabilities of the game. \cite{Aghassi2006} develop the robust equilibrium concept where each player optimizes
against the worst-case expected cost over the range of model ambiguity.
This paradigm is extended to Markov games in \cite{Kardes2011}, and
the existence of robust Markov perfect equilibria is demonstrated. \cite{Aghassi2006,Kardes2011} formulate robust Markov perfect equilibria as multilinear systems. 

Games with risk preferences are not artificial; rather, they emerge organically
from many real problems. Traffic equilibrium problems with risk-averse
agents are analyzed in \cite{bell2002risk} with non-cooperative game
theory. The preferences of risk-aware adversaries are modeled in Stackelberg
security games in \cite{qian2015robust}, and a computational scheme
for robust defender strategies is presented. 

\paragraph*{Contributions of This Work} We make three main contributions in this paper: 
\begin{enumerate}
	\item We develop a model for risk-aware Markov games where agents
	have time-consistent risk preferences. This model specifically addresses
	\emph{both} sources of risk in a Markov game: (i) the risk from the
	stochastic state transitions and (ii) the risk from the randomized
	strategies of the other players. 
	\item We propose a notion of `risk-aware' Markov perfect equilibria
	for this game. We show that there exist risk-aware equilibria in stationary
	strategies. 
	\item We create a practical simulation-based $Q$-learning type algorithm
	for computing risk-aware Markov perfect equilibria, and we show that
	it converges to an equilibrium almost surely. This algorithm is model-free and so does not require any knowledge of the true model, and thus can search for equilibria purely by observations.
\end{enumerate}

\section{Risk-aware Markov Games}
In this section, we develop risk-aware Markov games. Our game consists of the following ingredients: finite set of players $\mathcal{I}$; finite set of states $\mathcal{S}$; finite set of actions $\mathcal{A}^{i}$ for each player $i\in\mathcal{I}$; strategy profiles $\mathcal{A}:=\times_{i\in\mathcal{I}}\mathcal{A}^{i}$; state-action pairs $\mathcal{K}:=\mathcal{S}\times\mathcal{A}$; transition kernel $P(\cdot|s,\,a)\in\mathcal{P}(\mathcal{S})$ (here $\mathcal{P}(\mathcal{S})$ denotes the distribution over $\mathcal{S}$) for all $(s,\,a)\in\mathcal{K}$, and cost functions $c^{i}: \mathcal{S}\times\mathcal{A}\rightarrow\mathbb{R}$ for all players $i\in\mathcal{I}$.

Each round $t\geq0$ of the game has four steps: (i) first, all
players observe the current state $s_{t}\in\mathcal{S}$; (ii) second,
each player $i\in\mathcal{I}$ chooses $a_{t}^{i}\in\mathcal{A}^{i}$
(all moves are simultaneous and independent, and the corresponding
strategy profile is $a_{t}=\left(a_{t}^{i}\right)_{i\in\mathcal{I}}$);
(iii) third, each player $i\in\mathcal{I}$ realizes cost $c^{i}\left(s_{t},\,a_{t}\right)$;
and (iv) lastly, the state transitions to $s_{t+1}$ according to $P\left(\cdot\,\vert\,s_{t},\,a_{t}\right)$.

We next characterize the players' strategies. 
In this work, we focus on `stationary strategies'. Stationary strategies prescribe
a player the same probabilities over his actions each time the player
visits a certain state, no matter what route he follows to reach that
state. Stationary strategies are more prevalent than normal strategies (which rely on the entire history), due to their mathematical tractability \cite{Vrieze2003,Fink1964,Kardes2011}. Furthermore, the memoryless property of stationary
strategies conforms to real human behavior \cite{Vrieze2003}.

We introduce some additional
notations to characterize stationary strategies $x$. Let $\mathcal{P}(\mathcal{A}^{i})$ denote the distribution over $\mathcal{A}^{i}$. For each
player $i\in\mathcal{I}$ and state $s\in\mathcal{S}$, $x_{s}^{i}\in\mathcal{P}\left(\mathcal{A}^{i}\right)$
is the mixed strategy over actions where $x_{s}^{i}\left(a^{i}\right)$ denotes the probability of choosing $a^{i}$ at state $s$. We define the strategy $x^{i}:=(x_{s}^{i})_{s\in\mathcal{S}}\in\mathcal{X}^{i}:=\times_{s\in\mathcal{S}}\mathcal{P}\left(\mathcal{A}^{i}\right)$
of player $i$, the multi-strategy $x:=\left(x^{i}\right)_{i\in\mathcal{I}}\in\mathcal{X}:=\times_{i\in\mathcal{I}}\mathcal{X}^{i}$
of all players, the complementary strategy $x^{-i}:=(x^{j})_{j\ne i}\in\mathcal{X}^{-i}:=\times_{j\ne i}\mathcal{X}^{j}$,
and the multi-strategy $x_{s}=\left(x_{s}^{i}\right)_{i\in\mathcal{I}}\in\mathcal{X}_{s}:=\times_{i\in\mathcal{I}}\mathcal{P}\left(\mathcal{A}^{i}\right)$
for all players in state $s\in\mathcal{S}$. We sometimes write a
multi-strategy as $x=(u^{i},\,x^{-i})$ to emphasize player $i$'s
strategy $u^{i}$. 

There are two sources of stochasticity in the cost sequence: the stochastic state transitions characterized by the transition kernel $P(\cdot|s,\,a)$, and the randomized mixed strategies of players characterized by $x^{-i}$. In this work, we consider the risk from \emph{both} sources of stochasticity. We begin by constructing the framework for evaluating the risk of sequences of random variables. A \emph{dynamic risk measure} is a sequence of conditional risk
measures each mapping a future stream of random costs into a risk
assessment at the current stage, following the definition of risk
maps from \cite{Shen2013}, and satisfying the stationary and time-consistency
property of \cite[Definition 3]{Ruszczynski2010} and \cite[Definition 1]{Shapiro2016}. We assume each conditional risk measure satisfies three axioms: normalization, convexity, and positive homogeneity , which were originally introduced for static risk measures in the pioneering paper \cite{ADEH}. Here ``convexity'' characterizes the risk-averse behavior of players. From \cite[Definition 1]{Shapiro2016}, a risk-aware optimal policy is \emph{time-consistent} if, the risk of the sub-sequence of random outcome from any future stage is optimized by the resolved policy. In the Appendix, we give explicit definitions of the above three axioms of risk measures, stationary and time-consistency risk preferences, and derivation of recursive evaluation of dynamic risk.

From \cite[Theorem 4]{Ruszczynski2010} and \cite[Proposition 4]{Shapiro2016}, time-consistency allows
for a recursive (iterative) evaluation of risk. The infinite-horizon discounted risk for player $i$ under multi-strategy $x$ will be:
\begin{align}
J_{s_{0}}^{i}(x^{i},\,x^{-i}):=&\rho^{i}(c^{i}(s_{0},\,a_{0})+\gamma\,\rho^{i}(c^{i}(s_{1},\,a_{1}) \nonumber
\\
&+\gamma\,\rho^{i}\left(c^{i}(s_{2},\,a_{2})+\cdot\cdot\cdot\right))),\label{Recursive_risk}
\end{align}
where $\rho^{i}$ is a one-step conditional risk measure that maps random cost from the next stage to current stage, with respect to the joint distribution of randomized mixed strategies and transition kernel. In Eq. (\ref{Recursive_risk}), each $c^{i}(s_{t},\,a_{t}),\,t\geq 1$ is governed by the joint distribution of randomized mixed strategies and transition kernel \[\times_{i\in\mathcal{I}} x_{s_{t}}^{i}(a_{t}^{i})P(s_{t}|s_{t-1},\,a_{t-1}),\] which is defined for fixed $(s_{t-1},\,a_{t-1})$ and for all $s_{t}$ and $a_{t}^{i}$. The initial cost $c^{i}(s_{0},\,a_{0})$ is only governed by the random mixed strategies distribution $\times_{i\in\mathcal{I}} x^{i}_{s_{0}}(a_{0}^{i})$. 

The corresponding best response function for player $i$ is: 
\begin{equation}
\min_{x^{i}\in\mathcal{X}^{i}}\,J_{s_{0}}^{i}(x^{i},\,x^{-i}). \label{MDP}
\end{equation}
Suppose we replace all $\rho^{i}$ with expectation $\mathbb{E}$ in Eq. (\ref{Recursive_risk}) which leads to $\mathbb{E}_{s}^{x}\left[\sum_{t=0}^{\infty}\gamma^{t}c^{i}\left(s_{t},\,a_{t}\right)\right]$, where $\mathbb{E}_{s}^{x}$ denotes expectation with respect to multi-strategies $x$, then Problem (\ref{MDP}) will become risk-neutral. Thus our formulation recovers the risk-neutral game as a special case. 

Denote the ingredients of game $\left\{ J_{s}^{i}(x^{i},\,x^{-i})\right\} _{s\in\mathcal{S},\,i\in\mathcal{I}}$ as $\left\{\mathcal{I},\,\mathcal{S},\,\mathcal{A},\,P,\,c,\,\rho\right\}$. In line with the classical definition of
Markov perfect equilibrium in \cite{Fink1964}, we now define risk-aware
Markov perfect equilibrium.

\begin{definition}
	\label{Definition 2.2} (Risk-aware Markov perfect equilibrium) A
	multi-strategy $x\in\mathcal{X}$ is a risk-aware Markov perfect equilibrium
	for $\left\{\mathcal{I},\,\mathcal{S},\,\mathcal{A},\,P,\,c,\,\rho\right\}$
	if
	\begin{equation}
	J_{s}^{i}(x^{i},\,x^{-i})\leq J_{s}^{i}(u^{i},\,x^{-i}),\,\forall s\in\mathcal{S},\,u^{i}\in\mathcal{X}^{i},\,i\in\mathcal{I}.\label{Markov_equilibrium}
	\end{equation}
\end{definition}

In Definition \ref{Definition 2.2}, each player $i\in\mathcal{I}$
implements a (risk-aware) stationary best response given the stationary
complementary strategy $x^{-i}$. It also states that $x$ is an equilibrium if and only if no player can reduce his discounted risk by unilaterally changing his strategy. 

\paragraph*{Existence of Stationary Equilibria} 

We prove the existence of stationary equilibira in this section. Let $v^{i}$ denote player $i$'s \textit{value function}, which is an estimate of the discounted risk starting from the
next state $S'$. For each player $i$, the \textit{value} of the stationary strategy
$x\in\mathcal{X}$ in state $s\in\mathcal{S}$ is defined to be $v^{i}\left(s\right):=J_{s}^{i}(x)$,
and $v^{i}:=\left(v^{i}\left(s\right)\right)_{s\in\mathcal{S}}$ is
the entire value function for player $i$. The space of value functions for all players is $\mathcal{V}:=\times_{i\in\mathcal{I}}\mathbb{R}^{|\mathcal{S}|}$, equipped with the supremum norm $\|v\|_{\infty}:=\max_{s\in\mathcal{S},\,i\in\mathcal{I}}|v^{i}\left(s\right)|$. Eq. (\ref{Recursive_risk}) states that
each player must evaluate the stage-wise risk of random variables on $\mathcal{A}\times\mathcal{S}$, formulated as 
\begin{equation}
c^{i}\left(s,\,A \right)+\gamma\,v^{i}\left(S'\right), \label{random_variable}
\end{equation}
where $A$ is the random strategy profile chosen from
$\mathcal{A}$ according to $x_s$, and $S'$ is the random next state visited (which first depends on $x$ through
the random choice of strategy profile $a$, and then depends on the transition
kernel $P\left(\cdot\,\vert\,s,\,a\right)$ after $a\in\mathcal{A}$
is realized). 

Recall that in state $s\in\mathcal{S}$, the probability that $a=(a^{i})_{i\in\mathcal{I}}\in\mathcal{A}$ is chosen and then the
system transitions to state $k\in\mathcal{S}$ is $\left(\times_{i\in\mathcal{I}}x_{s}^{i}\left(a^{i}\right)\right)P(k\,\vert\,s,\,a)$. The probability distribution of the strategy profile $a\in\mathcal{A}$ and next state visited $k\in\mathcal{S}$ is given by the matrix
\begin{align}
&P_{s}\left(u_{s}^{i},\,x_{s}^{-i}\right) \nonumber
\\
:=&\left[u_{s}^{i}\left(a^{i}\right)\left(\times_{j\ne i}x_{s}^{j}\left(a^{j}\right)\right)P(k\,\vert\,s,\,a)\right]_{\left(a,\,k\right)\in\mathcal{A}\times\mathcal{S}},\label{Ps}
\end{align}
where we explicitly denote the dependence on the multi-strategy $x_{s}=\left(u_{s}^{i},\,x_{s}^{-i}\right)$
in state $s$.  For simplicity, we often write $P_{s}$ instead of  $P_{s}\left(u_{s}^{i},\,x_{s}^{-i}\right)$ when it is not necessary to indicate the dependence on $(u, x)$. 

Let  $C_{s}^{i}\left(v^{i}\right):=\left(c^{i}\left(s,\,A\right)+\gamma\,v^{i}\left(S'\right)\right)$ be the random cost-to-go for player $i$ at state $s$. Based on the Fenchel-Moreau representation of risk \cite{follmer2002convex,Ruszczynski2006,Guigues2016}, the convex risk of random cost-to-go denoted by $\psi_{s}^{i}(u_{s}^{i},\,x_{s}^{-i},\,v^{i})$ can be computed as the worst-case expected cost-to-go
\begin{align*}
\psi_{s}^{i}(u_{s}^{i},\,x_{s}^{-i},\,v^{i}):=& \rho^{i}\left(c^{i}(s\,,A)+\gamma\,v^{i}\left(S'\right)\right) 
\\
=& \sup_{\mu\in\mathcal{M}_{s}^{i}(P_{s})}\left\{ \langle\mu,\,C_{s}^{i}\left(v^{i}\right)\rangle-b_{s}^{i}(\mu)\right\},
\end{align*} 
where $\left\{ \mathcal{M}_{s}^{i}(P_{s})\right\} _{s\in\mathcal{S},\,i\in\mathcal{I}}\subset\mathcal{P}(\mathcal{A}\times\mathcal{S})$
is the risk envelope of $\rho^{i}$ that depends on the distribution $P_{s}$, and $\left\{b_{s}^{i}\right\} _{s\in\mathcal{S},\,i\in\mathcal{I}}\text{ : }\mathcal{P}\left(\mathcal{A}\times\mathcal{S}\right)\rightarrow\mathbb{R}$ are convex functions satisfying $\inf_{\mu\in\mathcal{P}(\mathcal{A}\times\mathcal{S})}b_{s}^{i}(\mu)=0$ for all $i\in\mathcal{I}$ and $s\in\mathcal{S}$. To connect to risk-neutral games, we can just choose all $\mathcal{M}_{s}^{i}(P_{s})$ to be singletons $\{P_{s}\left(u_{s}^{i},\,x_{s}^{-i}\right)\}$ and $b_{s}^{i}(\mu)=0$ for all $\mu\in \mathcal{M}_{s}^{i}(P_{s})$, $i\in\mathcal{I}$, and $s\in\mathcal{S}$.

We next introduce further assumptions on $\rho^{i}$,
$\left\{ \mathcal{M}_{s}^{i}(P_{s})\right\} _{s\in\mathcal{S},\,i\in\mathcal{I}}$, and $\left\{b_{s}^{i}\right\} _{s\in\mathcal{S},\,i\in\mathcal{I}}$,
that will lead to the existence of stationary equilibria. 
\begin{assumption}
	\label{Assu:Fenchel} (i) All $\rho^{i}$ are law invariant, $\rho^{i}(X)=\rho^{i}(Y)$
	for all $X=_{D}Y$, where $=_{D}$ denotes equality in distribution.
	
	(ii) $\left\{ \mathcal{M}_{s}^{i}(P_{s})\right\} _{s\in\mathcal{S},\,i\in\mathcal{I}}\subset\mathcal{P}\left(\mathcal{A}\times\mathcal{S}\right)$
	is a collection of set-valued mappings where $\mathcal{M}_{s}^{i}(P_{s})$
	are closed and polyhedral convex for all $P_{s}$. Explicitly, there exists $M\geq1$ linear constraints and $[M]:=\{1,2,...,M\}$. Then $\mathcal{M}_{s}^{i}(P_{s})$ is defined as:  
	\begin{align}
			\left\{\mu\in\mathbb{R}^{|\mathcal{A}||\mathcal{S}|}:\begin{array}{cc}
			A_{s,\,m}^{i}\,\mu+f_{m}(P_{s})\geq h_{s,\,m}^{i},  m\in [M],\\
			e^{T}\mu=1,\\
			\mu\geq0,
			\end{array}\right\} \label{Linear formulation}
			\end{align}
	where $A_{s,\,m}^{i}$ are matrices, $f_{m}$ are linear functions in $P_{s}$ and $h_{s,\,m}^{i}$ are constants.
	
	(iii) All $\left\{b_{s}^{i}\right\} _{s\in\mathcal{S},\,i\in\mathcal{I}}$ are convex and Lipschitz continuous.
\end{assumption}

Formulation (\ref{Linear formulation}) explains how
$\mathcal{M}_{s}^{i}(P_{s})$ depends on $P_{s}$. In addition, if $f_{m}$
depends linearly on $P_{s}$, then $f_{m}$ also depends linearly
on $u_{s}^{i}$ and $x_{s}^{-i}$ by definition of $P_{s}$ in Eq. (\ref{Ps}).
In computational terms, this assumption is close
to \cite{Kardes2011} which assumes polyhedral uncertainty sets for
the transition probabilities in its robust Markov game model. This
assumption also corresponds to the one in \cite{Ferris2018} about
representation of agent risk preferences. 
\begin{example}
	Conditional value-at-risk (CVaR) is a widely investigated coherent risk measure that computes the conditional expectation of random losses exceeding a threshold with probability $\alpha$. 
	
	CVaR can be constructed from system (\ref{Linear formulation}) when we choose $M = 1$, $A_{s, m}^{i} = -e$, $f_{m}(P_{s}) = P_{s}/(1-\alpha^{i})$, and $h_{s, m}^{i} = 0$ with $m = 1$. 
\end{example}

The best response function $v_{\ast}$ corresponding to a risk-aware Markov perfect equilibrium, for all $s\in\mathcal{S},\,i\in\mathcal{I}$, satisfies
\begin{align}
v_{\ast}^{i}\left(s\right)=\, &  \min_{u_{s}^{i}\in\mathcal{P}\left(\mathcal{A}^{i}\right)}\,J_{s}^{i}(u^{i},\,x^{-i}) \label{eq:Markov_equilibrium} \nonumber
\\
= \, & \min_{u_{s}^{i}\in\mathcal{P}\left(\mathcal{A}^{i}\right)}\,\psi_{s}^{i}(u_{s}^{i},\,x_{s}^{-i},\,v_{\ast}^{i}),
\\
x_{s}^{i}\in\, & \arg\min_{u^{i}\in\mathcal{X}^{i}}\,J_{s}^{i}(u^{i},\,x^{-i}), \label{eq:Markov_equilibrium-1}
\end{align}
and $v^{i}_{\ast}$ may not be unique. In the mapping $C_{s}^{i}\left(v^{i}\right)$
on $\mathcal{A}\times\mathcal{S}$, the players control
the \textit{distribution} on $\mathcal{P}\left(\mathcal{A}\times\mathcal{S}\right)$
through their mixed strategies. Eqs. (\ref{eq:Markov_equilibrium}) - (\ref{eq:Markov_equilibrium-1})
together simply restate Eq. (\ref{Markov_equilibrium}). However,
Eqs. (\ref{eq:Markov_equilibrium}) - (\ref{eq:Markov_equilibrium-1})
give a computational recipe that can be encoded into an operator on
multi-strategies. We define this operator $\Phi$ on $\mathcal{X}$:
\begin{align}
& \Phi(x) :=\, \Big\{ \tilde{q}\in\mathcal{X}:\,\tilde{q}_{s}^{i}\in\arg\min_{u_{s}^{i}\in\mathcal{P}(\mathcal{A}^{i})}\psi_{s}^{i}(u_{s}^{i},\,x_{s}^{-i},\,v_{\ast}^{i}),\nonumber \\
& v_{\ast}^{i}\left(s\right)=\min_{u
_{s}^{i}\in\mathcal{P}\left(\mathcal{A}^{i}\right)}\psi_{s}^{i}(u_{s}^{i},\,x_{s}^{-i},\,v_{\ast}^{i}),\,\forall s\in\mathcal{S},\,i\in\mathcal{I}\Big\}.\label{Phi}
\end{align}

This operator returns the set of strategies for every player that are best
responses to all other players' strategies. 

The following Theorem 1 briefly describes the existence of stationary strategies with detailed proof in the Appendix.
\begin{theorem}
\label{Theorem 3.11} Suppose Assumption \ref{Assu:Fenchel} holds, then the game $\left\{ \mathcal{I},\,\mathcal{S},\,\mathcal{A},\,P,\,c,\,\rho\right\} $
has an equilibrium in stationary strategies.
\end{theorem}
\begin{proof}
(Proof sketch) Our proof of existence of risk-aware Markov perfect equilibrium
draws from \cite{Fink1964,Kardes2011}. The main idea 
is to show that $\Phi$ is a nonempty, closed, and convex subset of $\mathcal{X}$, and that $\Phi$ is upper semicontinuous. Then, we apply Kakutani's fixed point theorem to show that this correspondence $\Phi$
has a fixed point which coincides with a risk-aware Markov perfect equilibrium. 
\end{proof}
\section{A $Q$-Learning Algorithm}
We propose a simulation-based and asynchronous algorithm for computing
equilibria of the risk-aware game $\left\{\mathcal{I},\,\mathcal{S},\,\mathcal{A},\,P,\,c,\,\rho\right\}$, called
Risk-aware Nash $Q$-learning (RaNashQL). This algorithm does not require a model for the cost functions $\left\{ c^{i}\right\} _{i\in\mathcal{I}}$
or the transition kernel $P$, nor does not it require prior knowledge on $\mathcal{S}$. The algorithm has an outer-inner loop structure, where the risk estimation is performed in the inner loop and the equilibrium estimation is performed in the outer loop. 

In each iteration of RaQL, a collection of $Q$-values
for each player for all strategy profiles, is generated. The one-shot game formed by the collection of $Q$-values is called a \emph{stage game}. We will later formulate stage game explicitly. The outer-inner loop structure follows \cite{Huang2017,Jiang2017,Huang2018a}
where multiple ``stochastic approximation instances'' for both risk
estimation and $Q$-value updates are ``pasted'' together. We show that the Nash equilibria mapping for stage
games is non-expansive, and both the risk estimation error and equilibrium estimation error are bounded by the gap between the estimated $Q$-value and the $Q$-value under the equilibrium. These two conditions allow us to prove the convergence of the algorithm using the theory of stochastic approximation, as shown in \cite{Even-Dar2004}. 

For this section, we assume that our risk measures $\left\{ \rho^{i}\right\}$
have a special form as stochastic saddle-point problems to facilitate computation. Define a probability space $(\Omega,\mathcal{F}, P)$ and the space of essentially bounded random variables $\mathcal{L} = L_{\infty}(\Omega,\mathcal{F}, P)$. 

\begin{assumption}
\label{assu:Q-learning} (Stochastic saddle-point problem) For all
$i\in\mathcal{I}$,
\begin{equation}
\rho^{i}(X)=\min_{y\in\mathcal{Y}^{i}}\max_{z\in\mathcal{Z}^{i}}\mathbb{E}_{P}\left[G^{i}(X,\,y,\,z)\right],\,\forall X\in\mathcal{L}, \label{Saddle}
\end{equation}
where: (i) $\mathcal{Y}^{i}\subset\mathbb{R}^{d_{1}}$ and $\mathcal{Z}^{i}\subset\mathbb{R}^{d_{2}}$
are compact and convex with diameters $D_{\mathcal{Y}}$ and $D_{\mathcal{Z}}$,
respectively. (ii) $G^{i}$ is Lipschitz continuous on $\mathcal{L}\times\mathcal{Y}^{i}\times\mathcal{Z}^{i}$
with constant $K_{G}>1$. (iii) $G$ is convex in $y\in\mathcal{Y}^{i}$ and concave in $z\in\mathcal{Z}^{i}$. (iv) The subgradients of $G$ on $y$ and $z$ are Borel measurable
and uniformly bounded for all $X\in\mathcal{L}$.
\end{assumption}

In \cite[Theorem 3.2]{Huang2018a}, conditions on $G^{i}$
are given to ensure that the corresponding minimax structure (\ref{Saddle})
is a convex risk measure. Some examples of the functions $G^{i}$ are shown in the Appendix such that the corresponding risk-aware Markov perfect equilibria exist. For instance, CVaR can be written as: 
\begin{equation} 
	\textrm{CVaR}{}_{\alpha^{i}}(X):=\min_{\eta\in\mathbb{R}}\left\{ \eta+\frac{1}{1-\alpha^{i}}\mathbb{E}\left[\max\left\{ X-\eta,\,0\right\} \right]\right\}, \label{CVaR}
	\end{equation}
where $\alpha^{i}\in[0,\,1)$ is the risk tolerance for player $i$. 

\paragraph*{Risk-aware Nash $Q$-learning Algorithm}
RaNashQL is updated based on future equilibrium costs (which
depend on all players). In contrast, single-agent $Q$-learning updates
are only based on the player's own costs. Thus, to predict equilibrium
losses, every player must maintain and update a model for all other
player's costs and their risk assessments, which follows the settings in \cite{Hu2003}. 

For all $\left(s,\,a\right)\in\mathcal{S}\times\mathcal{A},\,i\in\mathcal{I}$,
\begin{align}
Q_{\ast}^{i}(s,\,a):= & \min_{y\in\mathcal{Y}^{i}}\max_{z\in\mathcal{Z}^{i}}\mathbb{E}_{P\left(\cdot\,\vert\,s,\,a\right)}\big\{ \nonumber
\\
& G^{i}\left(c^{i}(s\,,A)+\gamma\,v_{\ast}^{i}(S),\,y,\,z\right)\big\}, \label{Q-risk}
\end{align}
denotes the $Q$-values corresponding to a stationary equilibrium and its best response function
$v_{\ast}$. In the case of multiple equilibria, different Nash strategy
profiles may have different equilibrium $Q$-values, so the pair $(v^{i}_{\ast},\,Q^{i}_{\ast})$ may not be unique.

In a multi-agent
$Q$-learning algorithm, the agents play a sequence
of stage games where the payoffs are the current $Q$-values. In each state $s\in\mathcal{S}$, the corresponding
stage game is the collection $(Q^{i}(s))_{i\in\mathcal{I}}$, where
$Q^{i}(s):=\{Q^{i}(s,\,a)\text{ : }a\in\mathcal{A}\}$ is the array
of $Q$-values for player $i$ for all strategy profiles. Let $x_{s}$
be a Nash equilibrium of the stage game $(Q^{i}(s))_{i\in\mathcal{I}}$,
then the corresponding \emph{Nash $Q$-value} for all $i\in\mathcal{I}$ is denoted:
\[
Nash^{i}(Q^{j}(s))_{j\in\mathcal{I}}:=\sum_{a\in\mathcal{A}}\left(\times_{j\in\mathcal{I}} x_{s}^{j}\left(a^{j}\right)\right)Q^{i}\left(s,\,a\right),
\]
which gives each player's corresponding expected cost in state $s\in \mathcal{S}$ (with
respect to the $Q$-values) under $x_{s}$.

RaNashQL builds upon the algorithm in \cite{Hu2003} for the risk-aware
case. Figure 1 illustrates how players interact with others and update their equilibrium estimation through RaQL. 
\begin{figure}[h]
\begin{center}
	\includegraphics[scale=0.2]{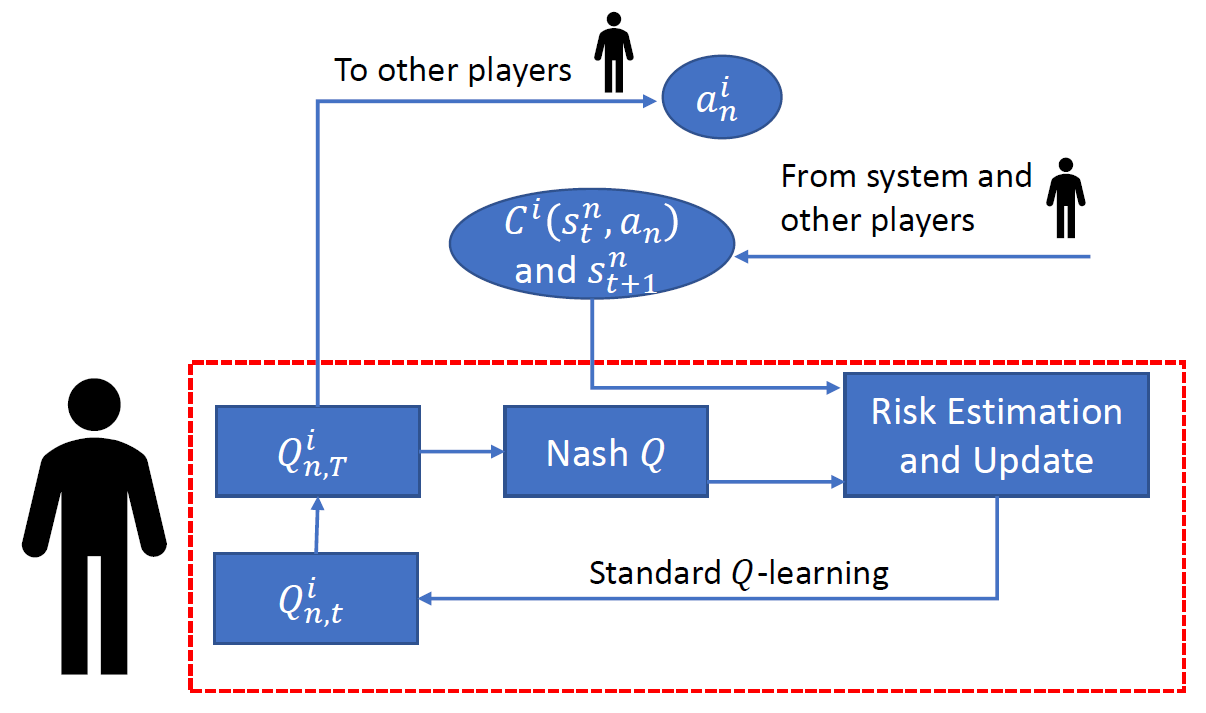}
\end{center}
\caption{Illustration of RaQL}
\end{figure}
Each player chooses an action based on a Nash equilibrium of their current $Q$-values, observed cost, other players' actions, and then the new state in each iteration. The $Q$-values follow a stochastic approximation-type update as in standard $Q$-learning. 

\begin{algorithm}
\caption{Risk-aware Nash $Q$-learning}
		
		\textbf{(Step 0) Initialize}: Let $n=1$, and $t=1$, get the initial
		state $s_{1}$. Let the learning agent be indexed by $i$. For all
		$s\in\mathcal{S}$ and $a^{i}\in\mathcal{A}^{i},\,i\in\mathcal{I}$,
		let $Q_{n,t}^{i}(s,\,a)=0$. 
		
		\textbf{For} $n=1,\,...,\,N$ \textbf{do}
		
		\textbf{(Step 1)} Choose $a_{n}^{i}$ based on the exploration policy
		$\pi$. Observe the actions and costs for all players, then
		observe a new state; 
		
		\textbf{$\qquad$For} $t=1,\,...,\,T$ \textbf{do }
		
		\textbf{$\qquad$(Step 2)} Compute the Nash $Q$-value; Compute the
		risk-aware cost-to-go for all players; 
		
		\textbf{$\qquad$(Step 3)} Update each $Q_{n,t}^{i},\,i\in\mathcal{I}$
		using stochastic approximation; 
		
		\textbf{$\qquad$(Step 4)} Stochastic approximation of risk measure
		by SASP; 
		
		\textbf{$\qquad$end for}
		
		\textbf{end for}
		
		\textbf{Return} Approximated $Q$-value $Q_{N,T}^{i},\,i\in\mathcal{I}$.
\end{algorithm}
The steps of RaNashQL are summarized in Algorithm 1, which contains $N$ and $T$ number of iterations for outer and inner loops, respectively. In Step 4, we use the stochastic approximation for saddle-point problems (SASP)
algorithm, \cite[Algorithm 2.1]{Nemirovski2005}. Classical stochastic
approximation may result in extremely slow convergence for degenerate
objectives (i.e. when the objective has a singular Hessian). However,
the SASP algorithm with a properly chosen parameter preserves a ``reasonable'' (close to $O(n^{-1/2})$) convergence
rate, even when the objective is non-smooth and/or degenerate. Thus,
SASP is a robust choice for solving problem (\ref{Saddle}). The extended formulations from Steps (0)-(4) in Algorithm 1 are given in the Appendix.

\paragraph*{Almost Sure Convergence} 
Let $\{Q_{n,T}\}_{i\in\mathcal{I}}$ be the $Q$-value estimations at iteration $n$ and $T$ (the end of each inner loop after the risk estimation has been done) from Algorithm 1. We would like to demonstrate the almost sure convergence of $Q_{n,T}^{i}$ to the risk-aware equilibrium
$Q$-values $Q_{\ast}^{i}$ for all players. \cite{Hu2003} introduce two conditions on the Nash equilibria of all the stage games that lead to almost sure convergence, a \emph{global optimal} point when every player receives his lowest cost at this point, and a \emph{saddle} point when each agent would receive a lower cost when at least one of the other players deviates. We found a special type of Nash equilibria that we call an $\mathcal{I}^{\prime}$-mixed point, which builds on \cite{Hu2003}, and plays a major role in our convergence analysis. 

\begin{definition} Let $(C^{i})_{i\in\mathcal{I}}$ denote the expected cost of all players as a function of the multi-strategy $x\in\mathcal{X}$. 
\label{Definition 1.9} A multi-strategy $x\in\mathcal{X}$
is a $\mathcal{I}^{\prime}$-mixed point of $\left(C^{i}\right)_{i\in\mathcal{I}}$
if: (i) it is a Nash equilibrium and (ii) there exists an index of players
$\mathcal{I}^{\prime}\subseteq\mathcal{I}$ such that: $C^{i}\left(x\right)\leq C^{i}\left(x'\right),\,\forall x'\in\mathcal{X},\,i\in\mathcal{I}^{\prime}$, and	$C^{i}\left(x^{i},\,x^{-i}\right)\leq C^{i}\left(x^{i},\,u^{-i}\right),\,\forall u^{-i}\in\mathcal{X}^{-i},\,i\in\mathcal{I}\backslash\mathcal{I}'$.
\end{definition}
Our
definition of `$\mathcal{I}^{\prime}$-mixed point' combines both notions of global optimal point and saddle point. From Definition \ref{Definition 1.9}, a subset
of players $\mathcal{I}^{\prime}\subseteq\mathcal{I}$ minimizes their
expected costs at $x$. The rest of the players $\mathcal{I}\backslash\mathcal{I}^{\prime}$
each would receive a lower expected cost when at least one of the
other players deviates. An example of an $\mathcal{I}^{\prime}$-mixed point in a one shot game follows.

\begin{example}
Player 1 has choices \emph{Up} and \emph{Down}, and Player 2 has
choices \emph{Left} and \emph{Right. }Player 1's loss is the first
entry in each cell, and Player 2's are the second. The first game
has a unique Nash equilibrium (Up, Left), which is a global optimal
point. The second game also has a unique Nash equilibrium (Down, Right),
which is a saddle-point. The third game has two Nash equilibrium:
a global optimum (Up, Left), and a mixed point (Down, Right).
In equilibrium (Down, Right), Player 1 receives a lower cost
if Player 2 deviates, while Player 2 receives a higher cost if Player
1 deviates. 

\begin{table}
	\begin{centering}
		\begin{tabular}{c|c|c|}
			\textbf{Game 1} & \emph{Left } & \emph{Right}\tabularnewline
			\hline 
			\emph{Up} & $0,\,1$ & $10,\,7$\tabularnewline
			\hline 
			\emph{Down} & $7,\,10$ & $11,\,8$\tabularnewline
			\hline 
		\end{tabular}\textbf{ }%
		\begin{tabular}{c|c|c|}
			\textbf{Game 2} & \emph{Left } & \emph{Right}\tabularnewline
			\hline 
			\emph{Up} & $5,\,5$ & $10,\,4$\tabularnewline
			\hline 
			\emph{Down} & $4,\,10$ & $8,\,8$\tabularnewline
			\hline 
		\end{tabular}\textbf{ }%
		\begin{tabular}{c|c|c|}
			\textbf{Game 3} & \emph{Left } & \emph{Right}\tabularnewline
			\hline 
			\emph{Up} & $0,\,1$ & $10,\,9$\tabularnewline
			\hline 
			\emph{Down} & $7,\,10$ & $8,\,8$\tabularnewline
			\hline 
		\end{tabular}
		\par\end{centering}
	\caption{Examples of $\mathcal{I}^{\prime}$-mixed point}
	
\end{table}
\end{example}

We now introduce the following additional assumptions for our analysis of RaNashQL.
\begin{assumption}
\label{Assumption 5.11} One of the following holds for all stage
games $(Q_{n,T}^{i}(s))_{i\in\mathcal{I}}$ for all $n$ and $s\in\mathcal{S}$
in Algorithm 1.

(i) Every $(Q_{n,T}^{i}(s))_{i\in\mathcal{I}}$
for all $n$ and $s\in\mathcal{S}$ has a global optimal point.

(ii) Every $(Q_{n,T}^{i}(s))_{i\in\mathcal{I}}$
for all $n$ and $s\in\mathcal{S}$ has a saddle point.

(iii) For any two stage games $Q,\,\tilde{Q}\in(Q_{n,T}^{i}(s))_{i\in\mathcal{I}}$
for all $n$ and $s\in\mathcal{S}$, we suppose $Q_{1}$ has a $\mathcal{I}_{1}$-mixed
point $x$ and $Q_{2}$ has a $\mathcal{I}_{2}$-mixed point $\tilde{x}$.
Then: For $i\in\mathcal{I}_{1}\cup(\mathcal{I}\backslash\mathcal{I}_{2})$,
then $Q^{i}\left(x\right)\geq\tilde{Q}^{i}\left(\tilde{x}\right)$; For $i\in\mathcal{I}_{2}\cup(\mathcal{I}\backslash\mathcal{I}_{1})$,
then $Q^{i}\left(x\right)\leq\tilde{Q}^{i}\left(\tilde{x}\right)$. 
\end{assumption}

Compared with \cite[Assumption 3]{Hu2003}, Assumption \ref{Assumption 5.11}(iii) enables wider application
of RaNashQL. In particular, even the indices $\mathcal{I}_{1}$ and $\mathcal{I}_{2}$
of all the stage games may differ across iterations. Next we list further standard assumptions on exploration in RaNashQL and its asynchronous updates. 
\begin{assumption}{\label{Assumption 4}}
(i) The exploration policy $\pi$ is $\varepsilon-$greedy, meaning with probability
$\varepsilon\in\left(0,\,1\right)$, action $a^{i}$ is chosen
uniformly from $\mathcal{A}^{i}$, and with probability $1-\varepsilon$, action $a^{i}$ is drawn from $\mathcal{A}^{i}$ according to $x_{s}^{i}$ which is the equilibrium of the stage game $\{Q^{i}(s)\}_{i\in\mathcal{I}}$; (ii) a single state-action pair is updated when it is observed in
each iteration.
\end{assumption}
By the Extended Borel-Cantelli Lemma \cite{Breiman1992}, the algorithm satisfying Assumption \ref{Assumption 4}(i)
will visit every state-action pair infinitely often with probability
one. 
\begin{theorem}
\label{Theorem 5.12} Suppose Assumptions \ref{Assumption 5.11} and \ref{Assumption 4} hold. For any
$T\geq1$, Algorithm
1 generates sequences $\left\{ Q_{n,T}^{i}\right\} _{n\geq1}$ such
that $Q_{n,\,T}^{i}\rightarrow Q_{\ast}^{i}$ almost surely as $n\rightarrow\infty$
for all $i\in\mathcal{I}$.
\end{theorem}
\begin{proof} (Proof sketch)
(i) Show that all $\mathcal{I}^{\prime}$-mixed points of a stage game have equal value, and the property also holds for global optimal points and saddle points. Consequently, from \cite{Hu2003}, the mapping from $Q$-values to Nash equilibrium (of the stage games) is non-expansive. 

(ii) Show that the Hausdorff distance
between the subdifferentials of the estimated risk on $\mathcal{Y}^{i}$ and $\mathcal{Z}^{i}$  (corresponding to Eq. (\ref{Saddle})), is bounded by a function of $\|Q_{n-1,T}^{i}-Q_{\ast}^{i}\|_{2}$. 

(iii) Show that the duality gaps of all the
saddle point estimation problems are bounded by a function of $\|Q_{n-1,T}^{i}-Q_{\ast}^{i}\|_{2}$. 

(iv) If the conditions in (i)-(iii) hold, then $Q_{n,\,T}^{i}$ from RaNashQL are a well-behaved stochastic approximation sequence \cite[Definition 7]{Even-Dar2004} that converges to $Q_{\ast}^{i}$ with probability one. 
\end{proof}
The full proof Theorem \ref{Theorem 5.12} is presented in the Appendix.

\cite[Theorem 4.7]{Huang2018a} shows that the single-agent
version of RaNashQL has complexity
\begin{equation}
\Omega\left(\left(S\,A\ln(S\,A/\delta\epsilon)/\epsilon^{2}\right)^{1/\beta}+(\ln(\sqrt{S\,A}/\epsilon))^{1/(1-\beta)}\right),\label{Benchmark}
\end{equation}
with probability $1-\delta$, where $S$ and $A$ denote the cardinality of state and actions spaces and $\beta\in(0,\,1]$ is the learning rate. In the multi-agent case, our conjecture is to replace
$A$ with $|\mathcal{A}|$ in the term (\ref{Benchmark}) to get a
rough estimate of the time complexity of RaNashQL. However, the explicit complexity bound is difficult to derive and remains for future research. In RaNashQL, there are multiple $Q$-values being updated in each iteration for each state, and their relationships are complex (they are linked by the solutions of a stage game, since each stage game may yield multiple Nash equilibria). 

In the Appendix, we also discuss (i) methods for computing Nash equilibria of stage games involving two or more players;  (ii) a rule for choosing a unique Nash equilibrium of stage games from multiple choices; (iii) the storage space requirement of RaNashQL. 

\section{A Queuing Control Application}
We apply our techniques to the single server exponential queuing system
from \cite{Kardes2011}.  In this packet switched network, it is service provider's (denoted as ``SP'' latter in the tables) benefit to increase the amount
of packets processed in the system. However, such an increase may
result in an increase in packets\textquoteright{} waiting times in
the buffer (called latency), and routers (denoted as ``R'' latter in the tables) are used to reduce packets\textquoteright{}
waiting times. Thus, the game arises
because the service provider and router choose their service rates to achieve competing objectives. 

The state space $\mathcal{S}$ represents the maximum number ($30$ in these experiments) of packets allowed in the system. We assume that the time until the admission of a new packet and the next service completion are both exponentially distributed. Therefore, the number of
packets in the system can be modeled as a birth and death process with fixed state transition probabilities. In the Appendix, we provide the explicit formulation of cost functions, state transition probabilities, as well as other parameter settings. We suppose that each player has the same two available actions (service rates) in every state. CVaR is the risk measure for both players in all the experiments. The player’s risk preferences are obtained by setting $\alpha^{i}$ for $i = 1, 2$, and we allow $\alpha^{1} \neq \alpha^{2}$. 

\paragraph*{Experiment I (RaNashQL vs. Nash $Q$-learning)} 

We compare RaNashQL with Nash $Q$-learning in \cite{Hu2003} in terms of their convergence rates. Given any precision $\epsilon>0$, we record the iteration count $n$
until the convergence criterion $\|Q_{n,\,T}^{i}-Q_{\ast}^{i}\|_{2}\leq\epsilon$
is satisfied. Figure 2 (top) reveals that RaNashQL
is more computationally expensive than Nash $Q$-learning. Table 2 shows the discounted cost under equilibrium by simulation (1000 samples). The first table reveals that incorporating risk will help the service provider reduce its mean cost, while increase the mean cost of the router. The second table shows that incorporating risk will help to reduce the overall cost to the entire system with only a slightly higher variance. 

The first part of Table 3 shows that the mean cost of service provider ($-44.31$)
is lower than that under the risk-neutral Markov perfect equilibrium
($-22.22$), and the mean cost of router ($59.64$) is lower than
that under the risk-aware Markov perfect equilibrium ($37.48$). This
result shows that incorporating risk preference can help decision
makers reach a new equilibrium that further reduces his mean cost compared to cases where both players are either risk-neutral or risk-aware. Similar phenomena can also be shown in the second part of Table 3. In the final part of Table 3, we construct a new two-player one-shot game where the risk preferences (risk-neutral and risk-aware) are the actions and the expected value from simulation will be outcome of the game. We find that a equilibrium is attained for this game when the router is risk-neutral and the service provider is risk-aware. This one-shot game demonstrates that the router should be risk-neutral when service provider is risk-aware, in order to reduce his expected cost. 

In the Appendix, we further explain the reason for the increase in variance in risk-aware games in Table 2 which is counter-intuitive. 

\begin{table}[]
\begin{center}
{\small{\begin{tabular}{cccccc}
		\hline 
		Player & Method & Mean  & Variance & $5$\%-CVaR & $10$\%-CVaR\tabularnewline
		\hline 
		{SP} & Neutral & $-22.22$ & $1.4736e-06$ & $-22.22$ & $-22.22$\tabularnewline
		\cline{2-6} 
		& CVaR & $-77.78$ & $407.84$ & $-69.34$ & $-68.26$\tabularnewline
		\hline 
		{R} & Neutral & $37.48$ & $7.32$ & $37.94$ & $38.18$\tabularnewline
		\cline{2-6} 
		& CVaR & $83.68$ & $491.20$ & $86.03$ & $87.54$\tabularnewline
		\hline 
\end{tabular}}}
\par\end{center}

\vspace{0.2cm}

\begin{center}
{\small{\begin{tabular}{cccc}
		\hline 
		Method & Mean & $5$\%-CVaR & 10\%-CVaR\tabularnewline
		\hline 
		\hline 
		Neutral & 15.26 & 15.72 & 15.96\tabularnewline
		\hline 
		CVaR & 5.9 & 16.69 & 19.28\tabularnewline
		\hline 
\end{tabular}}}
\par\end{center}
\caption{Simulation (Constructing CVaR with $\alpha^{1}=\alpha^{2}=0.1$)}
\end{table}
\begin{table}[]
\begin{center}
{\small{\begin{tabular}{cccccc}
		\hline 
		Player & Method & Mean  & Variance & $5$\%-CVaR & $10$\%-CVaR\tabularnewline
		\hline 
		\hline 
		{SP} & CVaR & $-44.31$ & $266.06$ & $-43.38$ & $-42.70$\tabularnewline
		\hline 
		{R} & Neutral & $59.64$ & $316.71$ & $61.18$ & $62.77$\tabularnewline
		\hline 
\end{tabular}}}
\end{center}

\vspace{0.2cm}

\begin{center}
{\small{\begin{tabular}{cccccc}
		\hline 
		Player & Method & Mean  & Variance & $5$\%-CVaR & $10$\%-CVaR\tabularnewline
		\hline 
		\hline 
		{SP} & CVaR & $-54.76$ & $26.05$ & $-54.71$ & $-54.67$\tabularnewline
		\hline 
		{R} & Neutral & $70.56$ & $31.03$ & $71.56$ & $71.81$\tabularnewline
		\hline 
\end{tabular}}}

\vspace{0.3cm}

{\small{\begin{tabular}{cc|c|c|}
		& \multicolumn{1}{c}{} & \multicolumn{2}{c}{Router}\\
		& \multicolumn{1}{c}{} & \multicolumn{1}{c}{Risk-neutral}  & \multicolumn{1}{c}{Risk-aware} \\\cline{3-4}
		\multirow{2}*{Service Provider}  & Risk-neutral & $(-22.22, 37.48)$ & $(-54.76, 70.56)$ \\\cline{3-4}
		& Risk-aware & $(-44.44, 59.64)$ & $(-77.78, 83.68)$ \\\cline{3-4}
\end{tabular}}}
\end{center}
\caption{Simulation ( Constructing CVaR with $\alpha^{1}=0.95$, $\alpha^{2}=0.1$ for the first table, and $\alpha^{1}=0.1$, $\alpha^{2}=0.95$ for the second)}	
\end{table}

\paragraph*{Experiment II (RaNashQL vs. Multilinear System)} In this
experiment, we consider a special case where the risk only comes from state transitions
(this setting is basically a risk-aware interpretation of \cite{Kardes2011}). In this case, we can compute the risk-aware Markov equilibrium ``exactly'' using a multilinear
system and interior point algorithm as detailed in the Appendix. We evaluate performance in terms of the relative error
\[
\frac{\sqrt{\sum_{s\in\mathcal{S}}\left(Nash^{i}(Q_{n,\,T}^{j}(s))_{j\in\mathcal{I}}-v_{\ast}^{i}(s)\right)^{2}}}{\sqrt{\sum_{s\in\mathcal{S}}v_{\ast}^{i}(s)^{2}}},\,n\leq N,
\]
where $v_{\ast}^{i}$ is the value function corresponding to the equilibrium solved by multilinear system. The Appendix confirms that the service provider's strategy produced by RaNashQL converges almost surely to the one produced by multilinear system. From the Appendix, interior point algorithm finds a local optimum with 10471.975 seconds, and RaNashQL has relative error lower than 25\% with 5122.657 seconds. Thus, our approach possesses superior computational performance compared to an interior point
algorithm for solving multilinear systems.
\paragraph*{Experiment III (Computational Complexity Conjecture)} In this experiment, we explore the conjecture on the computational complexity of RaNashQL.  Given a fixed $\epsilon$, we could compute the complexity conjecture through formulation (\ref{Benchmark}). Figure 2 (bottom) shows that the relative errors of service provider
and router under computed complexity conjecture are bounded by $\epsilon$. Thus we derive a potential heuristic for the computational complexity of
solving a general sum game given the size of the game. In other words, each practitioner can estimate the upper bound of total complexity in computing the $\epsilon-$ equilibrium through this conjecture. 
\begin{figure}[h]
\begin{center}
\includegraphics[scale=0.2]{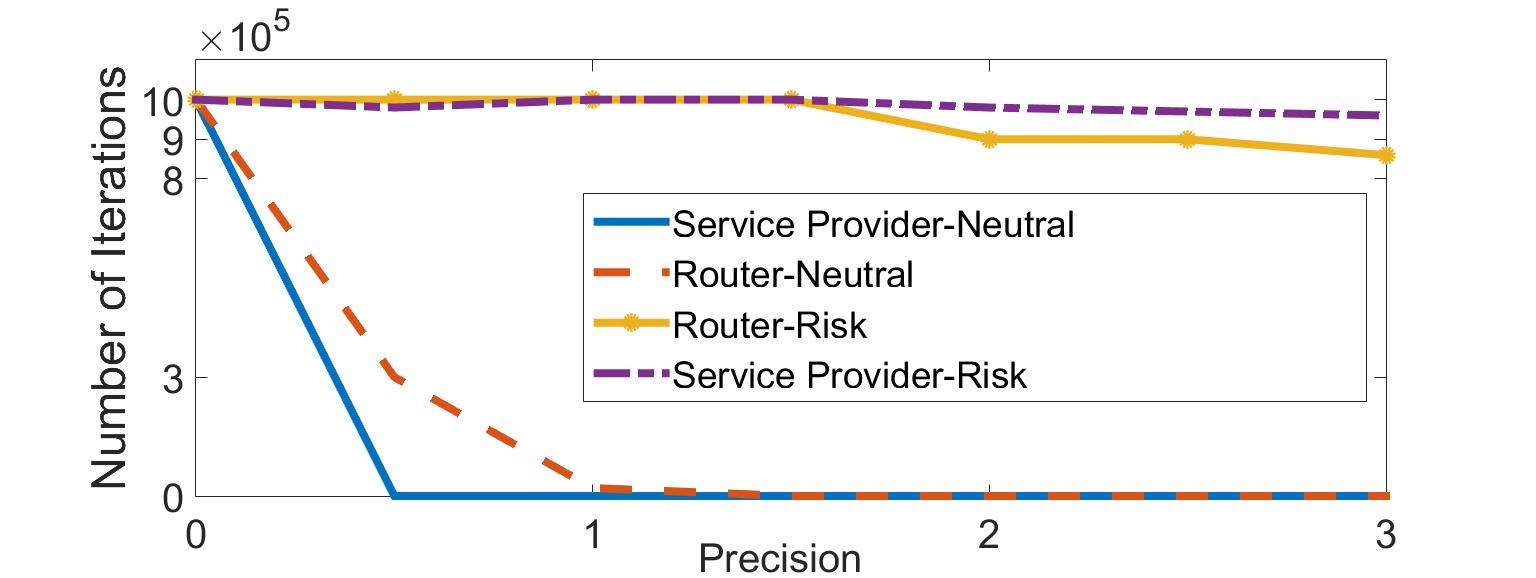}

\vspace{0.3cm}

\includegraphics[scale=0.2]{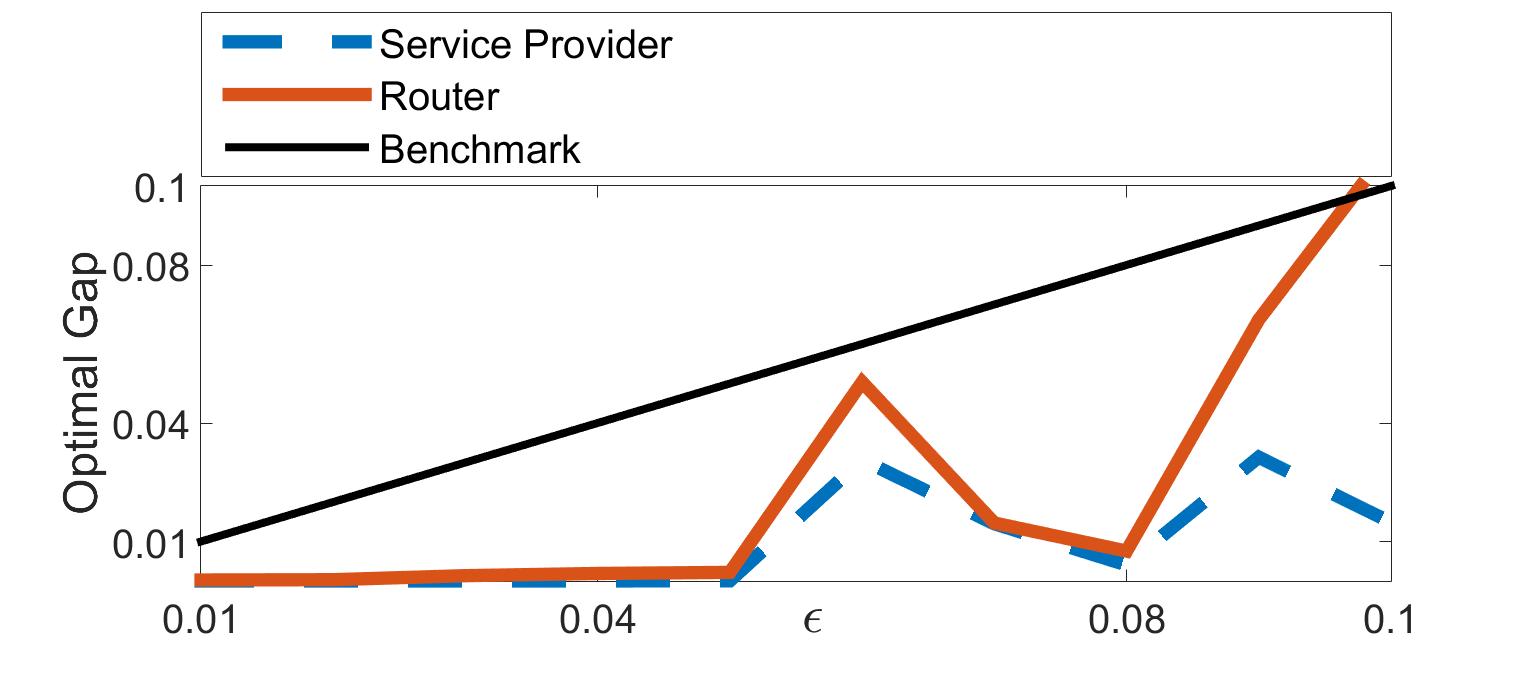}
\end{center}
\caption{Computational Complexity}
\end{figure}
\section{Conclusion}
In this paper, we propose a model and simulation-based algorithm for non-cooperative Markov games with time-consistent risk-aware players. This work has made the following contributions: (i) The model characterizes the “risk” from both the stochastic state transitions and the randomized strategies of the other players. (ii) We define risk-aware Markov perfect equilibrium and prove its existence in stationary strategies. (iii) We show that our algorithm converges to risk-aware Markov perfect equilibrium almost surely. (iv) From a queuing control numerical example, we find that risk-aware Markov games will reach new equilibria other than risk-neutral ones (this is the equilibrium shifting phenomenon). Moreover, the variance is increased for risk-aware Markov games, which is contrary to the variance reduction property of risk-aware optimization for single agents. The sum of expected cost over all players is reduced in risk-aware Markov game, compared to risk-neutral ones. In future research, we seek to improve
the scalability of our framework for large-scale Markov
games. 
\section*{Acknowledgements}
This work is supported by SRIBD International Postdoctoral Fellowship and the NUS Young Investigator Award ``Practical Considerations for Large-Scale
Competitive Decision Making''.

\bibliographystyle{plain}
\bibliography{ArXiv}

\newpage

\appendix
\section{Appendix}

\subsection{Dynamic Risk Measures}

In this section, we describe the risk measures in our
risk-aware Markov games. In our model, each player $i$ faces a sequence
of costs $X_{t}=c^{i}(s_{t},\,a_{t})$ for all $t\geq0$. There are
two sources of stochasticity in this cost sequence: (i) stochastic
state transitions characterized by the transition kernel $P(\cdot\,\vert\,s,\,a)$;
and (ii) the randomized mixed strategies of other players characterized
by $x^{-i}$. The key question is: how should player $i$ account
for both sources of stochasticity and evaluate the risk of the tail
subsequence $X_{t},\,X_{t+1},\ldots$ from the perspective of time
$t$?

We begin by formalizing some details about the risk of finite cost sequences
$X_{t,\,T}:=(X_{t},\,X_{t+1},\ldots,\,X_{T})$ before we consider
the risk of the infinite cost sequence $X_{0},\,X_{1},\ldots$ actually
faced by the players. For a reference distribution $P$ on $\left(\Omega,\,\mathcal{F}\right)$,
and we define $\text{\ensuremath{L_{t}:=\mathcal{L}_{\infty}(\Omega,\mathcal{\,F}_{t},\,P)}}$
and $L_{t,\,T}:=L_{t}\times L_{t+1}\times\cdots\times L_{T}$ for
all $0\leq t\leq T<\infty$.
\begin{definition}
	\label{Definition 2.1-2}(i) A mapping $\rho_{t,\,T}:\,L_{t,\,T}\rightarrow L_{t}$,
	is called a \emph{conditional risk measure} if: $\rho_{t,\,T}(Z_{t,\,T})\leq\rho_{t,\,T}(X_{t,\,T})$
	for all $Z_{t,\,T},\,X_{t,\,T}\in L_{t,\,T}$ such that $Z_{t,\,T}\leq X_{t,\,T}$. 
	
	(ii) A \emph{dynamic risk measure} is a sequence of conditional risk
	measures $\{\rho_{t,\,T}\}_{t=0}^{T}$.
\end{definition}

Given a dynamic risk measure $\{\rho_{t,\,T}\}_{t=0}^{T}$, we may
define a larger family of risk measures $\rho_{t,\,\tau}$ for $0\leq t\leq\tau\leq T$
via the convention $\rho_{t,\,\tau}(X_{t},\ldots,\,X_{\tau})=\rho_{t,\,\tau}(X_{t},\ldots,\,X_{\tau},\,0,\ldots,\,0)$.

We now make our key assumptions about player risk preferences.
\begin{assumption}
	\label{Assumption 2.5} The dynamic risk measure $\{\rho_{t,\,T}\}_{t=0}^{T}$
	satisfies the following conditions:
	
	(i) (Normalization) $\rho_{t,\,T}(0,\,0,\,...,\,0)=0.$
	
	(ii) (Conditional translation invariance) For any $X_{t,\,T}\in L_{t,\,T}$,
	\[
	\rho_{t,\,T}(X_{t},\,X_{t+1},\,...,\,X_{T})=X_{t}+\rho_{t,\,T}(0,\,X_{t+1},\,...,\,X_{T}).
	\]
	
	(iii) (Convexity) For any $X_{t,\,T},\,Y_{t,\,T}\in L_{t,\,T}$ and
	$0\leq\lambda\leq1$, $\rho_{t,\,T}(\lambda\,X_{t,\,T}+(1-\lambda)Y_{t,\,T})\leq\lambda\,\rho_{t,\,T}(X_{t,\,T})+(1-\lambda)\rho_{t,\,T}(Y_{t,\,T})$.
	
	(iv) (Positive homogeneity) For any $X_{t,\,T}\in L_{t,\,T}$ and
	$\alpha\geq0$, $\rho_{t,\,T}(\alpha\,X_{t,\,T})=\alpha\,\rho_{t,\,T}(X_{t,\,T}).$
	
	(v) (Time-consistency) For any $X_{t,\,T},\,Y_{t,\,T}\in L_{t,\,T}$
	and $0\leq\tau\leq\theta\leq T$, the conditions $X_{k}=Y_{k}$ for
	$k=\tau,\,...,\,\theta-1$ and $\rho_{\theta,\,T}(X_{\theta},\,....,\,X_{T})\leq\rho_{\theta,\,T}(Y_{\theta},\,...,\,Y_{T})$
	imply $\rho_{\tau,\,T}(X_{\tau},\,....,\,X_{T})\leq\rho_{\tau,\,T}(Y_{\tau},\,...,\,Y_{T})$.
\end{assumption}

Many of these properties (monotonicity, convexity, positive homogeneity,
and translation invariance) were originally introduced for static
risk measures in the pioneering paper \cite{ADEH}. They have since
been heavily justified in other works including \cite{ruszczynski2006optimization,Bertsimas01112009,NPS09}. 

The next theorem gives a recursive formulation for dynamic risk measures
satisfying Assumption \ref{Assumption 2.5}. This representation is
the foundation of \cite{Ruszczynski2010} and subsequent works on
time-consistent risk measures. For this result, we define a mapping
$\rho_{t}\text{ : }L_{t+1}\rightarrow L_{t}$, where $t\geq0$, to
be a one-step (conditional) risk measure if $\rho_{t}(X_{t+1})=\rho_{t,\,t+1}(0,\,X_{t+1})$.
\begin{theorem}
	\label{Theorem 2.3}\cite[Theorem 1]{Ruszczynski2010} Suppose Assumption
	\ref{Assumption 2.5} holds, then
	\begin{align}
		\rho_{t,\,T}(X_{t},\,X_{t+1},...,\,X_{T},\ldots) = X_{t}+\rho_{t}(X_{t+1}+\rho_{t+1}(X_{t+2} +\cdot\cdot\cdot+\rho_{T}(X_{T})+\cdot\cdot\cdot)),\label{Recursive}
	\end{align}
	for all $0\leq t\leq T$, where $\rho_{t},\ldots,\,\rho_{T}$ are
	one-step conditional risk measures.
\end{theorem}

Now we may consider the risk of an infinite cost sequence. Based on
\cite{Ruszczynski2010}, the \emph{discounted} measure of risk $\rho_{t,\,T}^{\gamma}\text{ : }L_{t,\,T}\rightarrow\mathbb{R}$
is defined via
\begin{align*}
	\rho_{t,\,T}^{\gamma}(X_{t},\,X_{t+1},\ldots,\,X_{T}):= \rho_{t,\,T}(\gamma^{t}X_{t},\,\gamma^{t+1}X_{t+1},\ldots,\,\gamma^{T}X_{T}).
\end{align*}
Define $L_{t,\,\infty}:=L_{t}\times L_{t+1}\times\cdot\cdot\cdot$
for $t\geq0$ and $\rho^{\gamma}\text{ : }L_{0,\,\infty}\rightarrow\mathbb{R}$
via
\[
\rho^{\gamma}(X_{0},\,X_{1},\ldots):=\lim_{T\rightarrow\infty}\rho_{0,\,T}^{\gamma}(X_{0},\,X_{1},\ldots,\,X_{T}).
\]
To provide our final representation result, we introduce the additional
assumption that player risk preferences are stationary (they only depend
on the sequence of costs ahead, and are independent of the current
time).
\begin{assumption}
	\label{assu:stationary_preferences} (Stationary risk preferences) For
	all $T\geq1$ and $s\geq0$,
	\[
	\rho_{0,\,T}^{\gamma}(X_{0},\,X_{1},\ldots,\,X_{T})=\rho_{s,\,T+s}^{\gamma}(X_{0},\,X_{1},\ldots,\,X_{T}).
	\]
\end{assumption}

When Assumptions \ref{Assumption 2.5} and \ref{assu:stationary_preferences}
are satisfied, the corresponding dynamic risk measure is given by
the recursion:
\begin{align}
	\rho^{\gamma}(X_{0},\,X_{1},...,\,X_{T},\ldots)=X_{0}+\rho_{1}(\gamma X_{1}+\rho_{2}(\gamma^{2}X_{2}+\cdot\cdot\cdot+\rho_{T}(\gamma^{T}X_{T})+\cdot\cdot\cdot)),\label{Recursive_infinite}
\end{align}
where $\rho_{1},\,\rho_{2},\ldots$ are all one-step risk measures.
Based on representation (\ref{Recursive_infinite}), we may define
the risk-aware objective for player $i$ to be:
\begin{align}
	J_{s_{0}}^{i}(x^{i},\,x^{-i})=\rho^{i}(c^{i}(s_{0},\,a_{0})+\gamma\,\rho^{i}(c^{i}(s_{1},\,a_{1})+\gamma\,\rho^{i}(c^{i}(s_{2},\,a_{2})+\cdot\cdot\cdot))).\label{Recursive_risk}
\end{align}
where $\rho^{i}$ is a one-step conditional risk measure that maps
random cost from the next stage to current stage, with respect to the
joint distribution of randomized mixed strategies and transition
kernels. In formulation (\ref{Recursive_risk}), each $c^{i}(s_{t},\,a_{t}),\,t\geq1$
is governed by the joint distribution of randomized mixed strategies
and transition kernel
\[
\times_{i\in\mathcal{I}}~x_{s_{t}}^{i}(a_{t}^{i})P(s_{t}|s_{t-1},\,a_{t-1}),
\]
which is defined for fixed $(s_{t-1},\,a_{t-1})$ and for all $s_{t}$
and $a_{t}^{i}$. The distribution of $c^{i}$ is only governed by  $\times_{i\in\mathcal{I}}~x_{s_{0}}^{i}(a_{0}^{i})$.

\subsection{Proof of Theorem 1}

\subsubsection*{Fundamental Inequalities}

We make heavy use of the following fundamental inequalities and algebraic
identity. 
\begin{fact}
	\label{fact:min-max} Let $\mathcal{X}$ be a nonempty set and $f_{1},\,f_{2}\text{ : }\mathcal{X}\rightarrow\mathbb{R}$
	be functions on $\mathcal{X}$.
	
	(i) $|\min_{x\in\mathcal{X}}f_{1}\left(x\right)-\min_{x\in\mathcal{X}}f_{2}\left(x\right)|\leq\max_{x\in\mathcal{X}}|f_{1}\left(x\right)-f_{2}\left(x\right)|$.
	
	(ii) $|\max_{x\in\mathcal{X}}f_{1}\left(x\right)-\max_{x\in\mathcal{X}}f_{2}\left(x\right)|\leq\max_{x\in\mathcal{X}}|f_{1}\left(x\right)-f_{2}\left(x\right)|$.
\end{fact}

\begin{proof}
	For part (i), we compute
	\begin{align*}
		\min_{x\in\mathcal{X}}f_{1}\left(x\right)=\, & \min_{x\in\mathcal{X}}\left(f_{1}\left(x\right)-f_{2}\left(x\right)+f_{2}\left(x\right)\right)\\
		\leq\, & \min_{x\in\mathcal{X}}\left(f_{2}\left(x\right)+|f_{1}\left(x\right)-f_{2}\left(x\right)|\right)\\
		\leq\, & \min_{x\in\mathcal{X}}f_{2}\left(x\right)+\max_{x\in\mathcal{X}}|f_{1}\left(x\right)-f_{2}\left(x\right)|
	\end{align*}
	which gives
	\[
	\min_{x\in\mathcal{X}}f_{1}\left(x\right)-\min_{x\in\mathcal{X}}f_{2}\left(x\right)\leq\max_{x\in\mathcal{X}}|f_{1}\left(x\right)-f_{2}\left(x\right)|.
	\]
	By a symmetric argument, we have
	\[
	\min_{x\in\mathcal{X}}f_{2}\left(x\right)-\min_{x\in\mathcal{X}}f_{1}\left(x\right)\leq\max_{x\in\mathcal{X}}|f_{1}\left(x\right)-f_{2}\left(x\right)|,
	\]
	from which the desired conclusion follows. The proof for part (ii)
	is similar.
\end{proof}
Define $d_{\mathcal{H}}(\mathfrak{A},\,\mathfrak{B}\text{)}$ to be
the Hausdorff distance between nonempty subsets $\mathfrak{A}$ and
$\mathfrak{B}$ of $\mathbb{R}^{d}$ with respect to the Euclidean
norm $\|\cdot\|_{2}$, explicitly,
\[
	d_{\mathcal{H}}(\mathfrak{A},\,\mathfrak{B}\text{)}:=\max\left\{ \sup_{a\in\mathfrak{A}}\inf_{b\in\mathfrak{B}}\|a-b\|_{2},\,\sup_{b\in\mathfrak{B}}\inf_{a\in\mathfrak{A}}\|a-b\|_{2}\right\} .
	\]

\begin{fact}
\label{fact:Hausdorff} Let $\left(\mathcal{X},\,\|\cdot\|\right)$
be a normed space, $f\text{ : }\mathcal{X}\rightarrow\mathbb{R}$
be an $L_{f}$-Lipschitz function, and $\mathcal{X}_{1},\,\mathcal{X}_{2}\subset\mathcal{X}$.
Then
\[
|\min_{x\in\mathcal{X}_{1}}f\left(x\right)-\min_{x\in\mathcal{X}_{2}}f\left(x\right)|\leq L_{f}d_{\mathcal{H}}(\mathcal{X}_{1},\,\mathcal{X}_{2}).
\]
\end{fact}

\begin{proof}
Let $x_{1}^{*}\in\mathcal{X}_{1}$ be an optimal solution of $\min_{x\in\mathcal{X}_{1}}f\left(x\right)$.
There is an $x_{2}\in\mathcal{X}_{2}$ by definition of $d_{\mathcal{H}}$such
that $\|x_{1}^{*}-x_{2}\|\leq d_{\mathcal{H}}(\mathcal{X}_{1},\,\mathcal{X}_{2})$,
and so
\begin{align*}
	\min_{x\in\mathcal{X}_{2}}f\left(x\right)\leq f(x_{2})\leq f(x_{1}^{*})+L_{f}d_{\mathcal{H}}(\mathcal{X}_{1},\,\mathcal{X}_{2}) =\min_{x\in\mathcal{X}_{1}}f(x_{1})+L_{f}d_{\mathcal{H}}(\mathcal{X}_{1},\,\mathcal{X}_{2}),
\end{align*}
where the second inequality follows by $L_{f}$-Lipschitz continuity
of $f$. The other direction follows by symmetric reasoning.
\end{proof}
\begin{fact}
\label{fact: Identity} Define \[\delta_{0}(\epsilon):=\frac{\min\{\epsilon,\,1\}}{(2^{|\mathcal{I}|}-1)|\mathcal{A}|},\]
then $\left|\Pi_{i\in\mathcal{I}}x_{s}^{i}(a^{i})-\Pi_{i\in\mathcal{I}}y_{s}^{i}(a^{i})\right|\leq\epsilon$
holds for any $x,\,y\in\mathcal{X}$.
\end{fact}

\begin{proof}
We make use of the following algebraic identity
\begin{align*}
			&\left|\Pi_{i\in\mathcal{I}}x_{s}^{i}(a^{i})-\Pi_{i\in\mathcal{I}}y_{s}^{i}(a^{i})\right|
			\\
			= & \left|\sum_{\left\{ \Omega\subset\mathcal{I},\,\Omega\ne\emptyset\right\} }\left(\prod_{\begin{array}{c}
					i\in\Omega\end{array}}(x_{s}^{i}(a^{i})-y_{s}^{i}(a^{i})\right)\left(\prod_{\begin{array}{c}
					i\notin\Omega\end{array}}y_{s}^{i}(a^{i})\right)\right|\\
			\leq & \sum_{\left\{ \Omega\subset\mathcal{I},\,\Omega\ne\emptyset\right\} }\left|\prod_{\begin{array}{c}
					i\in\Omega\end{array}}(x_{s}^{i}(a^{i})-y_{s}^{i}(a^{i})\right|\\
			\leq & \sum_{\left\{ \Omega\subset\mathcal{I},\,\Omega\ne\emptyset\right\} }\left(\delta_{0}(\epsilon)\right)^{|\Omega|}\\
			\leq & \sum_{\left\{ \Omega\subset\mathcal{I},\,\Omega\ne\emptyset\right\} }\delta_{0}(\epsilon)\\
			= & \epsilon.
\end{align*}
\end{proof}

\subsubsection*{Existence of Stationary Equilibria}

This section develops the machinery for our proof of Theorem 1, which
is based on Kakutani's fixed point theorem.
\begin{theorem}
\label{Theorem 3.5} \cite{Kakutani1941} (Kakutani's fixed point
theorem) If $\mathcal{X}$ is a closed, bounded, and convex set in
Euclidean space, and $\Phi$ is an upper semicontinuous correspondence
mapping from $\mathcal{X}$ into the family of closed, convex subsets
of $\mathcal{X}$, then there exists an element $x\in\mathcal{X}$
such that $x\in\Phi(x)$. 
\end{theorem}

For a multi-strategy $x\in\mathcal{X}$, we define the operator $\mathcal{T}_{x}:\,\mathcal{V}\to\mathcal{V}$
via
\begin{align*}
[\mathcal{T}_{x}(v)]_{s}^{i}:=\min_{u_{s}^{i}\in\mathcal{P}\left(\mathcal{A}^{i}\right)}\{ \rho^{i}\left(c^{i}(s\,,A)+\gamma\,v^{i}\left(S'\right)\right)\text{ : }(A,\,S')\sim P_{s}(u_{s}^{i},\,x_{s}^{-i})\} ,\,\forall i\in\mathcal{I},\,s\in\mathcal{S}.\label{Operator}
\end{align*}
For simplification, we define, for any $v\in\mathcal{V}$, $i\in\mathcal{I}$
and $s\in\mathcal{S}$,
\begin{align*}
\psi_{s}^{i}(u_{s}^{i},\,x_{s}^{-i},\,v^{i}):=\rho^{i}\left(c^{i}(s\,,A)+\gamma\,v^{i}\left(S'\right)\right),\,(A,\,S')\sim P_{s}(u_{s}^{i},\,x_{s}^{-i}).
\end{align*}
Our proof of Theorem 1 has three main steps:
\begin{enumerate}
\item (Step 1) Show that $\mathcal{T}_{x}$ is a contraction.
\item (Step 2) Show that the cost-to-go function $\psi_{s}^{i}(u_{s}^{i},\,x_{s}^{-i},\,v^{i})$
is continuous.
\item (Step 3) Verify that the assumptions of Kakutani's fixed point theorem
are met for $\Phi$.
\end{enumerate}

\subsubsection*{Step 1: Show that $\mathcal{T}_{x}$ is a contraction}

We establish that the operator $\mathcal{T}_{x}$ is a contraction,
and subsequently that given any stationary strategy $x$, there is
a corresponding unique value function $v^{i}$ for all players. 
\begin{proposition}
\label{Theorem 3.2-1} For each $x\in\mathcal{X}$, $\mathcal{T}_{x}$
is a contraction with constant $\gamma\in\left(0,\,1\right)$.
\end{proposition}

\begin{proof}
Let $v,\,w\in\mathcal{V}$. For $i\in\mathcal{I}$ and $s\in\mathcal{S}$,
let $g_{s}^{i}$ attain the minimum in the definition of $\left[\mathcal{T}_{x}(v)\right]_{s}^{i}$,
\begin{align*}
	[\mathcal{T}_{x}(v)]_{s}^{i}=\min_{u_{s}^{i}\in\mathcal{P}(\mathcal{A}^{i})}\psi_{s}^{i}\left(u_{s}^{i},\,x_{s}^{-i},\,v^{i}\right)=\sup_{\mu\in\mathcal{M}_{s}^{i}(P_{s}(g_{s}^{i},\,x_{s}^{-i}))}\left\{ \langle\mu,\,C_{s}^{i}(v^{i})\rangle-\alpha_{s}^{i}(\mu)\right\} .
\end{align*}
Similarly, let $z_{s}^{i}$ attain the minimum in $\left[\mathcal{T}_{x}(w)\right]_{s}^{i}$,
\[
[\mathcal{T}_{x}(w)]_{s}^{i}=\sup_{\mu\in\mathcal{M}_{s}^{i}(P_{s}(z_{s}^{i},\,x_{s}^{-i}))}\left\{ \langle\mu,\,C_{s}^{i}(w^{i})\rangle-\alpha_{s}^{i}(\mu)\right\} .
\]
It follows that
\begin{align*}
	& [\mathcal{T}_{x}(v)]_{s}^{i}-[\mathcal{T}_{x}(w)]_{s}^{i}\\
	=\, & \sup_{\mu\in\mathcal{M}_{s}^{i}(P_{s}(g_{s}^{i},\,x_{s}^{-i}))}\left\{ \langle\mu,\,C_{s}^{i}(v^{i})\rangle-\alpha_{s}^{i}(\mu)\right\} -\sup_{\mu\in\mathcal{M}_{s}^{i}(P_{s}(z_{s}^{i},\,x_{s}^{-i}))}\left\{ \langle\mu,\,C_{s}^{i}(w^{i})\rangle-\alpha_{s}^{i}(\mu)\right\} \\
	\leq\, & \sup_{\mu\in\mathcal{M}_{s}^{i}(P_{s}(z_{s}^{i},\,x_{s}^{-i}))}\left\{ \langle\mu,\,C_{s}^{i}(v^{i})\rangle-\alpha_{s}^{i}(\mu)\right\} -\sup_{\mu\in\mathcal{M}_{s}^{i}(P_{s}(z_{s}^{i},\,x_{s}^{-i}))}\left\{ \langle\mu,\,C_{s}^{i}(w^{i})\rangle-\alpha_{s}^{i}(\mu)\right\} \\
	\leq\, & \gamma\sup_{\mu\in\mathcal{M}_{s}^{i}(P_{s}(z_{s}^{i},\,x_{s}^{-i}))}\sum_{\left(a,\,s'\right)\in\mathcal{K}}\mu\left(a,\,s'\right)|v^{i}\left(s'\right)-w^{i}\left(s'\right)|\\
	\leq\, & \gamma\|v-w\|_{\infty},
\end{align*}
where the first inequality holds by choice of $g_{s}^{i}$. The argument
to upper bound $[\mathcal{T}_{x}(w)]_{s}^{i}-[\mathcal{T}_{x}(v)]_{s}^{i}$
is symmetric.
\end{proof}
Since $(\mathcal{V},\,\|\cdot\|_{\infty})$ is a complete metric space
and $\mathcal{T}_{x}$ is a contraction mapping on $\mathcal{V}$
by Proposition \ref{Theorem 3.2-1}, $\mathcal{T}_{x}$ has a unique
fixed point by the Banach fixed point theorem.
\begin{proposition}
\label{Theorem 3.3-1} For any stationary strategy $x\in\mathcal{X}$,
there exists a unique value function $v\in\mathcal{V}$ such that $\forall s\in\mathcal{S},\,i\in\mathcal{I}$, 
\begin{align*}
	v^{i}\left(s\right)=\min_{u_{s}^{i}\in\mathcal{P}\left(\mathcal{A}^{i}\right)}\rho^{i}(c^{i}(s,\,A)+\gamma\,v^{i}\left(S'\right))=\min_{u_{s}^{i}\in\mathcal{P}\left(\mathcal{A}^{i}\right)}\psi_{s}^{i}(u_{s}^{i},\,x_{s}^{-i},\,v^{i}).
\end{align*}
\end{proposition}

\subsubsection*{Step 2: Show that $\psi_{s}^{i}$ is continuous}

We want to establish continuity of $\psi_{s}^{i}\left(x_{s}^{i},\,x_{s}^{-i},\,v^{i}\right)$
in all its arguments for all $i\in\mathcal{I}$ and $s\in\mathcal{S}$. Firstly, we know
that the function $(\mu,\,v^{i})\rightarrow\langle\mu,\,C_{s}^{i}(v^{i})\rangle-\alpha_{s}^{i}(\mu)$
is Lipschitz continuous based on Assumption 1. We use $L$ to denote
the Lipschitz constant. Our argument is based on the following chain
of inequalities: 

\begin{align*}
& |\psi_{s}^{i}(x_{s}^{i},\,x_{s}^{-i},\,v^{i})-\psi_{s}^{i}(y_{s}^{i},\,y_{s}^{-i},\,w^{i})|\\
\leq\, & |\psi_{s}^{i}(x_{s}^{i},\,x_{s}^{-i},\,v^{i})-\psi_{s}^{i}(x_{s}^{i},\,x_{s}^{-i},\,w^{i})|+|\psi_{s}^{i}(x_{s}^{i},\,x_{s}^{-i},\,w^{i})-\psi_{s}^{i}(y_{s}^{i},\,y_{s}^{-i},\,w^{i})|\\
\leq\, & |\sup_{\mu\in\mathcal{M}_{s}^{i}(P_{s}(x_{s}))}\left\{ \langle\mu,\,C_{s}^{i}(v^{i})\rangle-\alpha_{s}^{i}(\mu)\right\} -\sup_{\mu\in\mathcal{M}_{s}^{i}(P_{s}(x_{s}))}\left\{ \langle\mu,\,C_{s}^{i}(w^{i})\rangle-\alpha_{s}^{i}(\mu)\right\} |
\\
& +|\sup_{\mu\in\mathcal{M}_{s}^{i}(P_{s}(x_{s}))}\left\{ \langle\mu,\,C_{s}^{i}(w^{i})\rangle-\alpha_{s}^{i}(\mu)\right\} -\sup_{\mu\in\mathcal{M}_{s}^{i}(P_{s}(y_{s}))}\left\{ \langle\mu,\,C_{s}^{i}(w^{i})\rangle-\alpha_{s}^{i}(\mu)\right\} |\\
\leq\, & \gamma\,\|v^{i}-w^{i}\|_{\infty}+L\,d_{\mathcal{H}}\left(\mathcal{M}_{s}^{i}(P_{s}(x_{s})),\,\mathcal{M}_{s}^{i}(P_{s}(y_{s}))\right),
\end{align*}
where we use Fact \ref{fact:Hausdorff} to obtain the last inequality.
Let $\text{Ex}\left(\mathcal{M}_{s}^{i}(P_{s}(x_{s}))\right)$ denote
the set of extreme points of (the bounded polyhedron) $\mathcal{M}_{s}^{i}(P_{s}(x_{s}))$,
then we also have 

\begin{align}
&|\psi_{s}^{i}(x_{s}^{i},\,x_{s}^{-i},\,v^{i})-\psi_{s}^{i}(y_{s}^{i},\,y_{s}^{-i},\,w^{i})|\nonumber
\\
\leq & \gamma\,\|v^{i}-w^{i}\|_{\infty}+L\,d_{\mathcal{H}}\left(\mathcal{M}_{s}^{i}(P_{s}(x_{s})),\,\mathcal{M}_{s}^{i}(P_{s}(y_{s}))\right),\nonumber \\
\leq & \gamma\,\|v^{i}-w^{i}\|_{\infty}+L\,d_{\mathcal{H}}(\text{Ex}(\mathcal{M}_{s}^{i}(P_{s}(x_{s}))),\,\text{Ex}(\mathcal{M}_{s}^{i}(P_{s}(y_{s})))),\label{Union equality}
\end{align}
where the last inequality is from \cite[Theorem 1]{Walkup1969}. 

We define the following metrics for stationary strategies and value
functions:
\begin{align*}
d_{\mathcal{X}_{s}}(x_{s},\,y_{s}):=\, & \max_{i\in\mathcal{I},\,a\in\mathcal{A}}|x_{s}^{i}\left(a\right)-y_{s}^{i}\left(a\right)|,\, \forall s\in\mathcal{S},\\
d_{\mathcal{V}^{i}}(v^{i},\,w^{i}):=\, & \max_{s\in\mathcal{S}}|v^{i}\left(s\right)-w^{i}\left(s\right)|,
\end{align*}
a metric on $\mathcal{X}_{s}\times\mathcal{V}^{i}$ is then given
by $d_{\mathcal{X}_{s}\times\mathcal{V}^{i}}\left((x_{s},\,v^{i}),\,(y_{s},\,w^{i})\right):=d_{\mathcal{X}_{s}}(x_{s},\,y_{s})+d_{V^{i}}(v^{i},\,w^{i})$.

In the next lemma, we show that the extreme points of $\mathcal{M}_{s}^{i}(P_{s}(x_{s}))$
and $\mathcal{M}_{s}^{i}(P_{s}(y_{s}))$ are ``close'' when the
stationary strategies $x_{s}$ and $y_{s}$ are ``close''.
\begin{lemma}
\label{Lem:Continuity}\cite[Theorem 3.1]{Batson1987} Choose $\epsilon>0$
and $B>0$. For $x_{s},\,y_{s}\in\mathcal{X}_{s}$, suppose

\[
d_{\mathcal{X}_{s}}(x_{s},\,y_{s})\leq\delta_{1}(\epsilon):=\frac{\min\{1,\,\epsilon\}}{B\,\sqrt{|\mathcal{S}|\,|\mathcal{A}|}(2^{|\mathcal{I}|}-1)|\mathcal{A}|},
\]
then
\begin{align*}
	\textrm{\ensuremath{d_{\mathcal{H}}\left(\textrm{Ex}\left(\mathcal{M}_{s}^{i}(P_{s}(x_{s}))\right),\,\textrm{Ex}\left(\mathcal{M}_{s}^{i}(P_{s}(y_{s}))\right)\right)}} & \leq\epsilon.
\end{align*}
\end{lemma}

\begin{proof}
By definition, we have
\begin{align*}
	\left[P_{s}(x_{s})\right]\left(a,\,k\right)-\left[P_{s}(y_{s})\right]\left(a,\,k\right)=\left(\Pi_{i\in\mathcal{I}}x_{s}^{i}(a^{i})-\Pi_{i\in\mathcal{I}}y_{s}^{i}(a^{i})\right)P(k\,\vert\,s,\,a).
\end{align*}
Now, suppose
\[
d_{\mathcal{X}_{s}}(x_{s},\,y_{s})\leq\delta_{0}(\epsilon):=\frac{\min\{\epsilon,\,1\}}{(2^{|\mathcal{I}|}-1)|\mathcal{A}|}.
\]
Using the Fact \ref{fact: Identity}, we have
\begin{align*}
	&\|P_{s}(x_{s})-P_{s}(y_{s})\|_{2} 
	\\
	&\leq\sqrt{|\mathcal{S}|\,|\mathcal{A}|}\left|\Pi_{i\in\mathcal{I}}x_{s}^{i}(a^{i})-\Pi_{i\in\mathcal{I}}y_{s}^{i}(a^{i})\right|\max_{a\in\mathcal{A},\,k\in\mathcal{S}}P(k\,\vert\,s,\,a)\\
	& \leq\sqrt{|\mathcal{S}|\,|\mathcal{A}|}\left|\Pi_{i\in\mathcal{I}}x_{s}^{i}(a^{i})-\Pi_{i\in\mathcal{I}}y_{s}^{i}(a^{i})\right|\\
	& =\sqrt{|\mathcal{S}|\,|\mathcal{A}|}\,\epsilon.
\end{align*}
By \cite[Theorem 3.1]{Batson1987}, it follows that there exists a
constant $B>0$ such that 
\begin{align*}
	d_{\mathcal{H}}\left(\text{Ex}\left(\mathcal{M}_{s}^{i}(P_{s}(x_{s}))\right),\,\text{Ex}\left(\mathcal{M}_{s}^{i}(P_{s}(y_{s}))\right)\right)\leq  B\|P_{s}-P_{s}^{\prime}\|_{2}\leq\sqrt{|\mathcal{S}|\,|\mathcal{A}|}\,B\,\epsilon,
\end{align*}
and the desired conclusion follows.
\end{proof}
As a consequence, we have the following lemma.
\begin{lemma}
\label{Lem:Continuity-2} There exist constants $B,\,C>0$ such that,
for any $\epsilon>0$, if
\begin{align*}
	d_{\mathcal{X}_{s}\times\mathcal{V}^{i}}\left((x_{s},\,v^{i}),\,(y_{s},\,w^{i})\right)\leq\delta_{2}(\epsilon):=\max\left\{ \frac{\min\{1,\,\epsilon\}}{2L\,B\,\sqrt{|\mathcal{S}|\,|\mathcal{A}|}(2^{|\mathcal{I}|}-1)|\mathcal{A}|},\,\frac{\epsilon}{2\gamma}\right\} ,
\end{align*}
then
\[
|\psi_{s}^{i}(x_{s}^{i},\,x_{s}^{-i},\,v^{i})-\psi_{s}^{i}(y_{s}^{i},\,y_{s}^{-i},\,w^{i})|\leq\epsilon.
\]
\end{lemma}

\begin{proof}
By inequality (\ref{Union equality}), we have
\begin{align*}
	& |\psi_{s}^{i}(x_{s}^{i},\,x_{s}^{-i},\,v^{i})-\psi_{s}^{i}(y_{s}^{i},\,y_{s}^{-i},\,w^{i})|\\
	\leq\, & \gamma\,\|v^{i}-w^{i}\|_{\infty}+L\,d_{\mathcal{H}}\left(\mathcal{M}_{s}^{i}(P_{s}(x_{s})),\,\mathcal{M}_{s}^{i}(P_{s}(y_{s}))\right)\\
	\leq\, & \gamma\,\|v^{i}-w^{i}\|_{\infty}
	\\
	&+L\,d_{\mathcal{H}}\left(\text{Ex}\left(\mathcal{M}_{s}^{i}(P_{s}(x_{s}))\right),\,\text{Ex}\left(\mathcal{M}_{s}^{i}(P_{s}(y_{s}))\right)\right)
	\\
	\leq&\epsilon.
\end{align*}
\end{proof}
As a consequence of these lemmas, we have the following proposition. 
\begin{proposition}
\label{Thm:Continuity} For all $i\in\mathcal{I}$ and $s\in\mathcal{S}$,
the function $\psi_{s}^{i}(u_{s}^{i},\,x_{s}^{-i},\,v^{i})$ is continuous
in all of its variables.
\end{proposition}

\subsubsection*{Step 3: Apply Kakutani's fixed point theorem}

In the following two technical results, we establish upper semicontinuity
of $\Phi$.
\begin{definition}
\label{Definition 2.9} A correspondence $\Phi:\,\mathcal{X}\rightarrow2^{\mathcal{X}}$
is upper semicontinuous if $y^{n}\in\Phi(x^{n})$, $\lim_{n\rightarrow\infty}x^{n}=x$,
and $\lim_{n\rightarrow\infty}y^{n}=y$ implies $y\in\Phi(x)$.
\end{definition}

The proof of Lemma \ref{Lem:Kakutani} below follows directly from
\cite{Fink1964} and our earlier Lemma \ref{Lem:Continuity-2}. The
proof of Lemma \ref{Lem:Kakutani-1} follows directly from Lemma \ref{Lem:Kakutani}
and \cite{Fink1964}. Lemma \ref{Lem:Kakutani} is then used to prove
Lemma \ref{Lem:Kakutani-1}, and Lemma \ref{Lem:Kakutani-1} is used
to establish upper semicontinuity.

We define the operator $\mathcal{D}_{s}^{i}$ as follows:
\[
\mathcal{D}_{s}^{i}(x_{s}^{-i},\,v^{i}):=\min_{u_{s}^{i}\in\mathcal{P}\left(\mathcal{A}^{i}\right)}\psi_{s}^{i}(u_{s}^{i},\,x_{s}^{-i},\,v^{i}),
\]
it returns the optimal risk-to-go for any $i\in\mathcal{I}$ and $s\in\mathcal{S}$
as a function of the complementary strategy $x_{s}^{-i}$ and the
value function $v^{i}$.
\begin{lemma}
The operator \label{Lem:Kakutani} $\mathcal{D}_{s}^{i}(x_{s}^{-i},\,v^{i})$ is
continuous in $x_{s}^{-i}$. Furthermore, the collection of functions
\[\{\mathcal{D}_{s}^{i}(\cdot,\,v^{i})\text{ : }\|v^{i}\|_{\infty}<\infty\},\]
is equicontinuous. 
\end{lemma}

\begin{proof}
Let
\begin{align*}
	& \mathcal{D}_{s}^{i}(x_{s}^{-i},\,v^{i})=\min_{u_{s}^{i}\in\mathcal{P}\left(\mathcal{A}^{i}\right)}\psi_{s}^{i}(u_{s}^{i},\,x_{s}^{-i},\,v^{i})=\psi_{s}^{i}(w_{s}^{i},\,x_{s}^{-i},\,v^{i}),\\
	&\mathcal{D}_{s}^{i}(y_{s}^{-i},\,v^{i})=\min_{u_{s}^{i}\in\mathcal{P}\left(\mathcal{A}^{i}\right)}\psi_{s}^{i}(u_{s}^{i},\,y_{s}^{-i},\,v^{i})=\psi_{s}^{i}(z_{s}^{i},\,y_{s}^{-i},\,v^{i}).
\end{align*}
Then
\begin{align*}
	\mathcal{D}_{s}^{i}(y_{s}^{-i},\,v^{i})-\mathcal{D}_{s}^{i}(x_{s}^{-i},\,v^{i})\leq\psi_{s}^{i}(w_{s}^{i},\,y_{s}^{-i},\,v^{i})-\psi_{s}^{i}(z_{s}^{i},\,x_{s}^{-i},\,v^{i}),
\end{align*}
and
\begin{align*}
	\mathcal{D}_{s}^{i}(x_{s}^{-i},\,v^{i})-\mathcal{D}_{s}^{i}(y_{s}^{-i},\,v^{i})\leq\psi_{s}^{i}(z_{s}^{i},\,x_{s}^{-i},\,v^{i})-\psi_{s}^{i}(z_{s}^{i},\,y_{s}^{-i},\,v^{i}).
\end{align*}
If $v^{i}$ is bounded, then the right-hand side of the above inequalities
can be made arbitrarily small through control of $x_{s}^{-i}$ and
$y_{s}^{-i}$ via Proposition \ref{Thm:Continuity}.
\end{proof}
Let us define a mapping from stationary strategies to value functions
via
\begin{align*}
\tau^{i}(x^{-i}):=\{v^{i}=(v^{i}\left(s\right))_{s\in\mathcal{S}}:\,,v^{i}\left(s\right)=\min_{u_{s}^{i}\in\mathcal{P}\left(\mathcal{A}^{i}\right)}\psi_{s}^{i}(u_{s}^{i},\,x_{s}^{-i},\,v^{i}),\,s\in\mathcal{S}\} ,\,\forall i\in\mathcal{I}.
\end{align*}
Each $\tau^{i}(x^{-i})$ returns the value function for player $i$
corresponding to a best response to the complementary strategy $x^{-i}$. Denote the $s^{th}$
element of $\tau^{i}(x^{-i})$ by $\tau_{s}^{i}(x^{-i})$, let $(x^{-i})_{n}$
be a sequence of mixed strategies of all players satisfying $\lim_{n\rightarrow\infty}(x^{-i})_{n}=x^{-i}$,
and let the corresponding value functions for player $i$ be $\tau^{i}(\left(x^{-i}\right)_{n})$.
The proof of the next Lemma \ref{Lem:Kakutani-1} follows directly
from Proposition \ref{Thm:Continuity} as shown in \cite{Fink1964}.
\begin{lemma}
\label{Lem:Kakutani-1} If $(x^{-i})_{n}\rightarrow x^{-i}$ and $\tau_{s}^{i}((x^{-i})_{n})\rightarrow v^{i}(s)$
as $n\rightarrow\infty$, then $\tau_{s}^{i}(x^{-i})=v^{i}(s)$.
\end{lemma}

The proof of our main result Theorem 1 is encapsulated in the following
three lemmas.
\begin{lemma}
For all $x\in\mathcal{X}$, the set $\Phi\left(x\right)$ is nonempty
and is a subset of $\mathcal{X}$.
\end{lemma}

\begin{proof}
By Proposition \ref{Thm:Continuity}, $\psi_{s}^{i}(u_{s}^{i},\,x_{s}^{-i},\,v^{i})$
is continuous in all of its arguments. By the Weierstrass theorem,
the minimum of this function on the compact set $\mathcal{P}(\mathcal{A}^{i})$
exists and is attained. Thus, the equality:
\[
v^{i}(s)=\min_{u_{s}^{i}\in\mathcal{P}\left(\mathcal{A}^{i}\right)}\psi_{s}^{i}(u_{s}^{i},\,x_{s}^{-i},\,v^{i}),
\]
can be established and therefore $\Phi(x)\neq\emptyset$. By definition,
$\Phi(x)\subseteq\mathcal{X}$ for all $x\in\mathcal{X}$.
\end{proof}
The following intermediate result plays a key role in proving that
$\Phi(x)$ is convex for all $x\in\mathcal{X}$. 
\begin{lemma}
\label{Lemma B 12} Suppose that $(g^{i})_{i\in\mathcal{I}},\,(z^{i})_{i\in\mathcal{I}}\in\Phi(x)$,
then given $x^{-i}\in\mathcal{X}^{-i}$ we have
\begin{align*}
	\mathcal{M}_{s}^{i}\left[P_{s}((1-\lambda)\,z^{i}+\lambda\,g_{s}^{i},\,x_{s}^{-i})\right]\subseteq\mathcal{M}_{s}^{i}\left[P_{s}(z_{s}^{i},\,x_{s}^{-i})\right]\cup\mathcal{M}_{s}^{i}\left[P_{s}(g_{s}^{i},\,x_{s}^{-i})\right],
\end{align*}
for any $\lambda\in[0,\,1]$.
\end{lemma}

\begin{proof}
From Assumption 1(ii), we know that $\mathcal{M}_{s}^{i}(P_{s})$
is a polyhedron characterized by formulation (6).
Suppose that $f_{m},\,m=1,\,....,\,M$ are linear in $u_{s}$, $x_{s}^{-i}$
, and transition kernel $P_{s}$. It follows that for any $\lambda\in[0,\,1]$
and $\mu\in\mathcal{M}_{s}^{i}\left[P_{s}((1-\lambda)\,z^{i}+\lambda\,g_{s}^{i},\,x_{s}^{-i})\right]$
satisfies for $m=1,\,...,\,M$, 
\[
A_{s,\,m}^{i}\,\mu+f_{m}((1-\lambda)\,z^{i}+\lambda\,g_{s}^{i},\,x_{s}^{-i},\,P_{s})\geq h_{s,\,m}^{i}. 
\]
Since $f_{m}$, for $m=1,\,....,\,M$ are linear (and thus quasiconvex)
functions, for any $\lambda\in[0,\,1]$
and $\mu\in\mathcal{M}_{s}^{i}(P_{s}((1-\lambda)\,z^{i}+\lambda\,g_{s}^{i},\,x_{s}^{-i}))$
for $m=1,\,...,\,M$, the inequalities
\begin{align*}
	& A_{s,\,m}^{i}\,\mu+\max\left\{ f_{m}(z_{s}^{i},\,x_{s}^{-i},\,P_{s}),\,f_{m}(g_{s}^{i},\,x_{s}^{-i},\,P_{s})\right\} \\
	\geq\, & A_{s,\,m}^{i}\,\mu+f_{m}((1-\lambda)\,z^{i}+\lambda\,g_{s}^{i},\,x_{s}^{-i},\,P_{s})\geq h_{s,\,m}^{i},
\end{align*}
also hold. Thus, for any $\lambda\in[0,\,1]$ and $\mu\in\mathcal{M}_{s}^{i}\left[P_{s}((1-\lambda)\,z^{i}+\lambda\,g_{s}^{i},\,x_{s}^{-i})\right]$,
at least one of $\mu\in\mathcal{M}_{s}^{i}\left[P_{s}(z_{s}^{i},\,x_{s}^{-i})\right]$
or $\mu\in\mathcal{M}_{s}^{i}\left[P_{s}(g_{s}^{i},\,x_{s}^{-i})\right]$
must hold.
\end{proof}
We now establish the convexity of $\Phi(x)$.
\begin{lemma}
\label{lem:Phi_convex}The set $\Phi(x)$ is convex.
\end{lemma}

\begin{proof}
Suppose that $(g^{i})_{i\in\mathcal{I}},\,(z^{i})_{i\in\mathcal{I}}\in\Phi(x)$.
For all $u_{s}^{i}\in\mathcal{P}(\mathcal{A}^{i})$ we have
\begin{align*}
	v^{i}(s)=\psi_{s}^{i}(g_{s}^{i},\,x_{s}^{-i},\,v^{i})=\psi_{s}^{i}(z_{s}^{i},\,x_{s}^{-i},\,v^{i})\leq\psi_{s}^{i}(u_{s}^{i},\,x_{s}^{-i},\,v^{i}).
\end{align*}
The above inequality holds based on the definition of operator $\Phi(x)$
via formulation (9), which returns the best responses
to all other players\textquoteright{} strategies. Hence, for any $\lambda\in[0,\,1]$
we have
\begin{align*}
	v^{i}(s)=\lambda\,\psi_{s}^{i}(z_{s}^{i},\,x_{s}^{-i},\,v^{i})+(1-\lambda)\psi_{s}^{i}(g_{s}^{i},\,x_{s}^{-i},\,v^{i})\leq\psi_{s}^{i}(u_{s}^{i},\,x_{s}^{-i},\,v^{i}).
\end{align*}
Then, we have
\begin{align*}
	& \lambda\,\psi_{s}^{i}(g_{s}^{i},\,x_{s}^{-i},\,v^{i})+(1-\lambda)\psi_{s}^{i}(z_{s}^{i},\,x_{s}^{-i},\,v^{i})\\
	=\, & \lambda\max_{\mu\in\mathcal{M}_{s}^{i}(P_{s}(z_{s}^{i},\,x_{s}^{-i}))}\left\{ \langle\mu,\,C_{s}^{i}(v^{i})\rangle-\alpha_{s}^{i}(\mu)\right\} \\
	& +(1-\lambda)\max_{\mu\in\mathcal{M}_{s}^{i}(P_{s}(g_{s}^{i},\,x_{s}^{-i}))}\left\{ \langle\mu,\,C_{s}^{i}(v^{i})\rangle-\alpha_{s}^{i}(\mu)\right\} ,
\end{align*}
based on the Fenchel-Moreau representation.
Furthermore, we have
\begin{align*}
			& \lambda\max_{\mu\in\mathcal{M}_{s}^{i}(P_{s}(z_{s}^{i},\,x_{s}^{-i}))}\left\{ \langle\mu,\,C_{s}^{i}(v^{i})\rangle-\alpha_{s}^{i}(\mu)\right\} 
			\\
			& +(1-\lambda)\max_{\mu\in\mathcal{M}_{s}^{i}(P_{s}(g_{s}^{i},\,x_{s}^{-i}))}\left\{ \langle\mu,\,C_{s}^{i}(v^{i})\rangle-\alpha_{s}^{i}(\mu)\right\} \\
			= & \max_{\mu\in\mathcal{M}_{s}^{i}(P_{s}(z_{s}^{i},\,x_{s}^{-i}))\cup\mathcal{M}_{s}^{i}(P_{s}(g_{s}^{i},\,x_{s}^{-i}))}\left\{ \langle\mu,\,C_{s}^{i}(v^{i})\rangle-\alpha_{s}^{i}(\mu)\right\} ,
\end{align*}
since $\psi_{s}^{i}(g_{s}^{i},\,x_{s}^{-i},\,v^{i})=\psi_{s}^{i}(z_{s}^{i},\,x_{s}^{-i},\,v^{i})$
holds in our setting. Thus,
\begin{align*}
	& \lambda\max_{\mu\in\mathcal{M}_{s}^{i}(P_{s}(z_{s}^{i},\,x_{s}^{-i}))}\left\{ \langle\mu,\,C_{s}^{i}(v^{i})\rangle-\alpha_{s}^{i}(\mu)\right\} 
	\\
	& +(1-\lambda)\max_{\mu\in\mathcal{M}_{s}^{i}(P_{s}(g_{s}^{i},\,x_{s}^{-i}))}\left\{ \langle\mu,\,C_{s}^{i}(v^{i})\rangle-\alpha_{s}^{i}(\mu)\right\} \\
	\geq\, & \max_{\mu\in\mathcal{M}_{s}^{i}(P_{s}((1-\lambda)\,z^{i}+\lambda\,g_{s}^{i},\,x_{s}^{-i}))}\left\{ \langle\mu,\,C_{s}^{i}(v^{i})\rangle-\alpha_{s}^{i}(\mu)\right\} ,
\end{align*}
by Lemma \ref{Lemma B 12}. Consequently,
\begin{align*}
	\psi_{s}^{i}(u_{s}^{i},\,x_{s}^{-i},\,v^{i})\geq\, &v^{i}(s)\\
	=\, & \lambda\,\psi_{s}^{i}(z_{s}^{i},\,x_{s}^{-i},\,v^{i})+(1-\lambda)\psi_{s}^{i}(g_{s}^{i},\,x_{s}^{-i},\,v^{i})\\
	\geq\, & \psi_{s}^{i}(\lambda z_{s}^{i}+(1-\lambda)g_{s}^{i},\,x_{s}^{-i},\,v^{i})\geq v^{i}(s),
\end{align*}
and hence $\lambda(z^{i})_{i\in\mathcal{I}}+(1-\lambda)(g^{i})_{i\in\mathcal{I}}\in\Phi(x)$.
\end{proof}
The next lemma completes our proof. 
\begin{lemma}
\label{Lem:Kakutani-2} $\Phi$ is an upper semicontinuous correspondence
on $\mathcal{X}$.
\end{lemma}

\begin{proof}
Suppose $x_{n}\rightarrow x$, $y_{n}\rightarrow y$, and $y_{n}\in\Phi(x_{n})$.
Taking a subsequence, we can suppose $\tau_{s}^{i}((x^{-i})_{n})\rightarrow v^{i}(s)$.
Using the triangle inequality, for any $s\in\mathcal{S}$ and $i\in\mathcal{I}$
we have
\begin{align*}
	&\left|\psi_{s}^{i}(y_{s}^{i},\,x_{s}^{-i},\,v^{i})-v^{i}(s)\right|\\
	\leq &  \left|\psi_{s}^{i}(y_{s}^{i},\,x_{s}^{-i},\,v^{i})-\psi_{s}^{i}(y_{s}^{i},\,x_{s}^{-i},\,\tau_{s}^{i}((x^{-i})_{n}))\right|+\left|\psi_{s}^{i}(y_{s}^{i},\,x_{s}^{-i},\,\tau_{s}^{i}((x^{-i})_{n}))-v^{i}(s)\right|\\
	= & \left|\psi_{s}^{i}(y_{s}^{i},\,x_{s}^{-i},\,v^{i})-\psi_{s}^{i}(y_{s}^{i},\,x_{s}^{-i},\,\tau_{s}^{i}((x^{-i})_{n}))\right|+\left|\tau_{s}^{i}((x^{-i})_{n})-v^{i}(s)\right|\rightarrow0,
\end{align*}
as $n\rightarrow\infty$. The above equality holds by definition of
a best response and therefore, $v^{i}(s)=\psi_{s}^{i}(y_{s}^{i},\,x_{s}^{-i},\,v^{i})$.
By Lemma \ref{Lem:Kakutani-1}, we also have $\tau_{s}^{i}((x^{-i})_{n})=v^{i}(s)$.
Thus, we have established
\begin{align*}
	v^{i}(s)=&\psi_{s}^{i}(y_{s}^{i},\,x_{s}^{-i},\,v^{i})
	\\
	=&\tau_{s}^{i}((x^{-i})_{n})=\min_{u_{s}^{i}\in\mathcal{P}\left(\mathcal{A}^{i}\right)}\psi_{s}^{i}(u_{s}^{i},\,x_{s}^{-i},\,v^{i}),
\end{align*}
and so $y\in\Phi(x)$, completing the proof that $\Phi$ is an upper
semicontinuous correspondence. The fact that $\Phi(x)$ is a closed
set for any $x\in\mathcal{X}$ follows from the definition of upper
semicontinuity.
\end{proof}

\subsection{Examples of Saddle-Point Risk Measures}

For our $Q$-learning algorithm, we specifically focus on risk measures
that can be estimated by solving a stochastic saddle-point problem
such as as Problem (11). The following result, based
on \cite[Theorem 3.2]{Huang2018a}, gives special conditions on $G^{i}$
for the corresponding risk function $\rho^{i}$ in Problem (11).
\begin{theorem}
\label{Theorem 5.2} Suppose there is a collection of functions $\left\{ h_{z}\right\} _{z\in\mathcal{Z}^{i}}$
such that: (i) $h_{z}$ is $P$-square summable for every $y\in\mathcal{Y}^{i}$,
$z\in\text{\ensuremath{\mathcal{Z}}}^{i}$; (ii) $y\rightarrow h_{z}\left(X-y\right)$
is convex; (iii) $z\rightarrow h_{z}\left(X-y\right)$ is concave;
and (iv) $G^{i}:\,\mathcal{L}\times\mathcal{Y}^{i}\times\mathcal{Z}^{i}\rightarrow\mathbb{R}$
is given by $G^{i}(X,\,y,\,z)=y+h_{z}(X-y)$, then the saddle-point risk
measure (Problem (11)) is a convex risk measure.
\end{theorem}

We now give some examples of functions $\left\{ h_{z}\right\} _{z\in\mathcal{Z}^{i}}$
satisfying the conditions of Theorem \ref{Theorem 5.2} such that a corresponding risk-aware Markov perfect equilibrium exists.
\begin{example}
The distance between any probability distribution and a reference
distribution may be measured by a $\phi$-divergence function, several
examples of $\phi$-divergence functions are shown in Table 4. We
can, in principle, approximate convex $\phi$-divergence functions
with piecewise linear convex functions of the form $\hat{\phi}(\mu)=\max_{j\in\mathcal{J}}\left\{ \langle d_{j},\,\mu\rangle+g_{j}\right\} $.
Using this form of $\hat{\phi}$, we may define a corresponding
set of probability distributions:
\begin{align}
	\mathcal{M}_{s}^{i}(P_{s})=\{ \mu:\,\mu=P_{s}\circ\xi_{s},\,B_{s}\,\mu=e,\,\mu\geq0,\,B_{s}\,P_{s}\circ\left(d_{j}\circ\xi_{s}+g_{j}\right)\leq\alpha^{i}\cdot e,\,\forall j\in\mathcal{J}\} ,\label{Phi divergence}
\end{align}
for constants $\alpha^{i}\in\left(0,\,1\right)$ for all $i\in\mathcal{I}$.
Based on \cite[Lemma 1]{Postek2015}, the risk measure corresponding
to (\ref{Phi divergence}) has the form
\begin{align}
			\rho(X) = \inf_{b\geq0,\,\eta\in\mathbb{R}}\left\{ \eta+b\,\alpha^{i}+b\,\mathbb{E}_{P_{s}}\left[\hat{\phi}^{\ast}\left(\frac{X-\eta}{b}\right)\right]\right\} ,\label{Counterpart-1}
\end{align}
where $\hat{\phi}^{\ast}$ is the convex conjugate of $\hat{\phi}$.

(i) Let $\phi_{z}$ denote a family of $\phi$-divergence functions
parameterized by $z\in\mathcal{Z}^{i}$ that is concave in $\mathcal{Z}^{i}$,
and let $\hat{\phi}_{z}$ and $\phi_{z}^{\ast}$ denote the corresponding
piecewise linear approximation and its convex conjugate, respectively.
Then, we may define
\begin{align*}
	\mathcal{M}_{z}:=\{ &\mu:\,\mu=P_{s}\circ\xi_{s},\,B_{s}\,\mu=e,\,\mu\geq0,\,
	\\
	& B_{s}\,P_{s}\circ\hat{\phi_{z}}(\xi_{s})\leq\alpha^{i}\cdot e\},
\end{align*}
and the risk measure corresponding to $\cup_{z\in\mathcal{Z}^{i}}\mathcal{M}_{z}$
is
\begin{align}
	\rho(X)=\inf_{b\geq0,\,\eta\in\mathbb{R}}\max_{z\in\mathcal{Z}^{i}}\Bigg\{ \eta+b\,\alpha^{i}+b\,\mathbb{E}_{P_{s}}\left[\phi_{z}^{\ast}\left(\frac{X-\eta}{b}\right)\right]\Bigg\} .\label{Counterpart}
\end{align}
Suppose we choose $h_{z}$ (from Theorem \ref{Theorem 5.2}) to be
\[
\frac{h_{z}(X-\eta)}{b}=\phi_{z}^{\ast}\left(\frac{X-\eta}{b}\right),
\]
for any $\eta\in\mathbb{R}$ and $b>0$. Assume $X$ has bounded support
$[\eta_{\textrm{min}},\,\eta_{\textrm{max}}]$, then formulation (\ref{Counterpart})
becomes
\[
\rho(X)=\min_{\eta\in[\eta_{\textrm{min}},\,\eta_{\textrm{max}}]}\max_{z\in\mathcal{Z}^{i}}\left\{ \eta+\mathbb{E}_{P_{s}}\left[h_{z}(X-\eta)\right]\right\} ,
\]
which conforms to the saddle-point structure in Problem (11).

(ii) To recover CVaR, we let $\alpha^{i}\in\left(0,\,1\right)$ for
all $i\in\mathcal{I}$ and choose the $\phi$-divergence function
\[
\phi(x)=\begin{cases}
0 & 0\leq x\leq\frac{e}{1-\alpha^{i}},\\
\infty & \textrm{otherwise},
\end{cases}
\]
and we take
\begin{align*}
	\mathcal{M}_{s}^{i}(P_{s})=\Big\{\mu:\,\mu=P_{s}\circ\xi_{s},\,B_{s}\,\mu=e,\,\mu\geq0,\,0\leq\mu\leq\frac{P_{s}}{1-\alpha^{i}}\Big\} .\label{Mp}
\end{align*}
If we take the convex conjugate of this $\phi$-divergence function
and substitute it into Eqs. (\ref{Counterpart-1}), we obtain
\begin{align*}
	\rho(X)=\min_{\eta\in[\eta_{\textrm{min}},\,\eta_{\textrm{max}}]}\{\eta+(1-\alpha^{i})^{-1}\mathbb{E}_{P_{s}}\left[\max\left\{ X-\eta,\,0\right\} \right]\} ,
\end{align*}
corresponding to $h_{z}(X-\eta)=-(1-\alpha^{i})^{-1}\max\left\{ X-\eta,\,0\right\} $
for all $z\in\mathcal{Z}^{i}$.
\end{example}

\begin{table*}
\centering{}%
\begin{tabular}{|c|c|c|}
	\hline 
	\textbf{Name}  & $\phi(t)\quad t\geq0$  & $\phi^{*}(s)$\tabularnewline
	\hline 
	\hline 
	Kullback-Leibler  & $t\log t-t+1$  & $\exp(s)-1$\tabularnewline
	\hline 
	Burg entropy  & $-\log t+t-1$  & $-\log(1-s),\quad s<1$\tabularnewline
	\hline 
	$\chi^{2}$ distance  & $\frac{1}{t}(t-1)^{2}$  & $2-2\sqrt{1-s},\quad s<1$\tabularnewline
	\hline 
	Modified $\chi^{2}$ distance  & $(t-1)^{2}$  & $\begin{cases}
	-1 & \quad s<-2\\
	s+s^{2}/4 & \quad s\geq-2
	\end{cases}$\tabularnewline
	\hline 
	Hellinger distance  & $(\sqrt{t}-1)^{2}$  & $\frac{s}{1-s},\quad s<1$\tabularnewline
	\hline 
	$\chi$- divergence  & $|t-1|^{\theta}$  & $s+(\theta+1)\left(\frac{|s|}{\theta}\right)^{\theta/(\theta-1)}$\tabularnewline
	\hline 
	Variation distance  & $|t-1|$  & $\max\left\{ -1,\,s\right\} ,\quad s\leq1$\tabularnewline
	\hline 
	Cressie-Read  & $\frac{1-\theta+\theta t-t^{\theta}}{\theta(1-\theta)},\quad t\neq0,1$  & $\frac{1}{\theta}(1-s(1-\theta))^{\theta/(1-\theta)}-\frac{1}{\theta},\quad s<\frac{1}{1-\theta}$\tabularnewline
	\hline 
\end{tabular}\caption{Examples of $\phi$-divergence functions and their convex conjugate
	functions}
\end{table*}

\subsection{RaNashQL Implementation Details}

We give further details on each step of Algorithm 1 as follows. We
will shortly require the definition
\[
\Pi_{\mathcal{Y}\times\mathcal{Z}}[(y,\,z)]:=\arg\min_{(y^{\prime},\,z^{\prime})\in\mathcal{Y}\times\mathcal{Z}}\|(y,\,z)-(y^{\prime},\,z^{\prime})\|_{2},
\]
i.e., the Euclidean projection onto $\mathcal{Y}\times\mathcal{Z}$.
\begin{itemize}
\item \textbf{Step 0}: Initialization:
\begin{itemize}
	\item \textbf{Step 0a}: Initialize all $Q$-values $Q_{1,1}^{i}(s,\,a)$
	for all $(s,\,a)\in\mathbb{K}$ and $i\in\mathcal{I}$;
	\item \textbf{Step 0b}: Initialize $\left(y_{0,t}^{i}(s,\,a),\,z_{0,t}^{i}(s,\,a)\right)$
	for all $t\leq T$, $(s,\,a)\in\mathbb{K}$, and $i\in\mathcal{I}$.
\end{itemize}
\item \textbf{Step 1}: For all $(s,\,a)\in\mathbb{K}$ and $i\in\mathcal{I}$,
set \[\left(y_{n,1}^{i}(s,\,a),\,z_{n,1}^{i}(s,\,a)\right)=\left(y_{n-1,T}^{i}(s,\,a),\,z_{n-1,T}^{i}(s,\,a)\right),\]
and $Q_{n,1}^{i}(s,\,a)=Q_{n-1,T}^{i}(s,\,a)$. All agents observe
the current state $s_{t}^{n}$:
\begin{itemize}
	\item \textbf{Step 1a}: Generate an action $a_{n}^{i}$ from policy $\pi$
	(which gives some positive probability to all actions);
	\item \textbf{Step 1b}: Observe actions $a_{n}=(a_{n}^{i})_{i\in\mathcal{I}}$,
	costs $\left\{ c^{i}(s_{t}^{n},\,a_{n})\right\} _{i\in\mathcal{I}}$,
	and next state $s_{t+1}^{n}\sim P\left(\cdot\,\vert\,s_{t}^{n},\,a_{n}\right)$.
\end{itemize}
\item \textbf{Step 2}: Compute Nash $Q$-values $v_{n-1}^{i}(s_{t+1}^{n})=Nash^{i}\left(Q_{n-1,T}^{j}(s_{t+1}^{n})\right)_{j\in\mathcal{I}}$
for all $i\in\mathcal{I}$.Compute
	\begin{align*}
		\hat{q}_{n,t}^{i}(s_{t}^{n},\,a^{n})=G^{i}(c^{i}(s_{t}^{n},\,a_{n})+\gamma\,v_{n-1}^{i}(s_{t+1}^{n}),\,y_{n,t}^{i}(s_{t}^{n},\,a_{n}),\,z_{n,t}^{i}(s_{t}^{n},\,a_{n})),
	\end{align*}
	and
	\begin{align}
		\left(y_{n,t}^{i}(s_{t}^{n},\,a_{n}),\,z_{n,t}^{i}(s_{t}^{n},\,a_{n})\right)=\frac{1}{t-\tau_{\ast}(t)+1}\sum_{\tau=\tau_{\ast}(t)}^{t}\left(y_{n,\tau}^{i}(s_{t}^{n},\,a_{n}),\,z_{n.\tau}^{i}(s_{t}^{n},\,a_{n})\right),\label{Q-learning_saddle}
	\end{align}
	for all $i\in\mathcal{I}$. This step observes a new state and computes
	the estimated $Q$-value $\hat{q}_{n,t}^{i}$;

\item \textbf{Step 3}: For all $(s,\,a)\in\mathbb{K}$, and $i\in\mathcal{I}$,
compute
\begin{align}
	Q_{n,t}^{i}(s,\,a)=\left(1-\text{\ensuremath{\theta}}_{\beta}^{n}(s,\,a)\right)Q_{n-1,T}^{i}(s,\,a)+\theta_{\beta}^{n}(s,\,a)\,\hat{q}_{n,t}^{i}(s_{t}^{n},\,a^{n}).\label{Q-value}
\end{align}
This update is the same as in standard $Q$-learning w.r.t. the outer
loop.
\item \textbf{Step 4}: Update
\begin{align}
			& \left(y_{n,t+1}^{i}(s_{t}^{n},\,a^{n}),\,z_{n,t+1}^{i}(s_{t}^{n},\,a^{n})\right)\nonumber \\
			= & \Pi_{\mathcal{Y}^{i}\times\mathcal{Z}^{i}}\left\{ \left(y_{n,t}^{i}(s_{t}^{n},\,a^{n}),\,z_{n,t}^{i}(s_{t}^{n},\,a^{n})\right)\right. \nonumber
			\\ 
			& -\left.\lambda_{t,\alpha}\psi\left(c^{i}(s_{t}^{n},\,a^{n})+\gamma\,v_{n-1}^{i}(s_{t+1}^{n}),\,y_{n,t}^{i}(s_{t}^{n},\,a^{n}),\,z_{n,t}^{i}(s_{t}^{n},\,a^{n})\right)\right\},\label{SASP}
\end{align}
for all $i\in\mathcal{I}$, and
\begin{align}
			& \psi^{i}\left(v_{n-1}^{i}(s_{t+1}^{n}),\,y_{n,t}^{i}(s_{t}^{n},\,a^{n}),\,z_{n,t}^{i}(s_{t}^{n},\,a^{n})\right)\nonumber \\
			= & \left(\begin{array}{c}
				H_{\mathcal{Y}}G_{y}^{i}(c^{i}(s_{t}^{n},\,a^{n})+
				\\
				\gamma\,v_{n-1}^{i}(s_{t+1}^{n}),\,y_{n,t}^{i}(s_{t}^{n},\,a^{n}),\,z_{n,t}^{i}(s_{t}^{n},\,a^{n}))\\
				\\
				-H_{\mathcal{Z}}G_{z}^{i}(c^{i}(s_{t}^{n},\,a^{n})+
				\\
				\gamma\,v_{n-1}^{i}(s_{t+1}^{n}),\,y_{n,t}^{i}(s_{t}^{n},\,a^{n}),\,z_{n,t}^{i}(s_{t}^{n},\,a^{n}))
			\end{array}\right).\label{SASP-1}
\end{align}
This is the risk estimation step, it updates the current iterate of
the risk corresponding to each selected state-action pair.
\end{itemize}

\subsection{Assumptions for RaNashQL}

We now list the necessary definitions and assumptions for our algorithm,
most of which are standard in the stochastic approximation literature.
We first define a probability space $(\Omega,\,\mathcal{G},\,P)$
where
\[
\mathcal{G}=\sigma\left\{ (s_{t}^{n},\,a^{n}),\,n\leq N,\,t\leq T\right\} ,
\]
is the $\sigma$-algebra for the history state-action pairs up to iteration $T$ and $N$, and the filtration is
\[
\mathcal{G}_{t}^{n}=\left\{ \sigma\left\{ (s_{\tau}^{i},\,a_{\tau}^{i}),\,i<n,\,\tau\leq T\right\} \cup\left\{ (s_{\tau}^{n},\,a_{\tau}^{n}),\,\tau\leq t\right\} \right\} ,
\]
for $t\leq T$ and $n\leq N$, with $\mathcal{G}_{t}^{0}=\left\{ \varnothing,\,\Omega\right\} $
for all $t\leq T$. This filtration is nested, $\mathcal{G}_{t}^{n}\subseteq\mathcal{G}_{t+1}^{n}$
for $1\leq t\leq T-1$ and $\mathcal{G}_{T}^{n}\subseteq\mathcal{G}_{0}^{n+1}$.
The following assumption reflects our exploration requirement.
\begin{assumption}
\label{Assu:epsilon_greedy} ($\varepsilon$-greedy policy) There
is an $\varepsilon>0$ such that the policy $\pi$ satisfies, for
any $n\leq N,\,t\leq T$, and all $(s,\,a)\in\mathcal{K}$, $\mathbb{P}\left((s_{t}^{n},\,a^{n})=(s,\,a)\,|\,\mathcal{G}_{t-1}^{n}\right)\geq\varepsilon$
and $\mathbb{P}\left((s_{1}^{n},\,a^{n})=(s,\,a)\,|\,\mathcal{G}_{T}^{n-1}\right)\geq\varepsilon$.
\end{assumption}

In particular, let $x_{s}\in\mathcal{X}_{s}$ be a Nash equilibrium
of the stage game $(Q^{i}(s))_{i\in\mathcal{I}}$. Then, with probability
$\varepsilon\in\left(0,\,1\right)$, the action $a^{i}$ is chosen
randomly from $\mathcal{A}^{i}$, and with probability $1-\varepsilon$,
the action $a^{i}$ is drawn from $\mathcal{A}^{i}$ according to
$x_{s}^{i}$. Assumption \ref{Assu:epsilon_greedy} guarantees, by
the Extended Borel-Cantelli Lemma in \cite{Breiman1992}, that we
will visit every state-action pair infinitely often with probability
one.

The next assumption contains
our requirements on the step-sizes for the $Q$-value update.
\begin{assumption}
\label{Assu:Q_step-size} For all $(s,\,a)\in\mathbb{K}$ and for
all $n\leq N,\,t\leq T$, the step-sizes for the $Q$-value update
satisfy: $\sum_{n=1}^{\infty}\theta_{\beta}^{n}(s,\,a)=\infty$ and
$\sum_{n=1}^{\infty}\theta_{\beta}^{n}(s,\,a)^{2}<\infty$ for all
$t\leq T$ and $(s,\,a)\in\mathbb{K}$ a.s. Let $\#(s,\,a,\,n)$ denote
one plus the number of times, until the beginning of iteration $n$,
that the state-action pair $(s,\,a)$ has been visited, and let $N^{s,a}$
denote the set of outer iterations where action $a$ was performed
in state $s$. The step-sizes $\theta_{\beta}^{n}(s,\,a)$ satisfy
$\theta_{\beta}^{n}(s,\,a)=\frac{1}{[\#(s,a,n)]^{\beta}}$ if $n\in N^{s,a}$
and $\theta_{\beta}^{n}(s,\,a)=0$ otherwise.
\end{assumption}

Assumption \ref{Assu:Q_step-size} reflects the asynchronous nature
of the $Q$-learning algorithm as stated in \cite{Even-Dar2004},
only a single state-action pair is updated when it is observed in
each iteration, which can be implemented when there is no initial knowledge on the state space. 

\subsection{Proof of Theorem 2}

In this section, we develop the proof of almost sure convergence of
RaNashQL (Theorem 2). This proof uses techniques
from the stochastic approximation literature \cite{Kushner2003,Borkar2000,Borkar2008,Huang2017,Huang2018a},
and is based on the following steps:
\begin{enumerate}
\item Show that the stage game Nash equilibrium operator is non-expansive.
\item Bound the saddle-point estimation error in terms of the error between estimated $Q$-value to the optimal $Q$-value. 
\item Apply the classical stochastic approximation convergence theorem.
\end{enumerate}

\subsubsection*{Preliminaries}

We first present preliminary definitions and properties that will
be used in the proof of Theorem 2. For any $(s,\,a)\in\mathcal{K}$,
we define $(y_{n,\ast}^{i}(s,\,a),\,z_{n,\ast}^{i}(s,\,a))$ to be
a saddle point of 
\begin{align*}
(y(s,\,a),\,z(s,\,a))\rightarrow \mathbb{E}_{P_{s}(a,\,s')}\left[G^{i}\left(c^{i}(s,\,a)+\gamma\,v_{n-1}^{i}(s^{\prime}),\,y(s,\,a),\,z(s,\,a)\right)\right],
\end{align*}
for each $(s,\,a)\in\mathcal{K}$, where \[v_{n-1}^{i}(s^{\prime})=Nash^{i}\left(Q_{\ast}^{j}(s^{\prime})\right)_{j\in\mathcal{J}}.\]
Similarly, we define $(y_{\ast}^{i}(s,\,a),\,z_{\ast}^{i}(s,\,a))$
to be a saddle point of
\begin{align*}
&(y(s,\,a),\,z(s,\,a))
\\
\rightarrow&\mathbb{E}_{P_{s}(a,\,s')}\left[G^{i}\left(c^{i}(s,\,a)+\gamma\,v_{\ast}^{i}(s^{\prime}),\,y(s,\,a),\,z(s,\,a)\right)\right],
\end{align*}
for each $(s,\,a)\in\mathcal{K}$, where \[v_{\ast}^{i}(s^{\prime})=Nash^{i}\left(Q_{\ast}^{j}(s^{\prime})\right)_{j\in\mathcal{J}}.\]
Let
\begin{align*}
\mathcal{S}_{n,t}^{i}:=\{(\partial G_{y}^{i}(c^{i}+\gamma\,v_{n-1}^{i},\,y_{n,t}^{i},\,z_{n,t}^{i}),\,\partial G_{z}^{i}(c^{i}+\gamma\,v_{n-1}^{i},\,y_{n,t}^{i},\,z_{n,t}^{i}))\},
\end{align*}
and
\begin{align*}
\overline{\mathcal{S}}_{n,t}^{i}:=\{(\partial G_{y}^{i}(c^{i}+\gamma\,v_{\ast},\,y_{n,t}^{i},\,z_{n,t}^{i}),\,\partial G_{z}^{i}(c^{i}+\gamma\,v_{\ast},\,y_{n,t}^{i},\,z_{n,t}^{i}))\},
\end{align*}
be the subdifferentials of $G^{i}$ with respect to $v_{n-1}^{i}$
and $v_{\ast}^{i}$, given $(y_{n,t}^{i},\,z_{n,t}^{i})$. The following
Lemma \ref{Lemma 5.5} bounds $d_{\mathcal{H}}(\mathcal{S}_{n,t}^{i},\,\overline{\mathcal{S}}_{n,t}^{i})$
in terms of $\|\cdot\|_{2}$. This result follows from the convergence
of subdifferentials of convex functions, the Lipschitz continuity
of $G^{i}$, and non-expansiveness of the Nash equilibrium mapping
for stage games. 
\begin{lemma}
\label{Lemma 5.5}\cite[Theorem 4.1]{Penot1993} Under the Lipschitz
continuity of function $G^{i}$, there exist constants $K_{\psi}^{(1)},\,K_{\psi}^{(2)}>0$,
such that
\begin{align}
	d_{\mathcal{H}}(\mathcal{S}_{n,t}^{i},\,\overline{\mathcal{S}}_{n,t}^{i})\leq K_{\psi}^{(1)}\|Q_{n-1,T}^{i}-Q_{\ast}^{i}\|_{2}+K_{\psi}^{(2)}\sqrt{\|Q_{n-1,T}^{i}-Q_{\ast}^{i}\|_{2}}.\label{HS modulus}
\end{align}
\end{lemma}

We conclude this preliminary subsection by showing that all mixed
points of a stage game have equal value.
\begin{lemma}
\label{Lemma 5.11} Let $x$ and $y$ be $\mathcal{I}^{\prime}-$mixed
points of $(C^{i})_{i\in\mathcal{I}}$, then $C^{i}(x)=C^{i}(y)$
for all $i\in\mathcal{I}$. 
\end{lemma}

\begin{proof}
Suppose $x$ is a $\mathcal{I}_{1}$-mixed point and $y$ is a $\mathcal{I}_{2}$-mixed
point. For $i\in\mathcal{I}_{1}\cup(\mathcal{I}/\mathcal{I}_{2})$,
we have $C^{i}(x)\leq C^{i}(y)$ and
\[
C^{i}(x)\geq C^{i}(y_{s}^{i},\,x_{s}^{-i})\geq C^{i}(y).
\]
The only consistent solution is $C^{i}(x)=C^{i}(y)$. Similarly, for
$i\in\mathcal{I}_{2}\cup(\mathcal{I}/\mathcal{I}_{1})$, we have the
similar argument. For $i\in\mathcal{I}_{1}\cap\mathcal{I}_{2}$ and
$i\in(\mathcal{I}/\mathcal{I}_{1})\cap(\mathcal{I}/\mathcal{I}_{2})$,
we have $C^{i}(x)=C^{i}(y)$ by the definitions of global optimal
point and saddle point.
\end{proof}

\subsubsection*{Step 1: Show that the stage game Nash equilibrium operator is non-expansive}

The following Lemma \ref{Lemma 6.1} shows that the Nash equilibrium
mapping is non-expansive. We use the following norm
\[
\|Q^{i}(s)-\tilde{Q}^{i}(s)\|_{\infty}:=\max_{a\in\mathcal{A}}|Q^{i}(s,\,a)-\tilde{Q}^{i}(s,\,a)|
\]
for $Q$-values in all state $s\in\mathcal{S}$.
\begin{lemma}
\label{Lemma 6.1} For $Q$ and $\tilde{Q}$, let $x$ and $\tilde{x}$
be Nash equilibria for $Q$ and $\tilde{Q}$, respectively. Then, for all $s\in\mathcal{S}$, 
\[
|\left[Q^{i}(s)\right](x)-\left[\tilde{Q}^{i}(s)\right](\tilde{x})|\leq\|Q^{i}(s)-\tilde{Q}^{i}(s)\|_{\infty}.
\]
\end{lemma}
\begin{proof}
Let $v^{i}(s)=\left[Q^{i}(s)\right](x)$ and $\tilde{v}^{i}(s)=\left[\tilde{Q}^{i}(s)\right](\tilde{x})$.
By Assumption 3, there are three possible cases
to consider:

\textbf{Case 1}: Both $x$ and $\tilde{x}$
are global optimal points. If $\left[Q^{i}(s)\right](x)\geq\left[\tilde{Q}^{i}(s)\right](\tilde{x})$,
we have
\begin{align*}
	&v^{i}(s)-\tilde{v}^{i}(s)
	\\
	=\, & \left[Q^{i}(s)\right](x)-\left[\tilde{Q}^{i}(s)\right](\tilde{x})\\
	\leq\, & \left[Q^{i}(s)\right](\tilde{x})-\left[\tilde{Q}^{i}(s)\right](\tilde{x})\\
	=\, & \sum_{a\in\mathcal{A}}\left(\prod_{i\in\mathcal{I}}\tilde{x}^{i}(a^{i})\right)\left(Q^{i}(s,\,a)-\tilde{Q}^{i}(s,\,a)\right)\\
	\leq\, & \sum_{a\in\mathcal{A}}\left(\prod_{i\in\mathcal{I}}\tilde{x}^{i}(a^{i})\right)\|Q^{i}(s)-\tilde{Q}^{i}(s)\|_{\infty}\\
	=\, & \|Q^{i}(s)-\tilde{Q}^{i}(s)\|_{\infty}.
\end{align*}
If $\left[Q^{i}(s)\right](x)\leq\left[\tilde{Q}^{i}(s)\right](\tilde{x})$,
the proof follows similarly. 

\textbf{Case 2}: Both $x$ and $\tilde{x}$
are saddle points. If $\left[Q^{i}(s)\right](x)\geq\left[\tilde{Q}^{i}(s)\right](\tilde{x})$,
we have
\begin{align}
	& v^{i}(s)-\tilde{v}^{i}(s) \nonumber
	\\
	=\, & \left[Q^{i}(s)\right](x)-\left[\tilde{Q}^{i}(s)\right](\tilde{x})\nonumber \\
	\leq\, & \left[Q^{i}(s)\right](x)-\left[\tilde{Q}^{i}(s)\right](x^{i},\,\tilde{x}^{-i})\label{eq:1}\\
	\leq\, & \left[Q^{i}(s)\right](x^{i},\,\tilde{x}^{-i})-\left[\tilde{Q}^{i}(s)\right](x^{i},\,\tilde{x}^{-i})\label{eq:2}\\
	\leq\, & \|Q^{i}(s)-\tilde{Q}^{i}(s)\|_{\infty},\nonumber 
\end{align}
where the first inequality (\ref{eq:1}) is by definition of Nash
equilibrium, and inequality (\ref{eq:2}) is from Assumption 3 (ii) 

\textbf{Case 3}: Both $x$ and $\tilde{x}$
are $\mathcal{I}_{1}$- and $\mathcal{I}_{2}$-mixed points, respectively.
Then, for $i\in\mathcal{I}_{1}\cup(\mathcal{I}/\mathcal{I}_{2})$,
and $\left[Q^{i}(s)\right](x)\geq\left[\tilde{Q}^{i}(s)\right](\tilde{x})$,
and by the arguments from \textbf{Cases 1} and\textbf{ 2}, we know
that $v^{i}(s)-\tilde{v}^{i}(s)\leq\|Q^{i}(s)-\tilde{Q}^{i}(s)\|_{\infty}$.
Alternatively, for $i\in\mathcal{I}_{2}\cup(\mathcal{I}/\mathcal{I}_{1})$
and $\left[Q^{i}(s)\right](x)\leq\left[\tilde{Q}^{i}(s)\right](\tilde{x})$,
we have $\tilde{v}^{i}(s)-v^{i}(s)\leq\|Q^{i}(s)-\tilde{Q}^{i}(s)\|_{\infty}$. 
\end{proof}

\subsubsection*{Step 2: Bound the saddle-point estimation error in terms of the $Q$-value
error}

In this step we will bound $\|(y_{n,t}^{i},\,z_{n,t}^{i})-(y_{n,\ast}^{i},\,z_{n,\ast}^{i})\|_{2}^{2}$
by a function of $\|Q_{n-1,T}^{i}-Q_{\ast}^{i}\|_{2}^{2}$ and then
we will bound $\|(y_{\ast}^{i},\,z_{\ast}^{i})-(y_{n,\ast}^{i},\,z_{n,\ast}^{i})\|_{2}$
by a function of $\|Q_{n-1,T}^{i}-Q_{*}^{i}\|_{2}$. Our intent is
to establish a relationship between the risk estimation error (which
depends on the saddle points of the risk measure) and the error between estimated $Q$-value and the optimal one. 
\begin{lemma}
\label{Lemma 5.5-1}\cite[Lemma 5.3]{Huang2018a} Suppose Assumptions
\ref{Assu:epsilon_greedy} and \ref{Assu:Q_step-size} hold, then
there exists a constant $0<\kappa<1/C\,K_{\psi}^{(1)}$ such that
\begin{align}
	\|(y_{n,\,t}^{i},\,z_{n,\,t}^{i})-(y_{n,\ast}^{i},\,z_{n,\ast}^{i})\|_{2}^{2}\leq\frac{C(\tau_{\ast}(t))^{-\alpha}}{\kappa(1-C(\tau_{\ast}(t))^{-\alpha}K_{\psi}^{(1)}\kappa)}\|Q_{n-1,T}^{i}-Q_{\ast}^{i}\|_{2}^{2},\label{Risk bound}
\end{align}
for all $t\leq T$ and $n\leq N$. 
\end{lemma}

\begin{proof}
As a consequence of Eq. (\ref{SASP}) in Step 4 of Algorithm 1, we
have
\begin{align}
	& \|(y_{n,\,t+1}^{i},\,z_{n,\,t+1}^{i})-(y_{n,\ast}^{i},\,z_{n,\ast}^{i})\|_{2}^{2}\nonumber \\
	=\, & \Biggl\Vert\prod_{\mathcal{Y}^{i}\times\mathcal{Z}^{i}}\left((y_{n,\,t}^{i},\,z_{n,\,t}^{i})-\lambda_{t,\alpha}\psi(v_{n-1}^{i},\,y_{n,\,t}^{i},\,z_{n,\,t}^{i})\right)-\prod_{\mathcal{Y}^{i}\times\mathcal{Z}^{i}}(y_{n,\ast}^{i},\,z_{n,\ast}^{i})\Biggr\Vert_{2}^{2}\nonumber \\
	\leq\, & \|(y_{n,\,t}^{i},\,z_{n,\,t}^{i})-(y_{n,\ast}^{i},\,z_{n,\ast}^{i})-\lambda_{t,\alpha}\psi^{i}(c^{i}+\gamma\,v_{n-1}^{i},\,y_{n,\,t}^{i},\,z_{n,\,t}^{i})\|_{2}^{2}\nonumber \\
	\leq\, & \|(y_{n,\,t}^{i},\,z_{n,\,t}^{i})-(y_{n,\ast}^{i},\,z_{n,\ast}^{i})\|_{2}^{2}+(H_{\mathcal{Y}}^{2}+H_{\mathcal{Z}}^{2})L^{2}C^{2}t^{-2\alpha}\nonumber \\
	& -2\left((y_{n,\,t}^{i},\,z_{n,\,t}^{i})-(y_{n,\ast}^{i},\,z_{n,\ast}^{i})\right)^{\top}Ct^{-\alpha}\psi^{i}(c^{i}+\gamma\,v_{n-1}^{i},\,y_{n,\,t}^{i},\,z_{n,\,t}^{i}),\label{projection operator-1}
\end{align}
where the first inequality holds by non-expansiveness of the projection
operator and the second inequality holds since the subgradients of
$G^{i}$ are bounded based on Assumption 2. Based
on Lemma \ref{Lemma 5.5}, we have
\begin{align*}
	& \|\psi^{i}(c^{i}+\gamma\,v_{n-1}^{i},\,y_{n,\,t}^{i},\,z_{n,\,t}^{i})-\psi^{i}(c^{i}+\gamma\,v_{\ast}^{i},\,y_{n,\,t}^{i},\,z_{n,\,t}^{i})\|_{2}\\
	\leq\, & d_{\mathcal{H}}(\mathcal{S}_{n,t}^{i},\,\overline{\mathcal{S}}_{n,t}^{i})\leq K_{\psi}^{(1)}\|Q_{n-1,T}^{i}-Q_{\ast}^{i}\|_{2} +K_{\psi}^{(2)}\sqrt{\|Q_{n-1,T}^{i}-Q_{\ast}^{i}\|_{2}}.
\end{align*}
Sum the terms $\left((y_{n,\,t}^{i},\,z_{n,\,t}^{i})-(y_{n,\ast}^{i},\,z_{n,\ast}^{i})\right)^{\top}Ct^{-\alpha}\cdot\psi^{i}(c^{i}+\gamma\,v_{n-1}^{i},\,y_{n,\,t}^{i},\,z_{n,\,t}^{i})$
from $\tau_{\ast}(t)$ to $t$, then divide by $\frac{1}{t-\tau_{\ast}(t)+1}$
to obtain:
\begin{align*}
	& \frac{1}{t-\tau_{\ast}(t)+1}\sum_{\tau=\tau_{\ast}(t)}^{t}\left((y_{n,\tau}^{i},\,z_{n,\tau}^{i})-(y_{n,\ast}^{i},\,z_{n,\ast}^{i})\right)^{\top}C\tau^{-\alpha}\left(\psi^{i}(c^{i}+\gamma\,v_{n-1}^{i},\,y_{n,\,t}^{i},\,z_{n,\,t}^{i})\right)\\
	\text{\ensuremath{\leq}}\, & \left((y_{n,\,t}^{i},\,z_{n,\,t}^{i})-(y_{n,\ast}^{i},\,z_{n,\ast}^{i})\right)^{\top}C(\tau_{\ast}(t))^{-\alpha}\left(\psi^{i}(c^{i}+\gamma\,v_{n-1}^{i},\,y_{n,\,t}^{i},\,z_{n,\,t}^{i})-\psi^{i}(c^{i}+\gamma\,v_{\ast}^{i},\,y_{n,\,t}^{i},\,z_{n,\,t}^{i})\right)\\
	\leq\, & \|(y_{n,\,t}^{i},\,z_{n,\,t}^{i})-(y_{n,\ast}^{i},\,z_{n,\ast}^{i})\|_{2}C(\tau_{\ast}(t))^{-\alpha}\|\psi^{i}(c^{i}+\gamma\,v_{n-1}^{i},\,y_{n,\,t}^{i},\,z_{n,\,t}^{i})-\psi^{i}(c^{i}+\gamma\,v_{\ast}^{i},\,y_{n,\,t}^{i},\,z_{n,\,t}^{i})\|_{2}\\
	\leq\, & C(\tau_{\ast}(t))^{-\alpha}\|(y_{n,\,t}^{i},\,z_{n,\,t}^{i})-(y_{n,\ast}^{i},\,z_{n,\ast}^{i})\|_{2}\left(K_{\psi}^{(1)}\|Q_{n-1,T}^{i}-Q_{\ast}^{i}\|_{2}+K_{\psi}^{(2)}\sqrt{\|Q_{n-1,T}^{i}-Q_{\ast}^{i}\|_{2}}\right),
\end{align*}
where the first inequality is by convexity of $G^{i}$ in $y$ and
concavity of $G^{i}$ in $z$. Using the standard inequality $2ab\leq a^{2}\kappa+b^{2}/\kappa$
for all $\kappa>0$, we see that
\begin{align}
	& -2\left((y_{n,\,t}^{i},\,z_{n,\,t}^{i})-(y_{n,\ast}^{i},\,z_{n,\ast}^{i})\right)^{\top}C(\tau_{\ast}(t))^{-\alpha}\left(\psi^{i}(c^{i}+\gamma\,v_{n-1}^{i},\,y_{n,\,t}^{i},\,z_{n,\,t}^{i})-\psi^{i}(c^{i}+\gamma\,v_{\ast}^{i},\,y_{n,\,t}^{i},\,z_{n,\,t}^{i})\right)\nonumber \\
	\geq\, & -C(\tau_{\ast}(t))^{-\alpha}K_{\psi}^{(1)}\|(y_{n,\,t}^{i},\,z_{n,\,t}^{i})-(y_{n,\ast}^{i},\,z_{n,\ast}^{i})\|_{2}^{2}\,\kappa \nonumber
	\\
	& -C(\tau_{\ast}(t))^{-\alpha}\|Q_{n-1,T}^{i}-Q_{\ast}^{i}\|_{2}^{2}/\kappa\nonumber \\
	& -C(\tau_{\ast}(t))^{-\alpha}K_{\psi}^{(2)}\|(y_{n,\,t}^{i},\,z_{n,\,t}^{i})-(y_{n,\ast}^{i},\,z_{n,\ast}^{i})\|_{2}\sqrt{\|Q_{n-1,T}^{i}-Q_{\ast}^{i}\|_{2}}.\label{Deduction}
\end{align}
By summing the right hand side of inequality (\ref{projection operator-1}) from
$\tau_{\ast}(t)$ to $t$, dividing by $\frac{1}{t-\tau_{\ast}(t)+1}$,
and combining with inequality (\ref{Deduction}) we obtain
\begin{align}
	& \frac{1}{t-\tau_{\ast}(t)+1}\sum_{\tau=\tau_{\ast}(t)}^{t}(\|(y_{n,\tau}^{i},\,z_{n,\tau}^{i})-(y_{n,\ast}^{i},\,z_{n,\ast}^{i})\|_{2}^{2}+(H_{\mathcal{Y}}^{2}+H_{\mathcal{Z}}^{2})L^{2}C^{2}\tau^{-2\alpha})\nonumber \\
	& -2\left((y_{n,\,t}^{i},\,z_{n,\,t}^{i})-(y_{n,\ast}^{i},\,z_{n,\ast}^{i})\right)^{\top}C(\tau_{\ast}(t))^{-\alpha}\times\psi^{i}(c^{i}+\gamma\,v_{n-1}^{i},\,y_{n,\,t}^{i},\,z_{n,\,t}^{i})\nonumber \\
	\geq\, & \|(y_{n,\,t}^{i},\,z_{n,\,t}^{i})-(y_{n,\ast}^{i},\,z_{n,\ast}^{i})\|_{2}^{2}-2\left((y_{n,\,t}^{i},\,z^{i,n,t})-(y_{n,\ast}^{i},\,z_{n,\ast}^{i})\right)^{\top}C(\tau_{\ast}(t))^{-\alpha}\times \psi^{i}(c^{i}+\gamma\,v_{n-1}^{i},\,y_{n,\,t}^{i},\,z_{n,\,t}^{i})\nonumber \\
	\geq\, & (1-C(\tau_{\ast}(t))^{-\alpha}K_{\psi}^{(1)}\kappa)\|(y_{n,\,t}^{i},\,z_{n,\,t}^{i})-(y_{n,\ast}^{i},\,z_{n,\ast}^{i})\|_{2}^{2} -C(\tau_{\ast}(t))^{-\alpha}\|Q_{n-1,T}^{i}-Q_{\ast}^{i}\|_{2}^{2}/\kappa\nonumber \\
	& -C(\tau_{\ast}(t))^{-\alpha}K_{\psi}^{(2)}\|(y_{n,\,t}^{i},\,z_{n,\,t}^{i})-(y_{n,\ast}^{i},\,z_{n,\ast}^{i})\|_{2}\sqrt{\|Q_{n-1,T}^{i}-Q_{\ast}^{i}\|_{2}}.\label{Deduction 2}
\end{align}
We further claim that we can choose $\kappa$ satisfying $0<\kappa<1/C\,K_{\psi}^{(1)}$
such that
\begin{align}
	&(1-C(\tau_{\ast}(t))^{-\alpha}K_{\psi}^{(1)}\kappa)\|(y_{n,\,t}^{i},\,z_{n,\,t}^{i})-(y_{n,\ast}^{i},\,z_{n,\ast}^{i})\|_{2}^{2}-C(\tau_{\ast}(t))^{-\alpha}\|Q_{n-1,T}^{i}-Q_{\ast}^{i}\|_{2}^{2}/\kappa\leq 0,\label{Negative}
\end{align}
since the right hand side of inequality (\ref{Negative}) will go to infinity when $\kappa$ approaches zero, while the left hand side is bounded by a constant. Then we achieve the desired result.
\end{proof}
The following lemma bounds the difference between $(y_{\ast}^{i},\,z_{\ast}^{i})$
and $(y_{n,\ast}^{i},\,z_{n,\ast}^{i})$.
\begin{lemma}
\label{Saddle convergence} \cite[Lemma 5.5]{Huang2018a} Under the
Lipschitz continuity of $G^{i}$ and Lemma \ref{Lemma 6.1}, there
exists $K_{S}>0$ such that for all $n\leq N$ we have
\begin{align}
	\|(y_{\ast}^{i},\,z_{\ast}^{i})-(y_{n,\ast}^{i},\,z_{n,\ast}^{i})\|_{2}\leq  K_{S}\,K_{G}\|Q_{n-1,T}^{i}-Q_{\ast}^{i}\|_{2}.\label{Relationship}
\end{align}
\end{lemma}

\begin{proof}
It can be shown that
\begin{align*}
			& \|(y_{\ast}^{i},\,z_{\ast}^{i})-(y_{n,\ast}^{i},\,z_{n,\ast}^{i})\|_{2}
			\\
			\leq\, & K_{S}\left\Vert \mathbb{E}_{s^{\prime}\sim P(\cdot|s,a)}\left[G^{i}\left(c^{i}+\gamma\,v_{n-1}^{i}(s^{\prime}),\,y_{\ast}^{i}(s,\,a),\,z_{\ast}^{i}(s,\,a)\right)\right]\right.\\
			& -\left.\mathbb{E}_{s^{\prime}\sim P(\cdot|s,a)}\left[G^{i}\left(c^{i}+\gamma\,v_{n-1}^{i}(s^{\prime}),\,y_{n,\ast}^{i}(s,\,a),\,z_{n,\ast}^{i}(s,\,a)\right)\right]\right\Vert _{2}\\
			\leq\, & K_{S}\max_{z\in\mathcal{Z}^{i}}\left\Vert \min_{y\in\mathcal{Y}^{i}}\mathbb{E}_{s^{\prime}\sim P(\cdot|s,a)}\left[G^{i}\left(c^{i}+\gamma\,v_{n-1}^{i}(s^{\prime}),\,y(s,\,a),\,z(s,\,a)\right)\right]\right.\\
			& -\left.\min_{y\in\mathcal{Y}^{i}}\mathbb{E}_{s^{\prime}\sim P(\cdot|s,a)}\left[G^{i}\left(c^{i}+\gamma\,v_{\ast}^{i}(s^{\prime}),\,y(s,\,a),\,z(s,\,a)\right)\right]\right\Vert _{2}\\
			\leq\, & K_{S}\max_{y\in\mathcal{Y}^{i},\,z\in\mathcal{Z}^{i}}\left\Vert \mathbb{E}_{s^{\prime}\sim P(\cdot|s,a)}\left[G^{i}\left(c^{i}+\gamma\,v_{n-1}^{i}(s^{\prime}),\,y(s,\,a),\,z(s,\,a)\right)\right]\right.\\
			& -\left.\mathbb{E}_{s^{\prime}\sim P(\cdot|s,a)}\left[G^{i}\left(c^{i}+\gamma\,v_{\ast}^{i}(s^{\prime}),\,y(s,\,a),\,z(s,\,a)\right)\right]\right\Vert _{2}\\
			\leq\, & K_{S}\,K_{G}\|v_{n,t-1}^{i}(s)-v_{\ast}^{i}(s)\|_{2}\leq K_{S}\,K_{G}\|Q_{n-1,T}^{i}-Q_{\ast}^{i}\|_{2},
\end{align*}
where the first inequality follows from \cite[Theorem 3.1]{Terazono2015}
and \cite[Proposition 3.1]{Levy2000} (results on the stability of
optimal solutions of stochastic optimization problems), the second and fourth inequalities
are due to Lemma \ref{Lemma 6.1}, and the third inequality is by
Lipschitz continuity of $G^{i}$.
\end{proof}

\subsubsection*{Step 3: Apply the classical stochastic approximation convergence
theorem}

This step completes the proof of Theorem 2 by applying the stochastic
approximation convergence theorem (as in \cite{Borkar2008}). We first
introduce a functional operator $H^{i}:\,\mathcal{V}\times\mathcal{Y}\times\mathcal{Z}\rightarrow\mathcal{V}$
for each player $i\in\mathcal{I}$, defined for all $(s,\,a)\in\mathcal{K},$
\begin{align}
\left[H^{i}(v,\,y,\,z)\right](s,\,a):=G^{i}\left(c^{i}(s,\,a)+\gamma\,v(s^{\prime}),y(s,\,a),\,z(s,\,a)\right),\label{Operator H}
\end{align}
where $s^{\prime}\sim P(\cdot\,\vert\,s,a)$. 

Eq. (13) can then be written as, $\forall(s,\,a)\in\mathcal{K}$,
\begin{align}
Q_{\ast}^{i}(s,\,a)=\mathbb{E}_{s^{\prime}\sim P(\cdot|s,a)}\left[H^{i}(v_{\ast}^{i},\,y_{\ast}^{i},\,z_{\ast}^{i})\right](s,\,a). \label{stable}
\end{align}
Next, for all $(s,\,a)\in\mathcal{K}$, we define two stochastic processes:
\begin{align}
\epsilon_{n,t}^{i}(s,\,a) & :=\left[H^{i}(v_{n-1}^{i},\,y_{n,\ast}^{i},\,z_{n,\ast}^{i})\right](s,\,a) -\left[H^{i}(v_{n-1}^{i},\,y_{n,t}^{i},\,z_{n,t}^{i})\right](s,\,a),\label{err 1}\\
\xi_{n,t}^{i}(s,\,a) & :=\left[H^{i}(v_{\ast}^{i},\,y_{\ast}^{i},\,z_{\ast}^{i})\right](s,\,a)-\left[H^{i}(v_{n-1}^{i},\,y_{n,\ast}^{i},\,z_{n,\ast}^{i})\right](s,\,a),\label{err 2}
\end{align}
for $t\leq T$ and $n\leq N$. The process $\epsilon_{n,t}^{i}$ represents
the risk estimation error (e.g. the duality gap in the corresponding
stochastic saddle-point problem) and the process $\xi_{n,t}^{i}$
represents the $Q$-value approximation error of $Q_{n,T}^{i}$ with
respect to $Q_{\ast}^{i}$. In this new notation, we may write Step
3 in Algorithm 1 as
\begin{align}
& Q_{n,t}^{i}(s,\,a)-Q_{n-1,T}^{i}(s,\,a)\nonumber \\
= & -\text{\ensuremath{\theta}}_{k}^{n}[Q_{n-1,T}^{i}(s,\,a)-Q_{\ast}^{i}(s,\,a)+\xi_{n,t}^{i}(s,\,a)+\epsilon_{n,t}^{i}(s,\,a)+Q_{\ast}^{i}(s,\,a)-H^{i}(v_{\ast}^{i},\,y_{\ast}^{i},\,z_{\ast}^{i})(s,\,a)],\label{Update}
\end{align}
for all $(s,\,a)\in\mathcal{K}$, $t\leq T$, and $n\leq N$. Based
on Eq. (\ref{stable}), we see that for all $(s,\,a)\in\mathcal{K}$, 
\[
\mathbb{E}\left[Q_{\ast}^{i}(s,\,a)-\left[H^{i}(v_{\ast}^{i},\,y_{\ast}^{i},\,z_{\ast}^{i})\right](s,\,a)\,\vert\,\mathcal{G}_{t+1}^{n-1}\right]=0. 
\]
By Lemma \ref{Lemma 5.5-1}, we know that
\begin{equation}
\|\epsilon_{n,t}^{i}\|_{2}^{2}\leq\frac{\gamma^{2}C\,K_{G}^{2}}{\kappa(1-C\,K_{\psi}^{(1)}\kappa)}\|Q_{n-1,T}^{i}-Q_{\ast}^{i}\|_{2}^{2},\label{Boundedness}
\end{equation}
by setting $t=1$ in inequality (\ref{Risk bound}). In particular, inequality (\ref{Boundedness})
shows that the conditional expectation w.r.t. $\mathcal{G}_{t+1}^{n-1}$
of the risk estimation error is bounded by $\|Q_{n-1,T}^{i}-Q_{\ast}^{i}\|_{2}^{2}$.
In addition, by Lipschitz continuity of $G^{i}$, we have
\begin{align}
\|\xi_{n,t}^{i}\|_{2}^{2} \leq\, & \gamma^{2}\,K_{G}^{2}|v_{n-1,T}^{i}(s)-v_{\ast}^{i}(s)|+\gamma^{2}\,K_{G}^{2}|(y_{\ast}^{i}(s,\,a),\,z_{\ast}^{i}(s,\,a))-(y_{n,\ast}^{i}(s,\,a),\,z_{n,\ast}^{i}(s,\,a))|\nonumber \\
\leq\, & \gamma^{2}\,K_{G}^{2}[\|Q_{n-1,T}^{i}-Q_{\ast}^{i}\|_{2}+\|(y_{\ast}^{i},\,z_{\ast}^{i})-(y_{n,\ast}^{i},\,z_{n,\ast}^{i})\|_{2}]\nonumber \\
\leq\, & \gamma^{2}\,K_{G}^{2}(1+K_{G}K_{S})\|Q_{n-1,T}^{i}-Q_{\ast}^{i}\|_{2}.\label{Induction}
\end{align}
An \emph{iterative stochastic algorithm} is of the form:
\begin{align*}
X_{t+1}(s)=&\left(1-\alpha_{t}(s)\right)X_{t}(s)+\alpha_{t}(s)\left((\mathcal{B}_{t}X_{t})(s)+w_{t}(s)\right),\,\forall s\in\mathcal{S},
\end{align*}
where $w_{t}$ is bounded zero-mean noise, $\alpha_{t}$ is the step
size, and each $\mathcal{B}_{t}$ belongs to a family $\mathcal{B}$
of pseudo-contraction mappings (see \cite{Bertsekas1996} for details).
\begin{definition}
\label{Definition 6.5}\cite[Definition 7]{Even-Dar2004} An iterative
stochastic algorithm is well-behaved if:
\begin{enumerate}
\item The step sizes $\{\alpha_{t}(s)\}$ satisfy: (i) $\sum_{t=0}^{\infty}\alpha_{t}(s)=\infty$,
(ii) $\sum_{t=0}^{\infty}\alpha_{t}^{2}(s)<\infty$, and (iii) $\alpha_{t}(s)\in(0,\,1)$
for all $s\in\mathcal{S}$.
\item There exists $B<\infty$ such that $|w_{t}(s)|\leq B$ for all $s\in\mathcal{S}$
and $t\geq0$.
\item For each $\mathcal{B}_{t}$, there exists $\gamma\in[0,\,1)$ and
$X^{\ast}$ such that $\|\mathcal{B}_{t}X-X^{\ast}\|\leq\gamma\|X-X^{\ast}\|$
for all $X$.
\end{enumerate}
\end{definition}
We define additional operators $\mathfrak{U}_{n,t}^{i}:\,\mathbb{R}^{|\mathcal{S}|\,|\mathcal{A}|}\rightarrow\mathbb{R}^{|\mathcal{S}|\,|\mathcal{A}|}$
on the $Q$-values for $i\in\mathcal{I}$, 
\begin{align}
\mathfrak{U}_{n,t}^{i}\left[Q^{i}\right]:=H^{i}(v^{i},\,y_{n,t}^{i},\,z_{n,t}^{i})-H^{i}(v_{\ast}^{i},\,y_{\ast}^{i},\,z_{\ast}^{i}),\label{U}
\end{align}
where $v^{i}$ is the value function in a Nash equilibrium of the
stage game $\left(Q^{i}(s)\right)_{i\in\mathcal{I}}$. We can then
formulate the process (\ref{Update}) as
\begin{align}
Q_{n,t}^{i}-Q_{\ast}^{i}=(1-\theta_{k}^{n})(Q_{n-1,T}^{i}-Q_{\ast}^{i})+\theta_{k}^{n}(\mathfrak{U}_{n,t}^{i}\left[Q_{n-1,T}^{i}-Q_{\ast}^{i}\right]+Q_{\ast}^{i}-H^{i}(v_{\ast}^{i},\,y_{\ast}^{i},\,z_{\ast}^{i})).\label{new update}
\end{align}
Noting that
\begin{align*}
\|\mathfrak{U}_{n,t}^{i}\left[Q_{n-1,T}^{i}\right]-\mathfrak{U}_{n,t}^{i}\left[Q_{\ast}^{i}\right]\|_{2}=\|H^{i}(v_{n-1}^{i},\,y_{n,\,t}^{i},\,z_{n,\,t}^{i})-H^{i}(v_{\ast}^{i},\,y_{n,\,t}^{i},\,z_{n,\,t}^{i})\|_{2},
\end{align*}
where $H^{i}$ is defined in (\ref{Operator H}). By leveraging the
non-expansiveness of the stage game equilibrium mapping in Lemma \ref{Lemma 6.1} and Lipschitz continuity
of $G$, it follows that the operators $\mathfrak{U}_{n,t}^{i}$ is
a pseudo-contraction. 
\begin{theorem}
\label{Theorem 2.6} For all $Q_{n-1,\,T}^{i}$,
we have $\|\mathfrak{U}_{n,t}^{i}\left[Q_{n-1,T}^{i}\right]-\mathfrak{U}_{n,t}^{i}\left[Q_{\ast}^{i}\right]\|_{2}\leq\gamma K_{G}\|Q_{n-1,T}^{i}-Q_{\ast}^{i}\|_{2}.$
\end{theorem}

In addition, from Assumption \ref{Assu:Q_step-size}, we know that
the update rule (\ref{new update}) satisfies Condition 1 of Definition
\ref{Definition 6.5}. Furthermore, based on Eq. (\ref{stable}),
we know that update rule (\ref{new update}) satisfies Condition 2
of Definition \ref{Definition 6.5}. For the following Lemma \ref{Lemma 2.7},
we bound the $l_{2}$-norm of $\mathfrak{U}_{n,t}^{i}\left[Q_{n-1,T}^{i}\right]$
in terms of the estimation error. Such results conform to Condition
3 in \cite[Definition 7]{Even-Dar2004} or Condition 3 in Definition
\ref{Definition 6.5}.
\begin{lemma}
\label{Lemma 2.7} Let
\[\gamma<\frac{1}{K_{G}}\min\left\{ \sqrt{\frac{\kappa(1-C\,K_{\psi}^{(1)}\kappa)}{C+(1+K_{G}K_{S})\kappa(1-C\,K_{\psi}^{(1)}\kappa)}},\,1\right\}.
\]
There exists $\gamma^{\prime}\in[0,\,1)$ such that
\[
\gamma^{\prime}=\gamma\,\sqrt{\frac{C\,K_{G}^{2}}{\kappa(1-C\,K_{\psi}^{(1)}\kappa)}+K_{G}^{2}(1+K_{G}K_{S})},
\]
and
\[
\|\mathfrak{U}_{n,t}^{i}\left[Q_{n-1,T}^{i}\right]\|_{2}\leq\gamma^{\prime}\|Q_{n-1,T}^{i}-Q_{\ast}^{i}\|_{2}.
\]
\end{lemma}

\begin{proof}
We have
\begin{align*}
	&\|\mathfrak{U}_{n,t}^{i}\left[Q_{n-1,T}^{i}\right]\|_{2}^{2}
	\\
	=\, & \left\Vert H^{i}(v^{i},\,y_{n,\,t}^{i},\,z_{n,\,t}^{i})-H^{i}(v_{\ast}^{i},\,y_{\ast}^{i},\,z_{\ast}^{i})\right\Vert _{2}^{2}\\
	=\, & \|\epsilon_{n,t}^{i}+\xi_{n,t}^{i}\|_{2}^{2}\\
	\leq\, & \|\epsilon_{n,t}^{i}\|_{2}^{2}+\|\xi_{n,t}^{i}\|_{2}^{2}\\
	\leq\, & \gamma^{2}\left(\frac{C\,K_{G}^{2}}{\kappa(1-C\,K_{\psi}^{(1)}\kappa)}+K_{G}^{2}(1+K_{G}K_{S})\right)\|Q_{n-1,T}^{i}-Q_{\ast}^{i}\|_{2}^{2},
\end{align*}
where the first equality holds by definition of $\mathfrak{U}_{n,t}^{i}$
in Eq. (\ref{U}), the second equality holds by Eq. (\ref{err 1})
and Eq. (\ref{err 2}), and the last inequality follows from inequalities
(\ref{Boundedness}) and (\ref{Induction}). Given the relationship
between $\gamma$ and $\gamma^{\prime}$, we get the desired result.
\end{proof}
Given the update rule (\ref{Update}), Theorem \ref{Theorem 2.6},
Lemma \ref{Lemma 2.7}, and unbiasedness of \[\mathbb{E}\left[Q_{\ast}^{i}-H^{i}(v_{\ast}^{i},\,y_{\ast}^{i},\,z_{\ast}^{i})|\mathcal{G}_{t+1}^{n-1}\right]=0.\]
We may now apply the stochastic approximation convergence theorem
\cite[Corollary 5]{Szepesvari1999} or \cite[Theorem 8]{Even-Dar2004}.
We conclude that $Q_{n,T}^{i}(s,\,a)\rightarrow Q_{\ast}^{i}(s,\,a)$
almost surely as $n\rightarrow\infty$, for all $i\in\mathcal{I}$
and $(s,\,a)\in\mathcal{K}$.

\subsection{Practical Implementation of RaNashQL}
There are several methods for computing Nash equilibria of stage
games. The Lemke-Howson algorithm for two player (bimatrix) games
is proposed in \cite{Lemke1964}. This algorithm is efficient in practice,
yet, in the worst case the number of pivot operations may be exponential
in the number of the game's pure strategies. Recently, \cite{Littman2012}
gives an algorithm for two player games that achieves polynomial-time
complexity. Polynomial-time approximation methods, such as \cite{Czumaj2017,Hemon2008,McKelvey1996},
have been proposed for general sum games with more than two players.

Implementation of RaNashQL is complicated by the fact that there
might be multiple Nash equilibria for a stage game. In RaNashQL, we
choose a unique Nash equilibrium either based on its expected loss,
or based on the order it is ranked in a list of solutions. Such an
order is determined by the action sequence, which has little to do
with the equilibrium conditions. For a two-player game, we calculate
Nash equilibria using the Lemke-Howson method (see \cite{Lemke1964}),
which can generate equilibrium in a certain order. 

We briefly discuss the storage requirement of RaNashQL.
RaNashQL needs to maintain $|\mathcal{I}|$ $Q$-values and $|\mathcal{I}|\times|\mathcal{S}|$
risk estimates (in terms of computing solutions of the corresponding
saddle-point problems). In each iteration, RaNashQL updates all $Q^{i}(s,\,a)$
for all $\left(s,\,a\right)\in\mathcal{S}\times\mathcal{A}$ and $i\in\mathcal{I}$.
Additionally, it updates $\left(y^{i}(s,\,a),\,z^{i}(s,\,a)\right)$
for all $i\in\mathcal{I}$ through SASP. The total number of entries
in each array $Q^{i}$ is $|\mathcal{S}|\times|\mathcal{A}|$. Since
RaNashQL has to maintain the $Q$-values for every player, the total
space requirement is $|\mathcal{I}|\times|\mathcal{S}|\times|\mathcal{A}|$.
The storage requirement for the risk estimation is similar. Therefore,
the storage requirement of RaNashQL in terms of space is linear in
the number of states, polynomial in the number of actions, and exponential
in the number of players. 

The algorithm's running time is dominated
by computation of Nash equilibrium for the $Q$-function updates.
In general, the complexity of equilibrium computation in matrix games
is unknown. As mentioned in the previous section, some commonly used
algorithms for two player games have exponential worst-case bounds,
and approximation methods are typically used for $n$-player games
(see \cite{McKelvey1996}).

\subsection{Experiment Settings}

We apply our techniques to the single server exponential queuing system
from \cite{Kardes2011}. In this packet switched network, packets
(blocks of data) are routed between servers over links shared with
other traffic. The service rate of each server can be set to different
levels and is controlled by a service provider (Player 1). Packets
are routed by a programmable physical device, called a router (Player
2). A router dynamically controls the flow of arriving packets into
a finite buffer at each server. The rates chosen by the service provider
and router depend on the number of packets in the system. In fact,
it is to the benefit of a service provider to increase the amount
of packets processed in the system. However, such an increase may
result in an increase in packets\textquoteright{} waiting times in
the buffer (called latency), and routers are used to reduce packets\textquoteright{}
waiting times. Thus, the game theoretic nature of the problem arises
because the service provider and router the have such competing objectives. 

The state space is $\mathcal{S}=\{0,\,1,\,...,\,S\}$, where $S<\infty$
is the maximum number of packets allowed in the system. Only one packet
can be in service at each time, while the remaining packets wait for
service in the buffer. The router admits one packet into the system
at each time. Every time a state is visited, the service provider
and the router simultaneously choose a service rate $\mu>0$ and an
admission rate $\lambda>0$. Suppose there are $s$ packets in the
system and the players choose the action tuple $(\mu,\,\lambda)$,
then the router incurs a holding cost $h(s)$ and a cost $\theta(\mu,\,\lambda)$
associated with having packets served at rate $\mu$ when it admits
packets at rate $\lambda$. If there are no packets in the system,
$\theta(\mu,\,\lambda)$ can be interpreted as the setup cost of the
server. These payoffs are modeled as being paid to the service provider,
since the players\textquoteright{} objectives are in conflict. The
service provider, in turn, pays the router $\beta(\mu,\,\lambda)$
which represents the reward to the router for choosing the rate $\lambda$.
It can also be interpreted as the setup cost of the router. The cost
functions for strategy profile $a=(\mu,\,\lambda)$ are then:
\[
c^{1}(s,\,a):=\beta(a)-\theta(a),
\]
and 
\[
c^{2}(s,\,a):=h(s)+\theta(a)-\beta(a).
\]
We assume that the time until the admission of a new packet and the
next service completion are both exponentially distributed with means
$1/\lambda$ and $1/\mu$, respectively. We can therefore model
the number of packets in the system as a birth and death process with
state transition probabilities:
\[
	P(k\,\vert\,s,\,a):=\begin{cases}
	\mu/(\lambda+\mu), & 1<s<S,\,k=s-1,\\
	\lambda/(\lambda+\mu), & 0<s<S-1,\,k=s+1,\\
	1, & s=0,\,k=1,\\
	1, & s=S,\,k=S-1.
	\end{cases}
	\]
\\
We choose the following parameters for our example:
\begin{itemize}
\item $S=30$.
\item Each player has the same two available actions in every state:
\begin{itemize}
	\item router: first action (denoted $\overline{\lambda}$) is to admit one
	packet into the system every $10$s; second action (denoted $\underline{\lambda}$)
	is to admit one packet every $25$s.
	\item service provider: first action (denoted $\overline{\mu}$) is to serve
	one packet every $11$s; second action (denoted $\underline{\mu}$)
	is to serve a packet every $20$s.
\end{itemize}
\item Holding costs are exponential $h(s)=a\,b^{\alpha s}$ for $s\geq1$
with $a=1.2$ and $b=e$, and $\alpha=0.2$ and $h\left(0\right)=0$.
We set costs: $\theta(\overline{\mu},\,\overline{\lambda})=\theta(\overline{\mu},\,\underline{\lambda})=110$,
$\theta(\underline{\mu},\,\overline{\lambda})=\theta(\underline{\mu},\,\underline{\lambda})=90$,
$\beta(\overline{\mu},\,\overline{\lambda})=60$, $\beta(\overline{\mu},\,\underline{\lambda})=30$,
$\beta(\underline{\mu},\,\overline{\lambda})=20$, and $\beta(\underline{\mu},\,\underline{\lambda})=70$.
\end{itemize}
In this setting, the router pays the service provider more when the
service rate is higher. Also, the router receives higher reward when
both players choose higher rates or lower rates. The router receives
lower reward when the admission and service rates do not match.

We conduct three experiments, where all risk-aware players\textquoteright{}
use CVaR. The CVaR for player $i$ is
\begin{align*}
&\textrm{CVaR}{}_{\alpha^{i}}(X)
:=\min_{\eta\in\mathbb{R}}\left\{ \eta+\frac{1}{1-\alpha^{i}}\mathbb{E}\left[\max\left\{ X-\eta,\,0\right\} \right]\right\} ,
\end{align*}
where $\alpha^{i}\in[0,\,1)$ is the risk tolerance for player $i$. 

When implementing RaNashQL, we use the Lemke-Howson method to compute
the Nash equilibria of stage games, and we update the $Q$-values based on the first Nash equilibrium
generated from the method. We run our experiments in Matlab R2015a on a computer with
an Intel Core i7 2.30GHz processor, 8GM RAM, running the 64-bit Windows
8 operating system. 

\subsection{Multilinear Systems \label{sec:Multilinear-System}}

In this section, we only focus on stochasticity from state transitions
and derive a multilinear system formulation for risk-aware Markov perfect equilibria. This is
to facilitate comparison with the robust Markov perfect equilibria in\cite{Kardes2011}. All $\rho^{i}$
are taken to be CVaR. Corresponding to CVaR, we define
\begin{align}
\mathcal{M}_{s}^{i}(P_{s})\equiv\Big\{&\mu\text{ : }-e^{\top}\mu+\frac{P_{s}}{1-\alpha^{i}}\geq0, \nonumber
\\
& e^{\top}\mu=1,\,\mu\geq0\Big\} .\label{General set}
\end{align}
The set (\ref{General set}) conforms to Assumption 1(ii). Suppose Assumption 1 holds, then $x^{*}\in\mathcal{X}$
is a risk-aware Markov perfect equilibrium if and only if $(x,\,v)\in\mathcal{X}\times\mathcal{V}$
is a solution of:
\begin{equation}
x_{s}^{i}\in\text{argmin}_{u_{s}^{i}\in\mathcal{P}(\mathcal{A}^{i})}\max_{\mu\in\mathcal{M}_{s}^{i}(P_{s})}\langle\mu,\,C_{s}^{i}(v^{i})\rangle,\label{PR}
\end{equation}
for all $s\in\mathcal{S},\,i\in\mathcal{I}$. Since we only consider the risk from the stochastic state transitions,
Problem (\ref{PR}) is equivalent to the system, for all $s\in\mathcal{S},\,i\in\mathcal{I}$, 
\begin{align}
x_{s}^{i}\in\text{argmin}_{u_{s}^{i}\in\mathcal{P}(\mathcal{A}^{i})}\max_{\mu\in\mathcal{M}_{s}^{i}(P_{s})} \sum_{a\in\mathcal{A}}\left[u_{s}^{i}(a^{i})\left(\Pi_{j\ne i}x_{s}^{j}(a^{j})\right)\right]
\left[c^{i}\left(s,\,a\right)+\gamma\sum_{k\in\mathcal{S}}P(k\,|\,s,\,a)v^{i}\left(k\right)\right].\label{PR2}
\end{align}
Formulation (\ref{PR2}) can be further rewritten as
\begin{align}
x_{s}^{i}\in&\text{argmin}_{u_{s}^{i}\in\mathcal{P}\left(\mathcal{A}^{i}\right)}\Bigg\{ q_{s}^{i}\text{ : } \nonumber
\\
& q_{s}^{i}\geq\max_{\mu\in\mathcal{M}_{s}^{i}(P_{s})}\sum_{a\in\mathcal{A}}\left[u_{s}^{i}(a^{i})\left(\Pi_{j\ne i}x_{s}^{j}(a^{j})\right)\right] \nonumber
\\
& \left[c^{i}\left(s,\,a\right)+\gamma\,\sum_{k\in\mathcal{S}}P(k\,|\,s,\,a)v^{i}\left(k\right)\right]\Bigg\} .\label{PR3}
\end{align}
We introduce the following notation for our multilinear system formulation:
\begin{enumerate}
\item Let $E_{s}^{i}(x_{s}^{-i},\,C^{i})\in\mathbb{R}^{|\mathcal{A}^{i}|{}^{|\mathcal{I}|-1}\times|\mathcal{A}^{i}|}$
denote the matrix whose rows are given by the vector
\begin{equation}
\left[\prod_{j\ne i}x_{s}^{j}(a^{j})c^{i}(s,\,(a^{-i},\,a^{i}))\right]_{a^{i}\in\mathcal{A}^{i}}.\label{E}
\end{equation}
\item Define the vector $z_{s}^{i}\in\mathbb{R}^{|\mathcal{S}|\,|\mathcal{A}|}$
as
\begin{equation}
		z_{s}^{i}:=\left[\prod_{j\ne i}u_{s}^{i}(a^{i})\left(\Pi_{j\ne i}x_{s}^{j}(a^{j})\right)v^{i}(k)\right]_{a^{j}\in\mathcal{A}^{j},\,a^{i}\in\mathcal{A}^{i},\,k\in\mathcal{S}}.\label{Y}
\end{equation}
\item Let $Y_{s}^{i}(x_{s}^{-i},\,v^{i})\in\mathbb{R}^{|\mathcal{S}|\,|\mathcal{A}|\times|\mathcal{A}^{i}|}$
be the matrix such that
\begin{equation}
Y_{s}^{i}(x_{s}^{-i},\,v^{i})u_{s}^{i}=z_{s}^{i}.\label{eq:Y}
\end{equation}
\item Let $t_{s}^{i}:=[t^{i}(k\,|\,s,\,a)]_{a\in\mathcal{A},\,k\in\mathcal{S}}$
be probability distributions that satisfy
\[
\sum_{k\in\mathcal{S}}\left[t^{i}(k\,|\,s,\,a)\right]=1,\,\forall s\in\mathcal{S},\,a\in\mathcal{A}.
\]
\item Finally, let
\[
T^{i}(x):=\left[\sum_{a\in\mathcal{A}}\prod_{\begin{array}{c}
i\in\mathcal{I}\end{array}}x_{s}^{i}(a^{i})t^{i}(k\,|\,s,\,a)\right]_{s,\,k\in\mathcal{S}},
\]
and denote the $s^{th}$ row of $T^{i}(x)$ by $[t_{s}^{i}(x)]^{\top}$.
\end{enumerate}
The next theorem uses the strategy of \cite{kardecs2011discounted}
to give a multilinear system formulation for the equilibrium conditions
(\ref{PR}) (which correspond to a risk-aware Markov perfect equilibrium).
\begin{theorem}
\label{Multilinear -1} A stationary strategy $x$ is a CVaR risk-aware
Markov perfect equilibrium point with value $\left\{ v^{i}\right\} _{i\in\mathcal{I}}$
if and only if for all $i\in\mathcal{I}$ and $s\in\mathcal{S}$,
there exist $m_{s}^{i},\,n_{s}^{i}\in\mathbb{R}^{|\mathcal{A}|},\,t_{s}^{i}\in\mathbb{R}^{|\mathcal{S}|\,|\mathcal{A}|}$
such that for any $a\in\mathcal{A}$, and for $h\in\mathcal{A}$,
$(v^{i},\,x_{s},\,m_{s}^{i},\,n_{s}^{i},\,t_{s}^{i})$ satisfies
\begin{align*}
v^{i}(s)=\, & e^{\top}E_{s}^{i}(x_{s}^{-i},\,C^{i})x_{s}^{i}+\gamma\left[t_{s}^{i}(x)\right]^{\top}v^{i}\\
v^{i}(s)\leq\, & e_{h}^{\top}\left[E_{s}^{i}(x_{s}^{-i},\,C^{i})\right]^{\top}e
\\
&+\gamma e_{h}^{\top}Y_{s}^{i}(x_{s}^{-i},\,v^{i})t_{s}^{i},\\
v^{i}(s)\geq\, & [-\frac{P_{s}}{1-\alpha^{i}}]^{\top}n_{s}^{i}
\\
& +e^{\top}m_{s}^{i}+e^{\top}E_{s}^{i}(x_{s}^{-i},\,C^{i})u_{s}^{i},\\
\gamma Y_{s}^{i}(x_{s}^{-i},\,v^{i})u_{s}^{i}\leq\, & -e^{\top}n_{s}^{i},\\
-e^{\top}t_{s}^{i}\leq\, & -\frac{P_{s}}{1-\alpha^{i}},\\
e^{\top}t_{s}^{i}=\, & 1,\\
e^{\top}x_{s}^{i}=\, & 1,\\
x_{s}^{i}\geq\, & 0,\\
n_{s}^{i}\leq\, & 0,
\end{align*}
where $e_{h}$ is the $h^{th}$ unit column vector of dimension $|\mathcal{A}|$,
$e^{\top}$ is a row vector of all ones of appropriate dimension,
$E_{s}^{i}(x_{s}^{-i},\,C^{i})$ is obtained from formulation (\ref{E})
and $Y_{s}^{i}(x_{s}^{-i},\,v^{i})$ is the matrix given by Eq.
(\ref{eq:Y}).
\end{theorem}

\begin{proof}
(Proof sketch) (Step 1) First we reformulate the inner maximization
(primal problem) in Problem (\ref{PR2}) as
\[
\max_{\mu\geq0}\,\left\{ \gamma\,\left\langle \mu,\,v^{i}\right\rangle \right\} +e^{\top}E_{s}^{i}(x_{s}^{-i},\,C^{i})u_{s}^{i},
\]
then take the dual of the first maximization term 
\[
	\max_{\mu\geq0}\left\{ \gamma\,\left\langle \mu,\,v^{i}\right\rangle \text{ : }-e^{\top}\mu+\frac{P_{s}}{1-\alpha^{i}}\geq0,\,e^{\top}\mu=1,\,\mu\geq0\right\} .
	\]
The dual problem is
\begin{align*}
\min_{n_{s}^{i}\geq0,\,m_{s}^{i}}\Bigg\{ &-\left(n_{s}^{i}\right)^{\top}\left(\frac{P_{s}}{1-\alpha^{i}}\right)-m_{s}^{i}\text{ : }
\\
&\gamma Y_{s}^{i}(x_{s}^{-i},\,v^{i})u_{s}^{i}+\left(n_{s}^{i}\right)^{\top}e+m_{s}^{i}e\leq0\Bigg\} ,
\end{align*}
which can be rewritten as
\begin{align*}
\min_{n_{s}^{i}\leq0,\,m_{s}^{i}}\Bigg\{ &\left(n_{s}^{i}\right)^{\top}\left(\frac{P_{s}}{1-\alpha^{i}}\right)+m_{s}^{i}\text{ : }
\\
& \gamma Y_{s}^{i}(x_{s}^{-i},\,v^{i})u_{s}^{i}-\left(n_{s}^{i}\right)^{\top}e-m_{s}^{i}e\leq0\Bigg\} .
\end{align*}
The primal problem has a non-empty bounded feasible region (by assumption)
so it attains its optimal value. Strong duality then holds, and so the dual problem
also attains its optimal value and the optimal values are equal. 

(Step 2) We combine the two minimization objectives to rewrite Problem
(\ref{PR}), use formulation (\ref{PR3}), and then derive a single
linear programming problem as
\begin{align}
\min_{u_{s}^{i},q_{s}^{i},m_{s}^{i},n_{s}^{i}}\, & q_{s}^{i}\label{p}\\
\textrm{s.t.}\, & q_{s}^{i}\geq\left(n_{s}^{i}\right)^{\top}\left(\frac{P_{s}}{1-\alpha^{i}}\right)+m_{s}^{i} \nonumber
\\
& +\gamma Y_{s}^{i}(x_{s}^{-i},\,v^{i})u_{s}^{i},\label{p1}\\
\, & \left(n_{s}^{i}\right)^{\top}e-m_{s}^{i}e\geq\gamma Y_{s}^{i}(x_{s}^{-i},\,v^{i})u_{s}^{i},\label{p2}\\
\, & e^{T}u_{s}^{i}=1,\label{p3}\\
\, & u_{s}^{i}\geq0,\,n_{s}^{i}\leq0.\label{p4}
\end{align}

(Step 3) For the final step, we take the dual of Problem (\ref{p})-(\ref{p4}).
And the resulting dual problem is
\begin{align*}
\max_{v^{i}(s),\,t_{s}^{i}}\, & v^{i}(s)\\
\textrm{s.t.}\, & e^{\top}\,t_{s}^{i}\geq\frac{P_{s}}{1-\alpha^{i}},\\
\, & e^{\top}t_{s}^{i}=1,\\
\, & v^{i}(s)\leq\,e_{h}^{\top}\left[E_{s}^{i}(x_{s}^{-i},\,C^{i})\right]^{\top}1+\gamma e_{h}^{\top}Y_{s}^{i}(x_{s}^{-i},\,v^{i})t_{s}^{i},\quad h=1,\,...,\,A.
\end{align*}
The desired multilinear system follows by strong duality since the
primal feasible region is non-empty and bounded.
\end{proof}
\begin{remark}
Nonlinear optimization methods have been used to solve multilinear
systems that arise from equilibrium computation of one-shot games with up to four players and
four actions per player in less than five minutes (see \cite{Aghassi2006,Datta2003}).
Homotopy methods have been used to solve multilinear systems that
arise from complete information stochastic games for two players,
two actions per player, and five states in less than one minute, see
\cite{Herings2004}. We adopt the methodology proposed in \cite[Section 5.2.2]{Aghassi2006}.
We first cast the multilinear constraints into an appropriate penalty
function, and then solve the resulting unconstrained minimization
problem with an interior point algorithm. This procedure will converge
a local minimum, but not necessarily to a global minimum.
\end{remark}

\subsection{Additional Experiments}
In the section, we provide supplementary materials for Experiment I and II. 
\\
\\
\textbf{Experiment I: }We compare RaNashQL for risk-aware Markov game
with Nash $Q$-learning in \cite{hu2003nash} for risk-neutral Markov
game, in terms of their convergence rates. Given any precision $\epsilon>0$,
we record the iteration count $n$ until the convergence criterion
$\|Q_{n,\,T}^{i}-Q_{\ast}^{i}\|_{2}\leq\epsilon$ is satisfied.
Here we choose $T=10$ and we choose $N=1\times10^{5}$ for RaNashQL
and $N=1\times10^{6}$ for Nash $Q$-learning, such that both methods
have the same total number of iterations. When $\epsilon$ is extremely
small e.g., $\epsilon=0.001$, the total number of iterations for
RaNashQL and Nash $Q$-learning for the two players are respectively:
$983443$ (Nash $Q$-learning, Service provider), $936761$ (Nash
$Q$-learning, Router), $999991$ (RaNashQL, Service provider and
Router), which are relatively equal. Moreover, Figure 1 shows that
the total number of iterations for Nash $Q$-learning decrease dramatically
as the increase of precision $\epsilon$, which reveals that RaNashQL
is more computationally expensive than Nash $Q$-learning in terms
of achieving the same convergence criterion.

\begin{figure*}[h]
\begin{centering}
\includegraphics[scale=0.6]{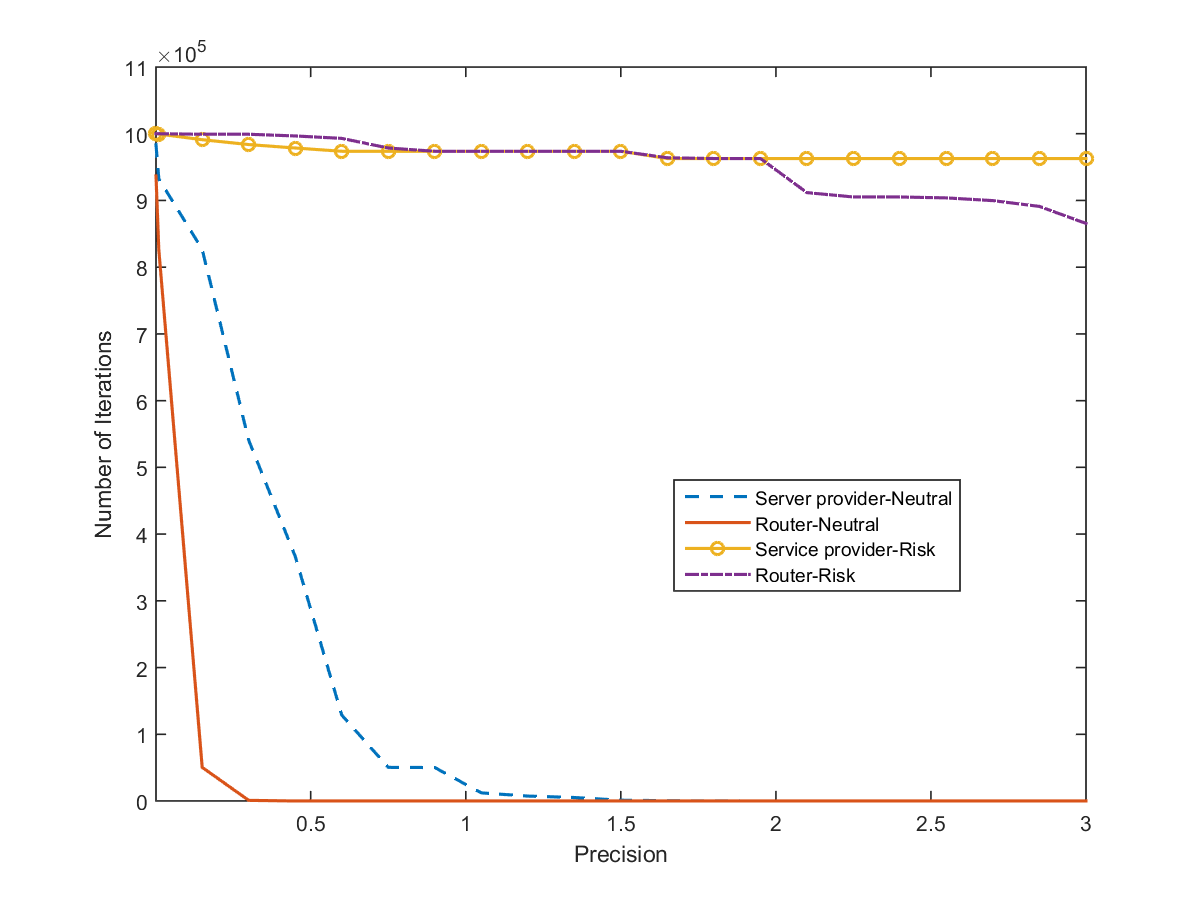}
\par\end{centering}
\centering{}\caption{Comparison between NashQL and RaNashQL}
\end{figure*}

Figure 2 presents the Markov perfect equilibrium for the risk-neutral
and risk-aware cases. It shows the equilibrium shifting when considering
the risk-awareness of players. It also shows that the both risk-neutral
and risk-aware Markov perfect equilibrium are sensitive to the perturbations
in the service rates, and risk-aware strategies for both players highly
fluctuate with the change of state (number of packet in the queuing
system). We also study how the risk tolerance level $\alpha^{i}$
(See Table 2) affects the risk-aware Markov perfect equilibrium, which
also shows the risk-aware Markov perfect equilibrium fluctuates with
the change of the risk tolerance level of CVaR. 

\begin{table}[H]
\begin{centering}
\begin{tabular}{ccc}
	\hline 
	& Service Provider ($\alpha^{1}$) & Router ($\alpha^{2}$)\tabularnewline
	\hline 
	\hline 
	Scenario 1 & 0.1 & 0.1\tabularnewline
	\hline 
	Scenario 2 & 0.2 & 0.2\tabularnewline
	\hline 
	Scenario 3 & 0.1 & 0.2\tabularnewline
	\hline 
	Scenario 4 & 0.2 & 0.1\tabularnewline
	\hline 
\end{tabular}
\par\end{centering}
\caption{Risk Tolerance Level $\alpha^{i}$}
\end{table}

\begin{figure*}
\begin{centering}
\includegraphics[scale=0.4]{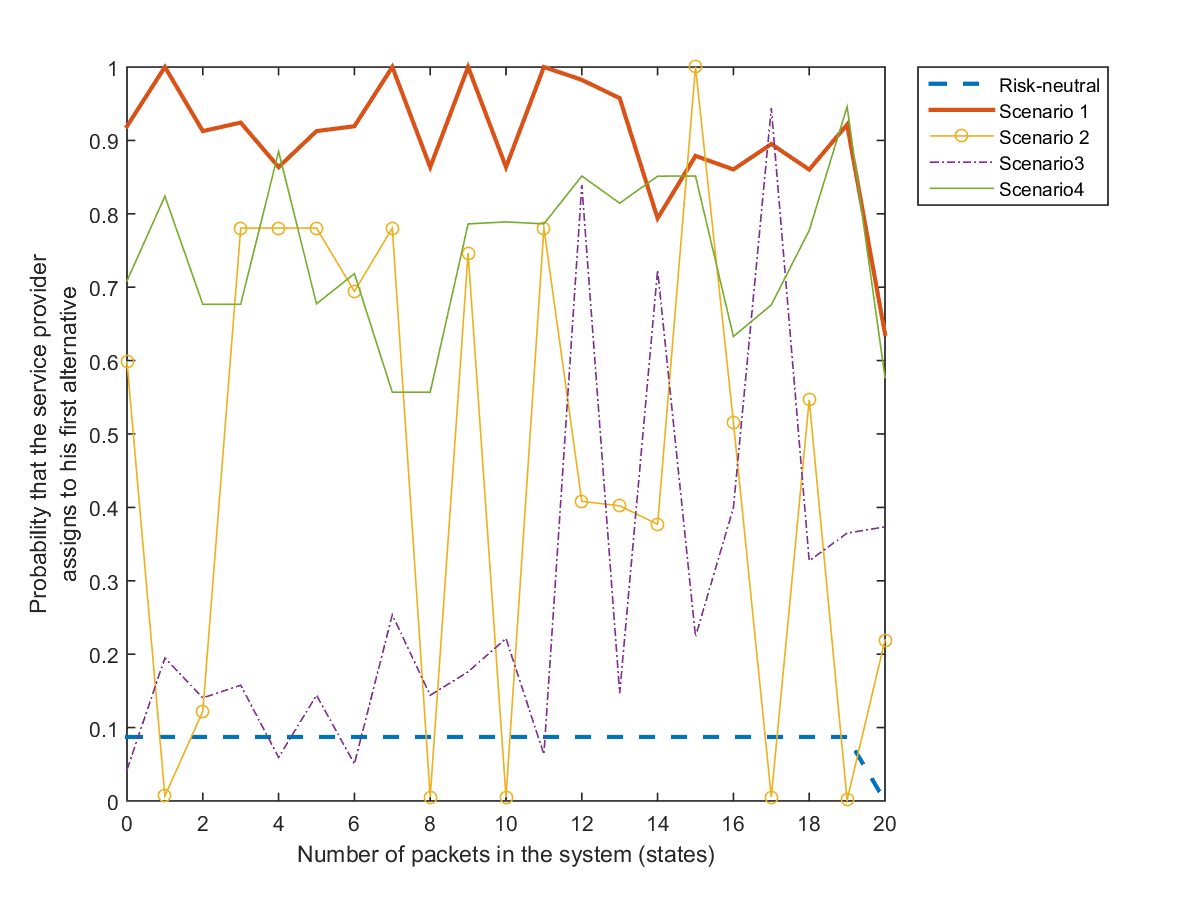}\includegraphics[scale=0.4]{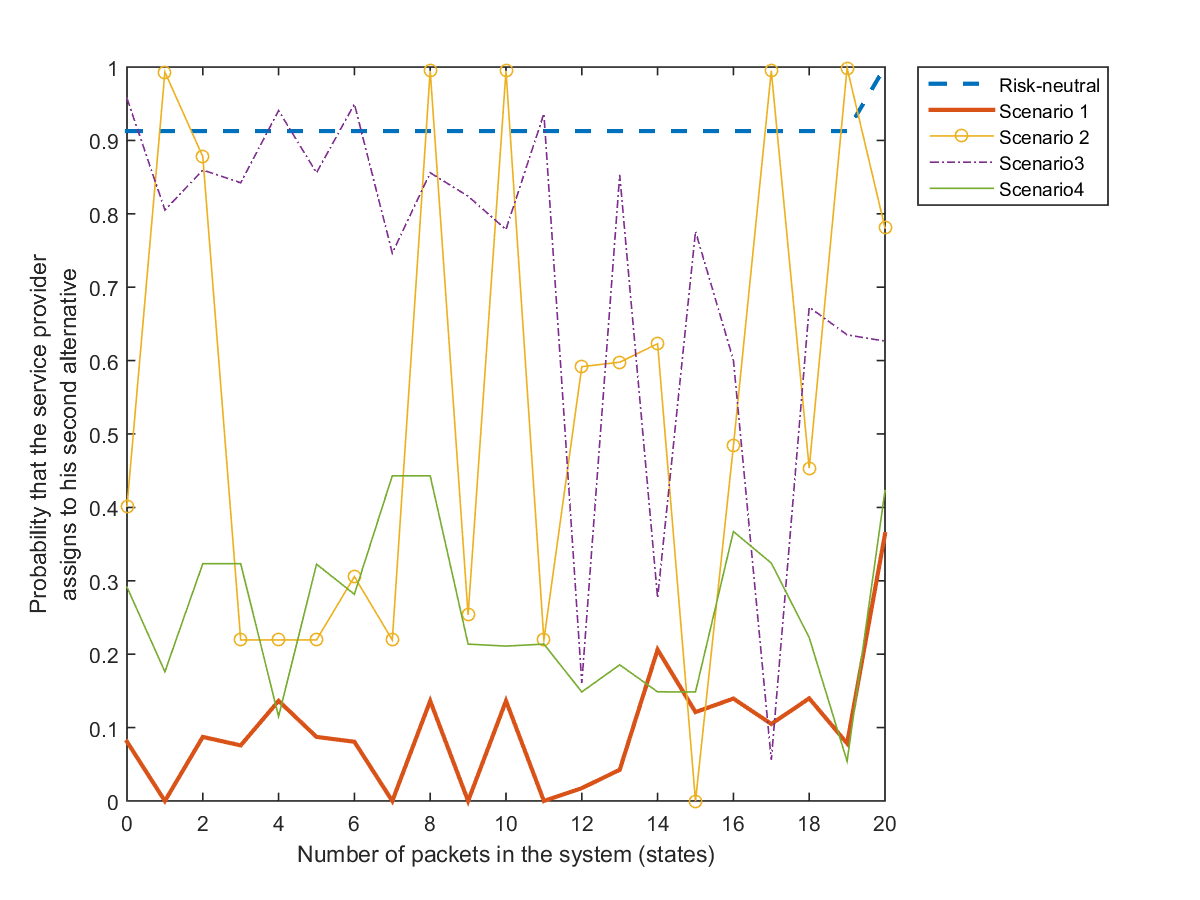}
\par\end{centering}
\begin{centering}
\includegraphics[scale=0.4]{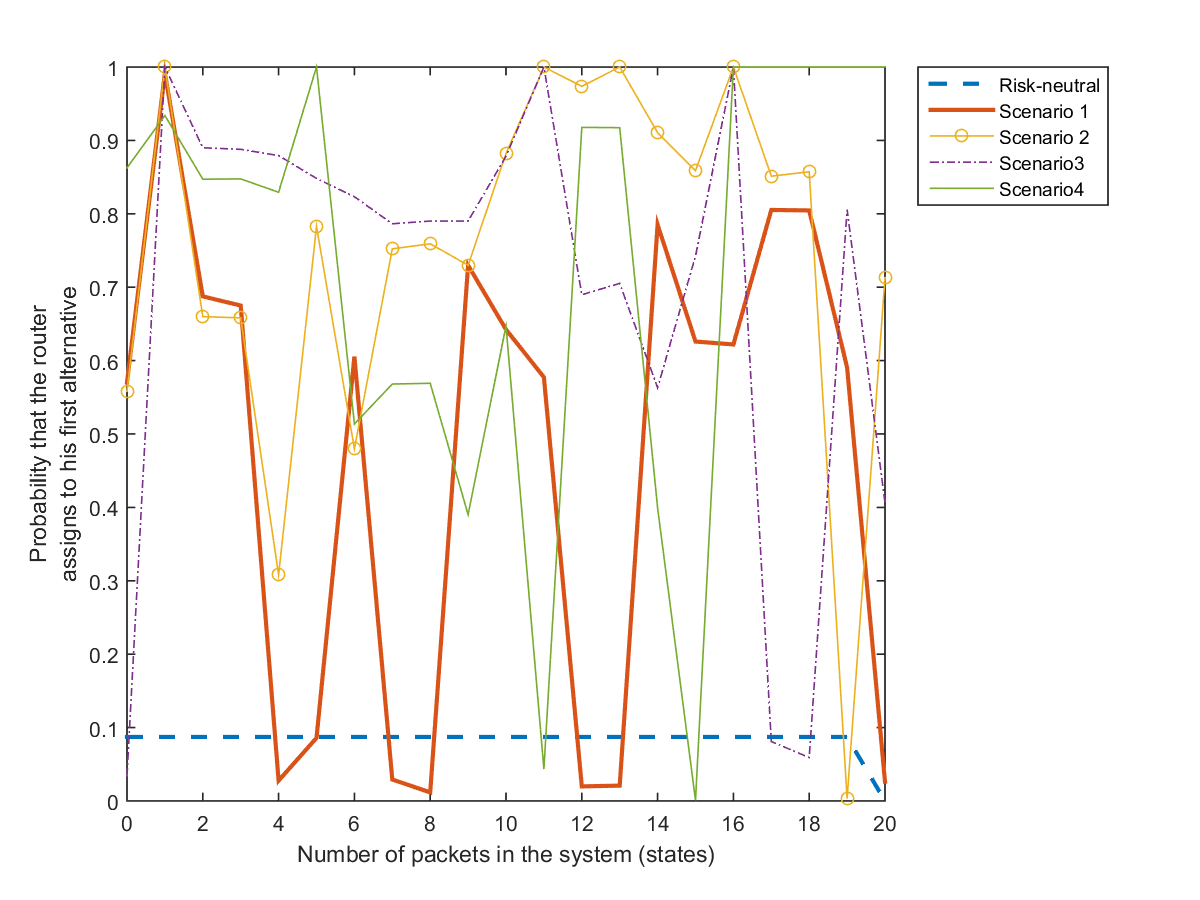}\includegraphics[scale=0.4]{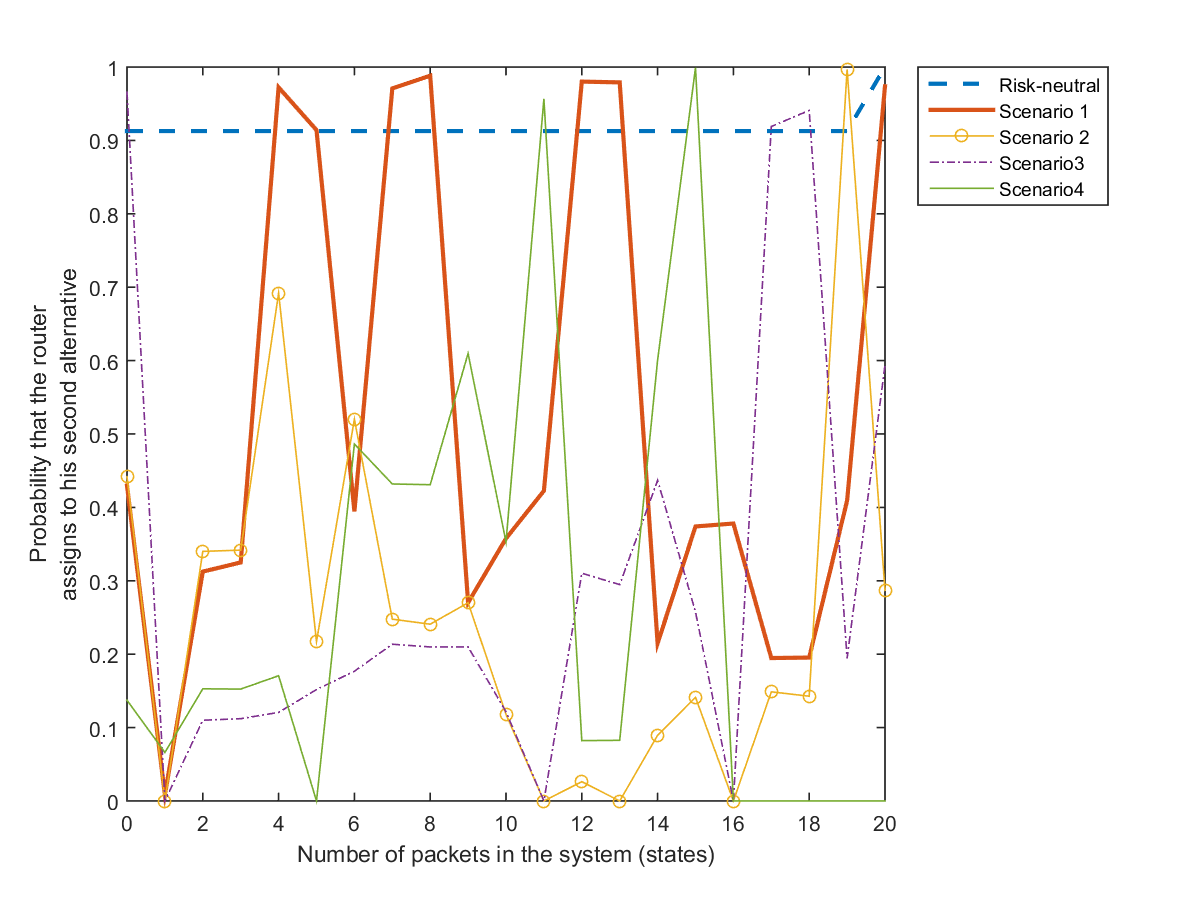}
\par\end{centering}
\caption{Comparison of Risk-Neutral and Risk-aware Markov Perfect Equilibrium}
\end{figure*}
Next, we evaluate the discounted cost under risk-neutral and risk-aware
Markov perfect equilibrium in simulation ($1000$ complete runs of the algorithm to compute the entire risk-aware Markov perfect equilibria). The risk
tolerance levels are selected as $\alpha^{1}=\alpha^{2}=0.1$, for
the risk-aware (CVaR) method in Table 3 here. Table 3 shows that considering
risk awareness will significantly increase the variance of the discounted
cost, which is contrary to expectation. The possible reason is the higher
fluctuation of risk-aware strategies with the change of state (number
of packet in the queuing system) than risk-neutral strategies.

\begin{table*}
\begin{centering}
\begin{tabular}{cccccc}
	\hline 
	Player & Method & Mean  & Variance & $5$\%-CVaR & $10$\%-CVaR\tabularnewline
	\hline 
	\hline 
	\multirow{2}{*}{Service Provider} & Risk-neutral & $-22.22$ & $1.4736e-06$ & $-22.22$ & $-22.22$\tabularnewline
	\cline{2-6} \cline{3-6} \cline{4-6} \cline{5-6} \cline{6-6} 
	& Risk-aware (CVaR) & $-77.78$ & $407.84$ & $-69.34$ & $-68.26$\tabularnewline
	\hline 
	\multirow{2}{*}{Router} & Risk-neutral & $37.48$ & $7.32$ & $37.94$ & $38.18$\tabularnewline
	\cline{2-6} \cline{3-6} \cline{4-6} \cline{5-6} \cline{6-6} 
	& Risk-aware (CVaR) & $83.68$ & $491.20$ & $86.03$ & $87.54$\tabularnewline
	\hline 
\end{tabular}
\par\end{centering}
\caption{Simulation for Risk-neutral Strategies and Risk-aware Strategies ($\alpha^{1}=\alpha^{2}=0.1$)}
\end{table*}
\textbf{Experiment II:} In this experiment, we consider a special
case where the risk only comes from the stochasticity from state transitions
(this setting is basically a risk-aware interpretation of \cite{kardecs2011discounted}
where the ambiguity is over the transition kernel). In this special
case, we can compute risk-aware Markov equilibrium using a multilinear
system as detailed in Section \ref{sec:Multilinear-System}. We evaluate
performance in terms of the relative error
\[
\frac{\sqrt{\sum_{s\in\mathcal{S}}\left(Nash^{i}(Q_{n,\,T}^{j}(s))_{j\in\mathcal{I}}-v_{\ast}^{i}(s)\right)^{2}}}{\sqrt{\sum_{s\in\mathcal{S}}v_{\ast}^{i}(s)^{2}}},\,n\leq N.
\]
In this experiment, we take the risk measure as $10\%$-CVaR. The
multilinear system is solved by an interior point algorithm within
$5\times10^{7}$ maximum function evaluation and $1\times10^{5}$
maximum iterations, and it converges to a local optimal solution in
$10471.975$ seconds. For RaNashQL, we choose $T=10$ and $N=2\times10^{6}$,
and the total implementation time for RaNashQL is $10245.314$ seconds.
The following Figure 3 validates the almost sure convergence of RaNashQL
to the service provider's strategy. For the router, the relative error
is large (around $190$\%). One possible reason is that RaNashQL converges
to different equilibria compared to the one obtained by the multilinear system solver. We
see that RaNashQL possesses superior computational performance
than interior point algorithm for this task, since
the relative error of service provider is within $25\%$ in $1\times10^{6}$
iterations, and the implementation time will be $5122.657$ seconds. 

\begin{figure*}
\begin{centering}
\includegraphics[scale=0.4]{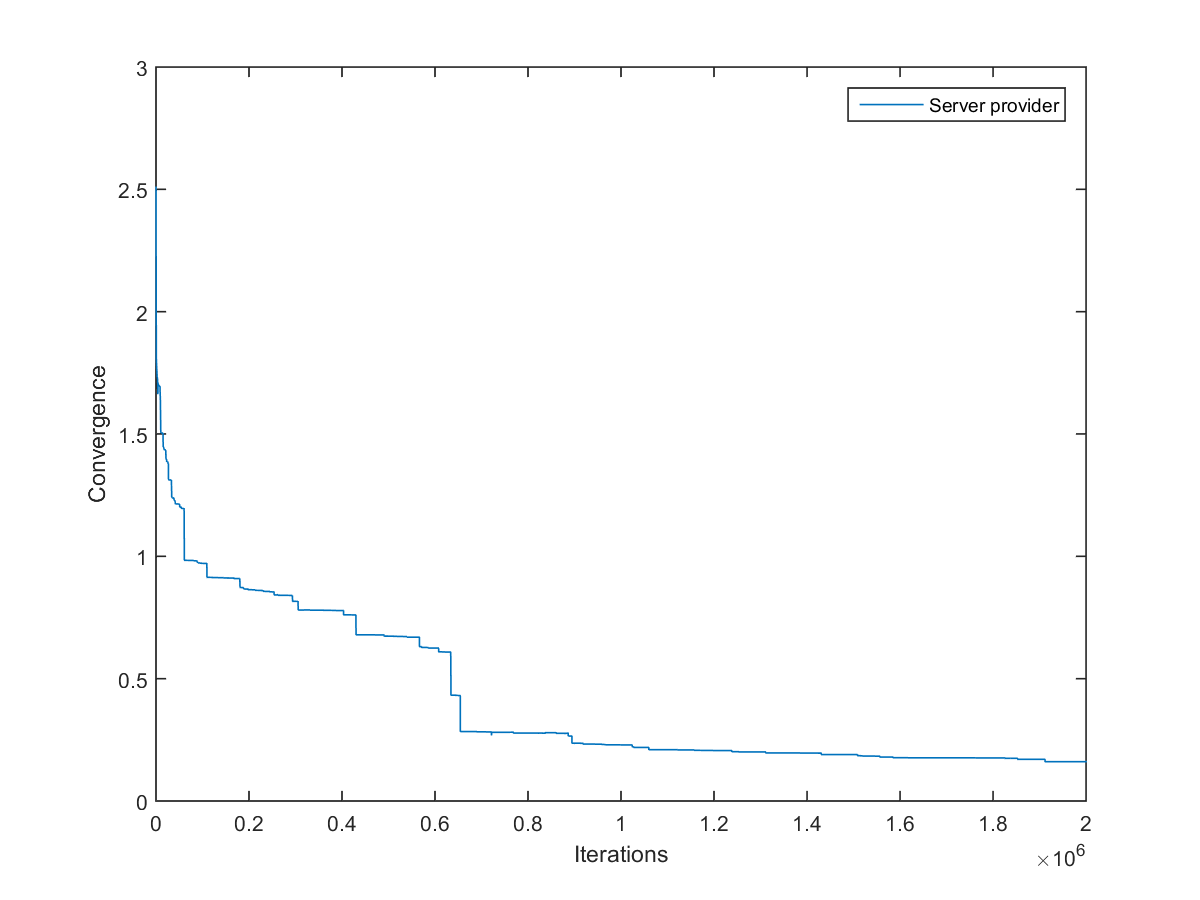}\includegraphics[scale=0.4]{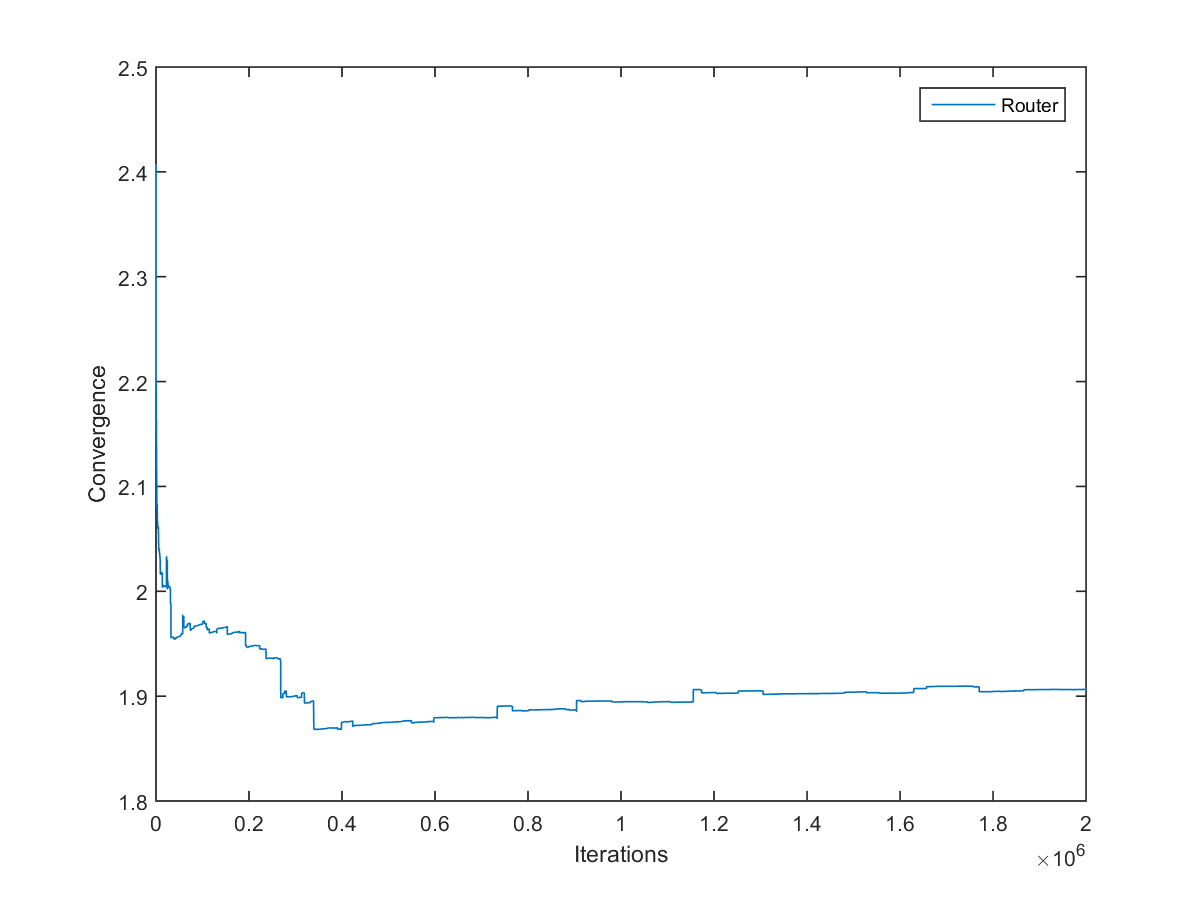}
\par\end{centering}
\caption{Almost Sure Convergence of RaNashQL}
\end{figure*}

\end{document}